\documentclass[12pt]{article}

\usepackage{fullpage}
\usepackage{changepage}

\usepackage{amsmath, amssymb}

\usepackage{graphicx}
\graphicspath{ {./img/} }
\usepackage{subcaption}
\usepackage{float}
\usepackage{enumerate}

\usepackage{natbib}
\usepackage{float}
\usepackage{algorithm}
\usepackage{algpseudocode}

\usepackage{derivative}
\usepackage{amsthm}
\usepackage{thmtools,thm-restate}

\newcommand{\blind}{0}

\newtheorem{theorem}{Theorem}
\newtheorem{lemma}{Lemma}

\newtheorem{proposition}{Proposition}
\newtheorem{definition}{Definition}
\newtheorem{method}{Algorithm}

\newcommand{\X}{X}
\newcommand{\Y}{Y}
\newcommand{\W}{W}

\renewcommand{\P}{\mbox{$\mathbb{P}$}}
\newcommand{\E}{\mbox{$\mathbb{E}$}}
\newcommand{\pb}{\mbox{$p_b$}}
\newcommand{\ps}{\mbox{$p_s$}}
\newcommand{\pbM}{\mbox{$p_{bM}$}}
\newcommand{\psM}{\mbox{$p_{sM}$}}
\newcommand{\mb}{\mbox{$m_b$}}
\newcommand{\ms}{\mbox{$m_s$}}
\newcommand{\n}{\mbox{$n$}}
\newcommand{\M}{\mbox{$M$}}
\newcommand{\p}{\mbox{$q$}}
\newcommand{\pM}{\mbox{$q_M$}}

\DeclareMathOperator{\pt}{p_{\mathrm{T}}}

\DeclareMathOperator*{\argmin}{arg\,min}

\DeclareMathOperator*{\argmax}{arg\,max}
\DeclareMathOperator*{\lin}{Lin}
\DeclareMathOperator*{\TV}{TV}

\DeclareMathOperator{\toD}{\overset{d}{\to}}
\DeclareMathOperator{\iid}{\stackrel{iid}{\sim}}
\DeclareMathOperator{\st}{\text{ s.t. }} %such that
\DeclareMathOperator{\as}{\text{ as }}
\DeclareMathOperator{\comma}{\text{ , }}

\DeclareMathOperator{\where}{\text{ where }}
\DeclareMathOperator{\textif}{\text{if }}
\DeclareMathOperator{\textand}{\text{ and }}

\DeclareMathOperator{\for}{\text{ for }}
\DeclareMathOperator{\Cr}{\mathbb{C}}
\DeclareMathOperator{\Sr}{\mathbb{S}}
\DeclareMathOperator{\Normal}{\mathcal{N}}

\let\norm\relax
\newcommand{\norm}[1]{\left\lVert #1 \right\rVert}

\newcommand\Revision[1]{{#1}}
\usepackage[pdfencoding=auto, psdextra]{hyperref}
\hypersetup{
	colorlinks,
	citecolor=blue,
	filecolor=black,
	linkcolor=black,
	urlcolor=black
}
\usepackage{cleveref}
\usepackage{etoc}
\begin{document}

\begin{center}

\begin{adjustwidth}{-2cm}{-2cm}
\centering
{\bf{\Large{Robust semiparametric signal detection in particle physics\\ with classifiers decorrelated via optimal transport}}}
\end{adjustwidth} 

\vspace*{.2in}

\begin{adjustwidth}{-2cm}{-2cm}
\centering
\begin{tabular}{ccccc}
Purvasha Chakravarti$^{\diamond,*}$ & Lucas Kania$^{\dagger,*}$  & Olaf Behnke$^{\ddagger}$ & Mikael Kuusela$^{\dagger}$ & Larry Wasserman$^{\dagger}$
\end{tabular}
\end{adjustwidth}

\vspace{.15in}

\begin{tabular}{c}
$^{\diamond}$Department of Statistical Science, University College London\\
$^\dagger$Department of Statistics and Data Science, Carnegie Mellon University\\
$^{\ddagger}$Deutsches Elektronen-Synchrotron (DESY)
\end{tabular}

% \vspace{.15in}

\begin{tabular}{cc}
\texttt{p.chakravarti@ucl.ac.uk}, 
\texttt{lucaskania@cmu.edu}, 
\end{tabular}
\begin{tabular}{ccc}
\texttt{olaf.behnke@desy.de},
\texttt{mkuusela@andrew.cmu.edu},
\texttt{larry@stat.cmu.edu}
\end{tabular}

\vspace{.15in}

\today

\end{center}

\begin{abstract}
Searches for signals of new physics in particle physics are usually done by training a supervised classifier to separate a signal model from the known Standard Model physics (also called the background model). However, even when the signal model is correct, systematic errors in the background model can influence supervised classifiers and might adversely affect the signal detection procedure. To tackle this problem, one approach is to use the (possibly misspecified) classifier only to perform a preliminary signal-enrichment step and then to carry out a signal detection test on the signal-rich sample. For this procedure to work, we need a classifier constrained to be decorrelated with one or more protected variables used for the signal-detection step. We do this by considering an optimal transport map of the classifier output that makes it independent of the protected variable(s) for the background. We then fit a semiparametric mixture model to the distribution of the protected variable after making cuts on the transformed classifier to detect the presence of a signal. We compare and contrast this decorrelation method with previous approaches, show that the decorrelation procedure is robust to moderate background misspecification, and analyze the power and validity of the signal detection test as a function of the cut on the classifier both with and without decorrelation. We conclude that decorrelation and signal enrichment help produce a stable, robust, valid, and more powerful test.
\end{abstract}

\begingroup
\renewcommand\thefootnote{*}
\footnotetext[1]{These two authors contributed equally to this paper. The names are in alphabetical order.}
\endgroup

\etocdepthtag.toc{mtchapter}
\etocsettagdepth{mtchapter}{section}
\etocsettagdepth{mtappendix}{none}
\etocsettagdepth{mtreferences}{section}
% {
% \renewcommand{\baselinestretch}{0}
% \normalsize
% \parskip=0em
% \renewcommand{\contentsname}{\normalsize Table of contents}
% \tableofcontents
% }

\pagebreak

\pagebreak
% \spacingset{1.9} % DON'T change the spacing!
%\addtolength{\textheight}{.5in}%
% Prevent main text sections from being added to the ToC
% \etocsetnexttocdepth{-1}
% \begin{bibunit}
\section{Introduction}
\label{sec:intro}
\Revision{
Particle physics seeks to understand the fundamental structure of matter by identifying its indivisible components. The Standard Model (SM) describes all known fundamental particles and their interactions. Physicists have long theorized the existence of additional fundamental particles beyond the SM and are now searching for empirical evidence of them \citep{evans2008lhc}.

At the Large Hadron Collider (LHC), two protons or heavy ions are accelerated in opposite directions and made to collide inside one of four underground detectors: ALICE, ATLAS, CMS, or LHCb. When a collision occurs, the kinetic energy of the particles transforms into mass, creating new particles \citep{Lyons08}. Each collision is called an event. Most events are unremarkable and form what physicists call background events. Occasionally, some collisions might produce theorized new fundamental particles. These collisions of interest are referred to as signal events. The problem is to detect the presence of the signal events among the background events. 

Physicists may not know the exact background or signal distributions, but they can often identify regions in the invariant-mass variable, reconstructed from the particles in the collision, where deviations from the background are likely. These are called signal regions. Figure~\ref{fig:data} illustrates this approach in an idealized scenario with both background and signal events, as well as a predefined signal region. First, physicists estimate the background by excluding the signal region (Figure~\ref{fig:fit_background}). Then, they interpolate the fitted background into the signal region (Figure~\ref{fig:extrapolate}). Finally, they test whether the observed distribution deviates significantly from the interpolated background (Figure~\ref{fig:signal_detection}).
However, in many new physics searches, the expected signal strength is very small, leading to tests with negligible power. To improve detection power, physicists train classifiers on auxiliary variables to suppress background and enhance the signal before performing the detection test.

\begin{figure}[ht]
\begin{subfigure}{0.25\textwidth}
    \centering
    \includegraphics[width=\linewidth]{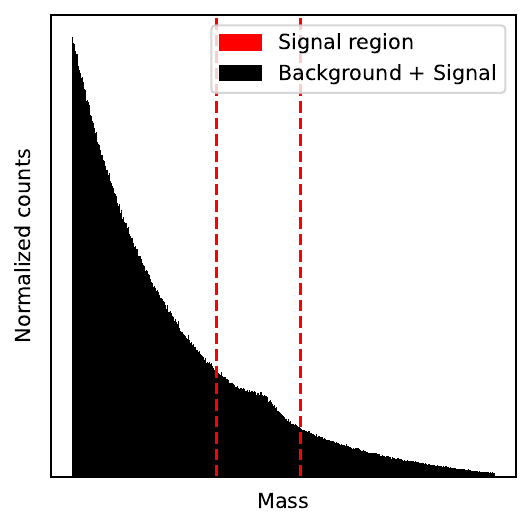}
    \caption{}
    \label{fig:data}
\end{subfigure}\hspace*{\fill}
\begin{subfigure}{0.25\textwidth}
    \centering
    \includegraphics[width=\linewidth]{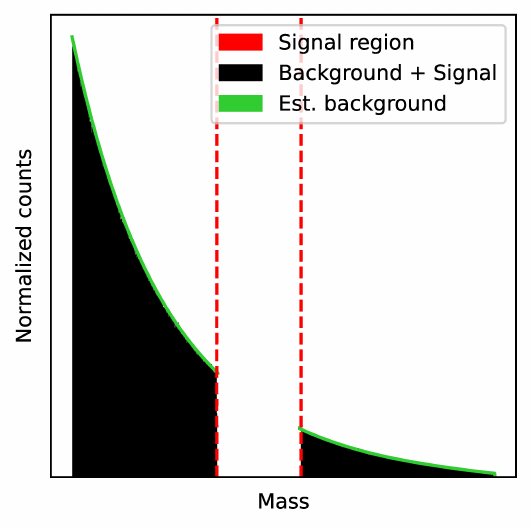}
    \caption{}
    \label{fig:fit_background}
\end{subfigure}\hspace*{\fill}
\begin{subfigure}{0.25\textwidth}
    \centering
    \includegraphics[width=\linewidth]{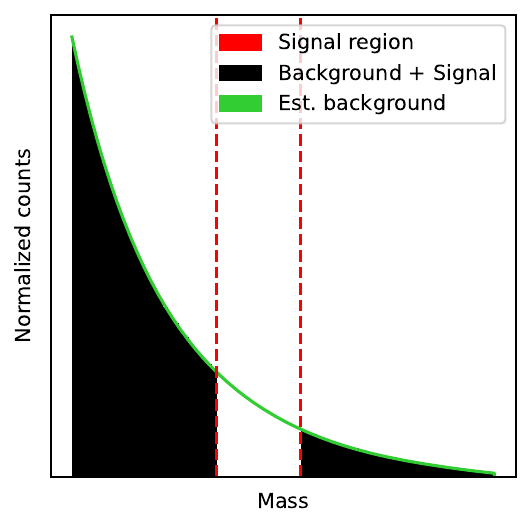}
    \caption{}
    \label{fig:extrapolate}
	\end{subfigure}\hspace*{\fill}
\begin{subfigure}{0.25\textwidth}
    \centering
    \includegraphics[width=\linewidth]{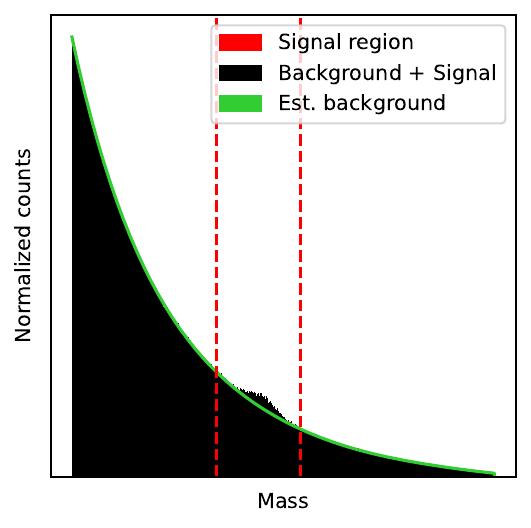}
    \caption{}
    \label{fig:signal_detection}
\end{subfigure}
\caption{Pictorial representation of signal detection in a mass spectrum.}
\end{figure}

Training boosted decision trees or neural networks on these auxiliary variables has been particularly effective at increasing the power to detect signal events \citep{atlas2017identification, atlas2017performance}. These supervised classifiers are typically trained on Monte Carlo simulations based on hypothesized background and signal models, then applied to real experimental data. However, discrepancies between the simulated training data and real test data can introduce inaccuracies. Simulating background events accurately can be difficult, since the data generation process often involves nuisance parameters and lacks a unique specification \citep{louppe2017learning}. As a result, these classifiers serve to filter and amplify the signal prior to testing, rather than to detect the signal directly.

\begin{figure}[ht]
\begin{center}
\begin{tabular}{cc}
        \text{\small Distribution of Mass} & \text{\small Distribution of Mass after Cut} \\
       \includegraphics[width=0.45\linewidth]{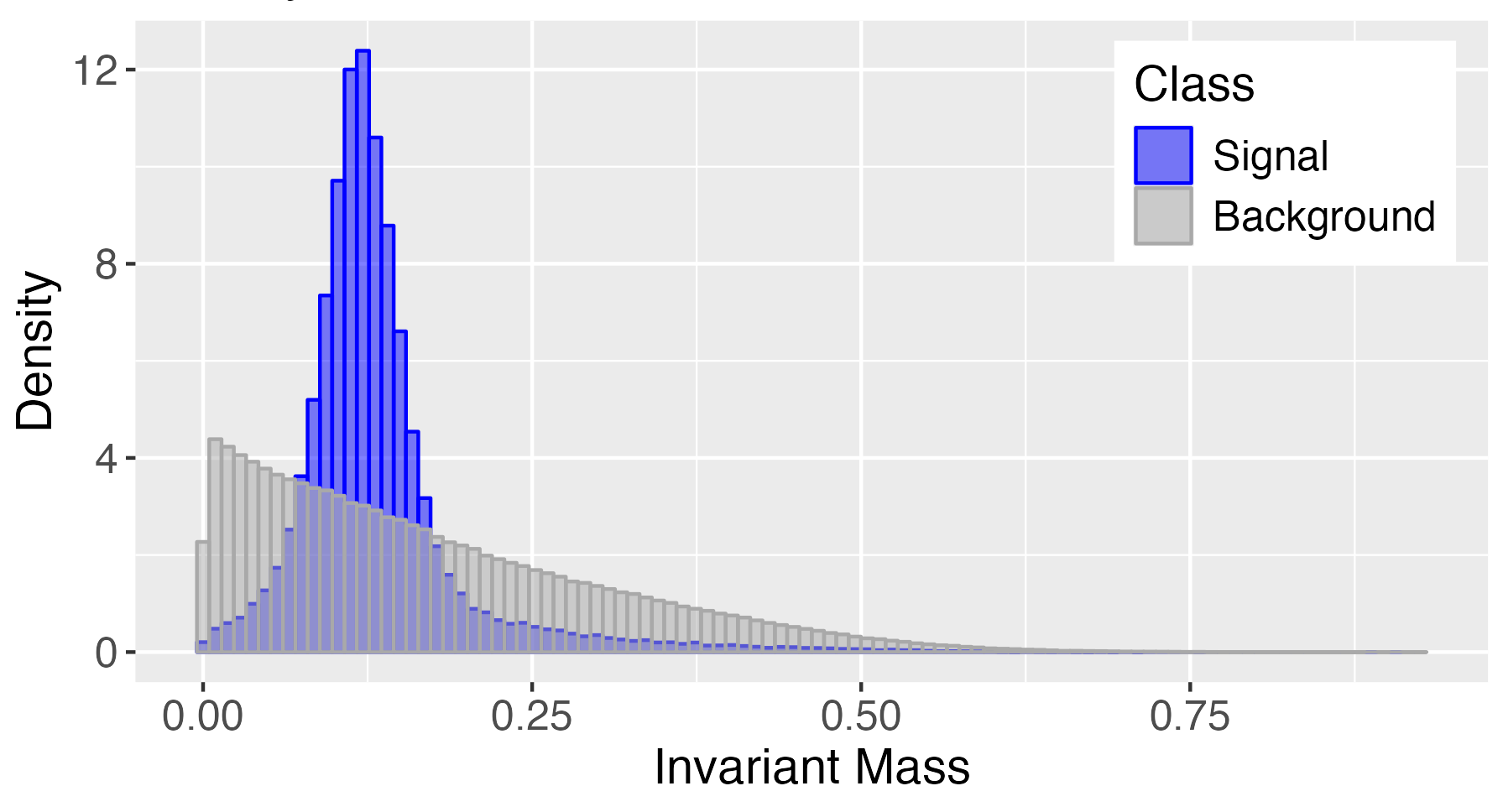} & \hspace*{-0.2in}\includegraphics[width=0.45\linewidth]{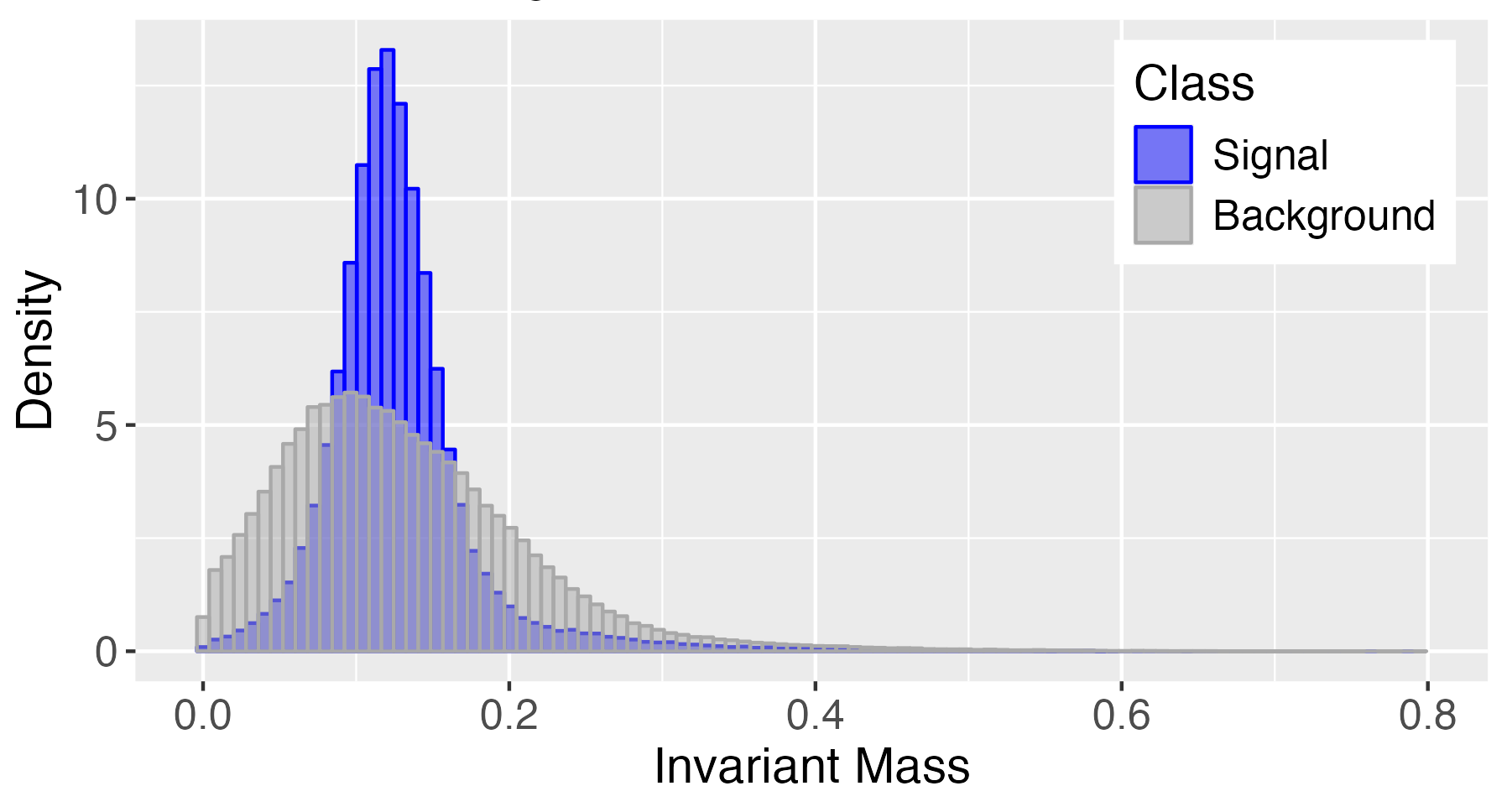}
\end{tabular}
\caption{Evidence of sculpting in the shape of the background distribution.\label{fig:sculpting}}
\end{center}
\end{figure}

A supervised classifier ($h$) trained on auxiliary variables of simulated background and signal samples is used to identify signal-rich regions where the background is suppressed. The signal detection is then performed on data within this signal-rich region, defined by predetermined cuts on the classifier output. If the classifier ($h$) is trained so that large values correspond to signal events, for a cut at $t \in [0,1]$, the signal-rich region is given by $\{x: h(x) > t\}$. Post filtering, the signal detection test is performed conditioned on the classifier in $\{x: h(x) > t\}$. This filtering can distort the background distribution of the invariant mass, making it difficult to distinguish the actual signal from background noise, a problem known as \textit{sculpting} (Figure~\ref{fig:sculpting}). 

The goal is to find a classifier that does not depend on the \emph{protected variable(s)} (invariant mass) being used for the final signal detection, ensuring that the background distribution remains unchanged, and the signal can be accurately identified. This process is known as decorrelation in high-energy physics. The proposed CDOT (Classifier Decorrelated using Optimal Transport) algorithm transforms any pre-trained classifier to be independent by composing the classifier output with an optimal transport map \citep{villani2021topics}. Further details of this algorithm are presented in Section~\ref{subsec:decorr}. 

\begin{figure}[ht]
\begin{center}
\includegraphics[width=0.9\textwidth]{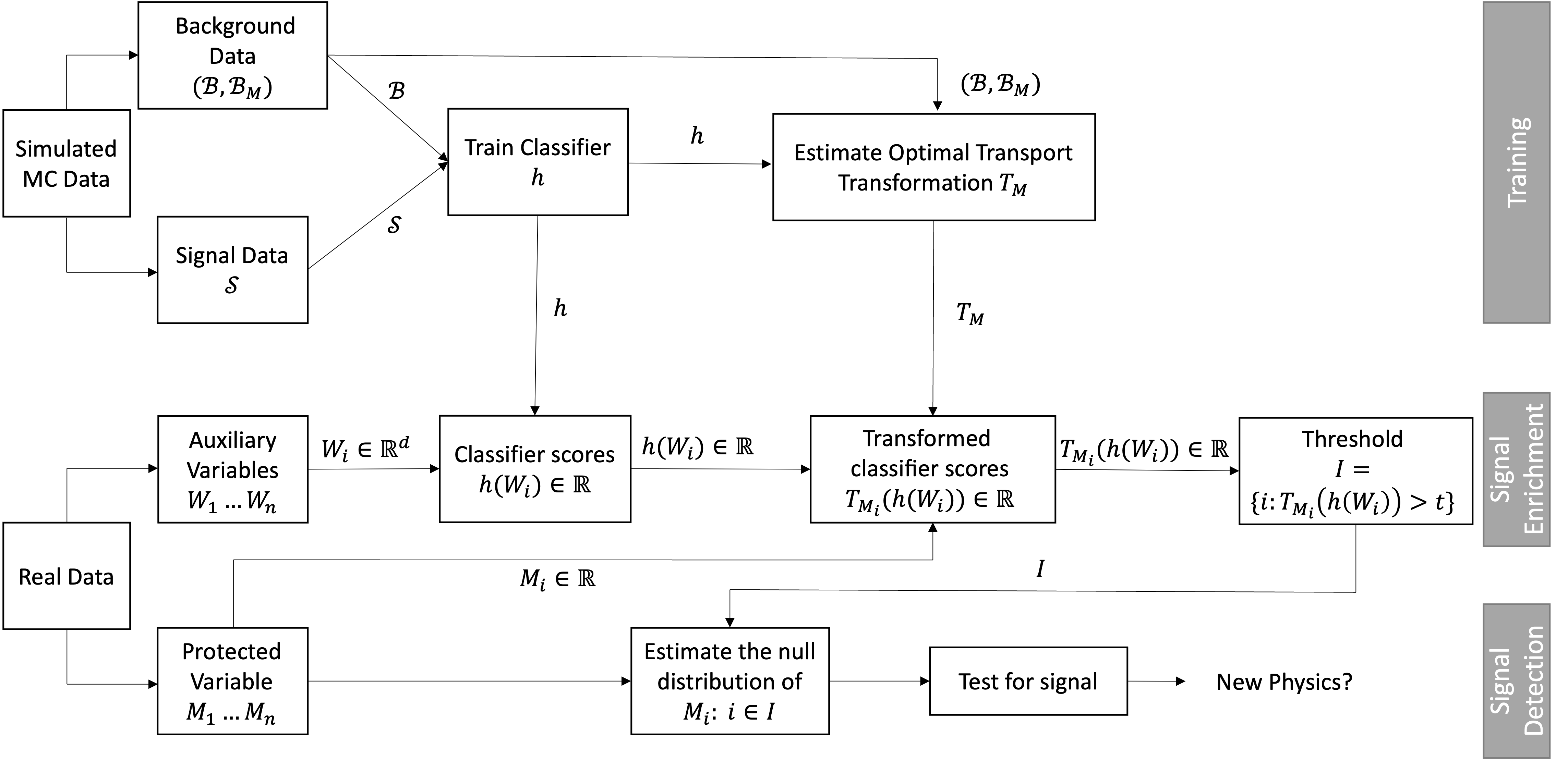}
\end{center}
\caption{Flowchart of the signal detection pipeline.}
\label{fig:flowchart}
\end{figure}

The main contribution of this paper is to present a comprehensive pipeline for signal detection, consisting of three main steps (demonstrated in Figure~\ref{fig:flowchart}):
\begin{enumerate}
    \item[(i)] \emph{Training Classifier and Decorrelation Map:} (a) \emph{Classifier:} A supervised classifier $h$ is trained on the auxiliary variables of Monte Carlo (MC) simulated background and signal data. (b) \emph{Decorrelation:} The optimal transport map ($T_M$) is determined using the proposed CDOT algorithm to obtain a decorrelated classifier, $T_M(h(\cdot))$.
    \item[(ii)]  \emph{Signal Enrichment:} The real data is filtered using the decorrelated classifier to enhance the signal.
    \item[(iii)] \emph{Signal Detection:} A semiparametric test is applied to the protected variable on the filtered data to perform signal detection. 
\end{enumerate}

The contributions of the paper can be summarized as follows:
\begin{enumerate}
\item {\bf Effect of Signal Enrichment and Decorrelation on Signal Detection Power:} We study---to the best of our knowledge, for the first time---how implementing a classifier filtration and a decorrelation algorithm before performing a signal detection test affects the power and validity of the test.  We show that signal enrichment combined with decorrelation yields a stable, robust, valid, and more powerful test. In contrast, using non-decorrelated classifiers causes sculpting, making the test unreliable.
\item {\bf Robustness to Background Distribution Misspecification:} {We perform a robustness study and show that the pipeline is robust to systematic misspecification of the background distribution. To the best of our knowledge, this is the first such robustness study in decorrelation literature.}
\item {\bf Comparison of CDOT with Other Decorrelation Algorithms:} {The proposed CDOT algorithm can be extended via geodesic morphing to create classifiers with varying degrees of decorrelation. This paper demonstrates the superiority of CDOT for high levels of decorrelation compared to previous methods.}
\item {\bf Optimality of Counting Experiments:} Assuming a parametric background and an arbitrarily complex signal distribution within the signal region, we prove that comparing the observed events in the signal region to the predicted background events yields an optimal estimator of the signal strength under the semiparametric framework \citep{bickelEfficientAdaptiveEstimation1993,vaartAsymptoticStatistics1998}.
\end{enumerate}

Lastly, one of our objectives is to also introduce the idea of decorrelation to statisticians since it has potential applications outside of particle physics.

\subsection{Related work}\label{sec:related}

Here we summarize the related literature on signal detection and decorrelation correspondingly. We focus exclusively on frequentist methods, as these are the most commonly used in high-energy physics to ensure control of error rates. Although Bayesian methods have also been explored in this field, they have yet to gain widespread adoption \citep{vandykRoleStatisticsDiscovery2014}. 

\subsubsection{Signal detection}\label{sec:related_signal_detection}

In signal detection, the aim is to determine whether the observed data originates from a background distribution or a mixture of background and signal distributions. Nonparametric, semiparametric, or parametric approaches are used depending on the available data.

{\bf Nonparametric:} When additional data from the background distribution is available, signal detection reduces to two-sample testing. That is, one can nonparametrically assess the divergence between the background and observed distributions using classifier-based tests \citep{cranmer2015approximating,kim2021,chakravartiModelindependentDetectionNew2023,gerber2023minimax}. 

{\bf Parametric:} Without additional background data, nonparametric methods cannot identify signal strength \citep{patraEstimationTwocomponentMixture2016}. A common strategy specifies distinct parametric families for background and signal and applies a likelihood ratio test \citep{chatrchyan2012observation, atlas2012observation,vandykRoleStatisticsDiscovery2014,moustafaAnomalyDetectionSystem2018}.

{\bf Semiparametric:} Semiparametric methods treat the background as known or parametric and impose structure on the signal. When the structure is too weak, one can identify only a lower bound on signal strength \citep{patraEstimationTwocomponentMixture2016}. A standard route to identifiability models the signal parametrically while keeping the background nonparametric \citep{rolkeEstimatingSignalPresence2012}. Weaker but informative conditions also work, for example, include assuming the signal is symmetric \citep{bordesSemiparametricTwocomponentMixture2010,maFlexibleEstimationSemiparametric2015} or has localized support \citep{Algeri20}. Our proposed signal-detection test assumes the latter and provides some optimality guarantees. 

Finally, we note that in this work we assume the signal region is known or can be well approximated from independent data; see \cref{sec:experiments}.  Alternatively, one may scan candidate regions and correct for multiple testing. This strategy applies to any test that uses a signal region and is standard in particle physics \citep{vandykRoleStatisticsDiscovery2014}.

\subsubsection{Decorrelation}\label{sec:related_decorrelation}

The goal of decorrelation in particle physics is to build classifiers that remain independent of the \emph{protected variable(s)} (such as invariant mass) used in signal detection tests. This independence preserves the background distribution, enabling accurate signal detection. Researchers typically use three main approaches to achieve decorrelation.

{\bf Independent Input Variables:} Modify the auxiliary variables before training the classifier to achieve independence. Early methods like Designing Decorrelated Taggers (DDT) \citep{dolen2016thinking} and Convolved SubStructure (CSS) \citep{moult2018convolved} used theoretical domain knowledge for this purpose, but are limited to specific experiments.

{\bf Decorrelation During Training:} Add a penalty term to the classification loss function to enforce independence during training. Examples include DisCo \citep{kasieczka2020disco}, which uses distance correlation, Adversarial Neural Networks (ANN) \citep{shimmin2017decorrelated, louppe2017learning}, which minimize dependence through adversarial training, and Moment Decomposition (MoDe) \citep{kitouni2021enhancing}, which penalizes dependence based on moments of the protected variable's CDF. Other methods include outlier exposure variational autoencoders (OE-VAE) \citep{cheng2023variational}.

{\bf Transforming Pre-trained Classifiers:} Apply an additional mapping to the output of pre-trained classifiers to achieve independence. The proposed CDOT (Classifier Decorrelated using Optimal Transport) algorithm uses an optimal transport map \citep{villani2021topics} for this transformation, which is found by a closed-form solution to the optimal transport problem and estimated using conditional kernel density estimators. A parallel body of work, CNOTS \citep{algren2024decorrelation}, uses the Kantorovich dual formulation and partially input convex neural networks (PICNNs) to estimate the optimal transport map. Both methods can be applied to multivariate protected variables. CDOT is designed for classifiers with univariate output, whereas CNOTS is capable of handling multivariate output classifiers. Although CDOT can be extended to work with classifiers with multivariate output, the optimal transport problem in this scenario lacks a closed-form solution (except when the variables are normally distributed). Consequently, it must be estimated numerically, which is computationally intensive. Other methods include \cite{klein2022decorrelation}, which finds the transformation using conditional normalizing flows. 
\cite{moreno2020interaction}, instead of modifying the classifier itself, adjusts the filtering threshold using quantile regression to ensure the background distribution of the filtered protected variable remains unchanged. Here, the threshold varies with the protected variable, making it a random variable. Although theoretically similar to transforming the classifier and setting a fixed threshold, the practical estimation method differs. 

Finally, we remark that the tests in our work rely on point estimates of optimal transport maps. Accounting for the uncertainty of these maps can improve inference. We leave this extension to future work, since quantifying uncertainty in transport maps \citep{del2025distributional,balakrishnan2025statistical} lies beyond the scope of this article.

\section{Datasets}
\label{sec:data}

Searches for new particles are conducted at large particle colliders like the Large Hadron Collider (LHC), where high-energy proton--proton (pp) collisions can produce heavy, unstable particles (called resonances). These particles decay into more stable particles, whose paths and energies are measured in detectors. From these measurements, physicists reconstruct the original particles. The challenge is to distinguish these rare signal events from the background. In our experiments, we consider two signal searches: the search for high-momentum W-boson production and high-mass exotic resonance production.

 {\bf High-Momentum W-Boson Production:} Some theorized new massive particles are predicted to decay into a final state containing a W-boson and other heavy particles  \citep{atlas2018performance}. The W-boson, being unstable, can decay into two quarks, which form jets. When these jets merge due to the high transverse momentum ($\pt$) of the W-boson, they create a large-radius ``W-jet." The challenge is to distinguish these W-jets (signal events) from other jets produced by pp collisions (background events). The invariant mass of the W-jet is a key discriminator, as the signal forms a peak while the background is smooth  \citep{atlas2018performance}. Although the W-boson itself is a well-known SM particle, detecting highly boosted W-bosons can be a strong indirect hint for new particles with much higher mass. Jet sub-structure variables (details in Table 1 of \cite{atlas2018performance})  are used as auxiliary variables for training classifiers in this experiment.

This experiment is referred to as W-tagging. W-tagging is one of the benchmark experiments used to compare decorrelation methods in particle physics. In \citet{atlas2018performance}, the ATLAS Collaboration performed a detailed study of some existing decorrelation methods for W-tagging. Here, we use the more recent W-tagging dataset from \cite{kasieczka2020disco} and \cite{kitouni2021enhancing}. The dataset includes events where gluons and light-flavor quarks form jets (background, represented by $pp \rightarrow jj$) and events producing two W-bosons decaying into jets (signal, represented by $pp \rightarrow WW$). The data is simulated using Pythia 8.219 \citep{sjostrand2008brief}, Delphes 3.4.1 \citep{de2014delphes}, and FastJet 3.0.1 \citep{cacciari2012fastjet}, with reconstructed jets having mass $m \in [50,300]$ GeV and momentum $\pt \in [300,400]$ GeV. The invariant mass of the jets is scaled to $[0, 1]$. More details on the simulation can be found in \cite{kasieczka2020disco} and \cite{atlas2018performance}.

{\bf High-Mass Exotic Resonance Production:} This involves detecting the production of an exotic high-mass particle that decays into two Higgs bosons, which then decay into four b-quarks (represented by $pp \rightarrow X \rightarrow HH \rightarrow 4b$). This happens as the Higgs boson is an unstable particle, making it impossible to directly detect it. The task is to identify events where four b-quarks originate from the resonance (signal events) versus events where the four b-quarks were produced by a different process (background events). The
b-quarks themselves hadronize to form streams of other particles, called b-jets. The kinematic properties of these b-jets help in signal identification and hence are used as auxiliary variables for training classifiers. The invariant mass of the high-mass resonance is a key discriminator and is used for signal detection (similar to W-tagging).

Here the signal distribution can be accurately approximated, but simulating the background is computationally challenging \citep{di2020higgs}. Instead, previous research assumes that there are additional events that are related to the background distribution via a distributional shift and can be used to estimate the background distribution \citep{bryant2018search, manoleBackgroundModelingDouble2022}. An example would be events that result in three b-quarks, making them impossible to have come from a signal event, but having similar kinematic properties as the background \citep{bryant2018search}. Since the background distribution is now being estimated, the simulations might not be accurate. Hence, we use this experiment to demonstrate the robustness of our entire pipeline to a misspecified background.

The dataset includes two kinds of background data:  QCD multijet events producing four b-quarks (4b background) and events producing four jets with exactly three b-jets (3b background). The data is simulated using MadGraph \citep{alwall2011madgraph}, with the invariant mass of the high-mass resonance set to $m \approx 400$ GeV and reconstructed jets having transverse momentum $\pt > 40$ GeV and pseudorapidity $|\eta| \leq 2.5$.  The classifier and CDOT training are performed on the 3b background, whereas the test is performed on the 4b background. As the kinematics of the two events are similar, but not exactly the same \citep{cms2022search}, this is a good check of robustness of the pipeline to moderate background misspecification.}

\section{Methods}
\label{sec:meth}

As discussed in the introduction in Section~\ref{sec:intro}, we have three sources of data at hand \begin{align*}
\text{MC Background }(\mathcal{B}): \quad & \X_1, \ldots, \X_{\mb} \sim \pb \\
\text{MC Signal }(\mathcal{S}): \quad & \Y_1, \ldots, \Y_{\ms} \sim \ps \\
\text{Experimental/Real data }(\mathcal{E}): \quad & W_1, \ldots, W_{\n} \sim \p = (1 - \lambda) \cdot \pb + \lambda \cdot \ps
\end{align*} Here $\pb$, $\ps$, and $\p$ are the densities of the simulated background, the simulated signal, and the experimental data, respectively. $\lambda \in [0, 1]$ is the signal strength, which represents the probability of a signal event in the experimental data. Furthermore, we have data on the protected variable $\M$, for all the three samples, i.e., $\M^\X_{i} \sim \pbM$ (Background), $\M^\Y_{i} \sim \psM$ (Signal) and $\M_{i} \sim \pM$ (Experimental, for brevity). Henceforth, we will drop the $\M$ from the density notation to interchangeably use $\pb$, $\ps$, and $\p$ for the protected variable as well, since from context, it should be clear which density we mean. We want to perform a signal detection test, i.e., test $H_0: \lambda = 0$. As discussed in Section~\ref{sec:intro}, since the signal is localized in the protected variable, the test is performed on the protected experimental data $\M_{i} \sim \p = (1 - \lambda) \pb + \lambda \ps$. Additionally, as explained in the introduction, we do not entirely trust the background simulations and so, do not use the background or signal simulations for the actual test. Hence, the test is able to robustly detect any signal that is localized in the protected variable, even if the background simulation is misspecified. 

% \begin{align*}
% p(z,m) =
% p(z)p(m|z)
% &=
% p(z)[ p(m|z,S=0)\P(S=0|z) + p(m|z,S=1)\P(S=1|z)] \\
% &=
% p(z)[ (1-\lambda(z)) p_b(m|z) + \lambda(z) p_s(m|z)],
% \end{align*}
Let $S \in \{0,1\}$ be an indicator for a signal event.
Now, let $Z= h(\W) = \hat{\P}(S = 1|\W)$ denote the output of a probabilistic classifier that separates background and signal data.
The joint density of $(Z, \M)$ for the experimental data can be written as \begin{equation}
p(z,m) =
p(z)[ (1-\lambda(z)) p_b(m|z) + \lambda(z) p_s(m|z)],\end{equation} where $\lambda(z) = \P(S = 1|Z = z)$. Now, $\lambda = \P(S=1) = \int \lambda(z) dP(z)$.
We want to test $H_0: \lambda = 0$. But as described in Section~\ref{sec:intro}, we want to perform the test after performing signal enrichment, i.e., perform the test on experimental data that passes the cut: $Z = h(\W) > t$, where $t > 0$ is some fixed threshold. So, instead of $\lambda$ we are now interested in $\lambda(z)$ for $z > t$. Under the null, $H_0: \lambda = 0$, we also have $\lambda(z) = 0 \ \forall  \ z$. Also,  we expect $\lambda(z)$ to increase as $z$ increases since the classifier is trained to detect signal, so a higher classifier output indicates a higher chance of the event being a signal event.

Considering the conditional density $p(z, m|Z > t)$, this has two effects: we are now interested in $\lambda(z)$ for $z > t$ which is larger than $\lambda$, so it increases the signal strength in the test data (good) but it replaces $\pb(m)$ and $\ps(m)$ with $\pb(m|Z > t)$ and $\ps(m|Z > t)$ which are closer together (bad) causing sculpting.  Ideally, we want the optimal classifier subject to the condition: $\pb(m|z) = \pb(m)$ and $\ps(m|z) = \ps(m)$. It is impossible to satisfy both of these conditions at once but in practice, it is sufficient to have $Z = h(\W)$ such that $\pb(m|z) = \pb(m)$, i.e., $h$ is constructed so that $h(X)$ is independent of the corresponding protected variable $\M$, given that $X$ is from the background distribution $\pb$. Note also that as the amount of enrichment ($t$) increases, while $\lambda(z)$ increases, the sample size decreases, and hence, too much enrichment could potentially reduce power. So there is a trade-off between sample size and increasing $\lambda(z)$, and one needs to be careful when choosing $t$.
So, now the joint density of $(Z, \M)$ becomes: $p(z,m) =
p(z)[ (1-\lambda(z)) \pb(m) + \lambda(z) \ps(m|z)]$  due to decorrelation. The proposed signal detection workflow has three steps (see Figure~\ref{fig:flowchart}):
\begin{itemize}
\item[(Step 1):] Training and Decorrelation: Train a probabilistic classifier ($h$) on the simulated signal and simulated background data to identify signal events. Then apply a transformation ($T_M(h)$) to decorrelate it, i.e., make it independent of the protected variable ($\M$) for the background. We present the decorrelation procedure where we find the transformation $T_M$ using optimal transport in more detail in Section \ref{subsec:decorr}.
\item[(Step 2):] Signal Enrichment: Filter the experimental data ($\W$'s) using the classifier selecting a signal-rich sub-sample $E_t = \{i: T_{M_i}(h(\W_i)) > t\}$ for some threshold $0 < t < 1$. 
\item[(Step 3):] Signal Detection:
We perform a signal detection test on the protected variable for the selected experimental data, i.e., on $\M_{i}$ for $i \in E_t$. Let $\M$ be distributed as $\p(m) = (1 - \lambda) \cdot \pb(m) + \lambda \cdot \ps(m).$
Here, $\pb(m)$ is a smooth background model, and $\ps(m)$ is a bump-like function.
We additionally have for the selected experimental data:
\begin{equation}
\label{eq:mixture}
\p(m|Z > t) = (1 - \lambda_t) \cdot \pb(m) + \lambda_t \cdot \ps(m|Z>t)
\end{equation}
where $Z = T_M(h(W))$ is the new transformed classifier and $\lambda_t = \P(S = 1|Z>t)$. We fit $\p(m|Z > t)$ using the selected experimental data, $\M_{i}$ for $i \in E_t$, and test $H_0: \lambda_t = 0$ versus $H_1:\lambda_t > 0$. We discuss the signal detection in detail in Section~\ref{sec:test}.
\end{itemize}

\subsection{Decorrelation via optimal transport}
\label{subsec:decorr}

The goal of the decorrelation procedure is to construct $h$ so that $h(\X)$ is independent of $m(\X)$, given $\X$ is from the background distribution $\pb$. 
%In this section, we discuss a few options of the ways in which this decorrelation can be done and discuss our proposed algorithm in detail.
One option is to transform the data $\X$ to make it independent of $\M$ for $\X \sim p_b$ 
%subject to making the transformed data as close to the original data as possible while keeping the marginal distribution of $\X$ intact. If $\X$ itself is independent of $\M$, then any classifier applied to it is also independent of $\M$, and hence any classifier can be used. 
However, if $\X$ is high-dimensional, this is computationally expensive, and a practical alternative is to consider transforming the classifier $h(\X)$ to be independent of $\M$ for $\X \sim \pb$. We want to do this such that the transformed classifier output is as close to the original as possible, while keeping the marginal distribution the same. Without loss of generality, one could choose any distribution for the new marginal distribution. We decided to keep the marginal distribution unchanged. 

\Revision{ One way to achieve this transformation is to solve an optimal transport problem. The Monge formulation of an optimal transport problem is to solve $\inf_T \int ||x - T(x)||^p dP(x),$ where $T$ is such that $P(\{x : T(x) \in A\}) = Q(A)$ for any set $A$ \citep{villani2021topics}. The minimizer $T$ (if it exists) is called the optimal transport map.}

The optimal transport problem that we are interested in solving is given by: minimizing $\E\left[(T_M(Z) - Z)^2\right]$ subject to $T_M(Z)$ independent of $M$ and $p_{bT_M(Z)}(z) = \pb(z)$, where $Z = h(\X)$ and $p_{bT_M(Z)}$ and $\pb$ are the marginal densities of $T_{\M}(Z)$ and $Z$, respectively. The solution to this optimization problem is given by the optimal transport map  $T_m$  from $\pb(z|m)$ to the marginal $\pb(z)$. This is shown in  Lemma~\ref{lemma:OT} in  Appendix~\ref{app::decorr}. Additionally, as our classifier $Z = h(X)$ is univariate,  $T_m$ has a closed-form solution given by $T_m(z) = F_z^{-1}(F_{z|m}(z))$ where $F_z$ is the marginal cdf of $Z = h(X)$ and $F_{z|m}$ is the conditional cdf of $Z$ given $M = m$ for $X \sim \pb$.

So, to estimate the optimal transport map $T_m$, we first need to estimate $F_z$, the marginal cdf of $Z$ and $F_{z|m}$, the conditional cdf of $Z$ given $M = m$,  where $Z = h(\X)$ is the classifier output on the background data $\X$. To estimate $F_z$, we use the empirical cdf of the background data $h(\X_i)$ to obtain an estimate $\widehat{F}_z(t)$.
Since we have sufficient background data, this is a good estimate of the marginal cdf. 

Estimating the conditional cdf $F_{z|m}$ poses a more challenging task. We use a kernel conditional distribution estimator (using \texttt{npcdist} in R package \texttt{np} \citep{li2008nonparametric, hayfield2008nonparametric, li2013optimal}) to estimate the conditional cdf. The challenge is that we need to find two optimal bandwidths for the kernel estimator, one along the $z$ dimension and one along the $m$ dimension. We note in our experiments that a fixed bandwidth works along the $z$ dimension, however, the optimal bandwidth along the $m$ dimension varies. Finding an adaptive optimal bandwidth w.r.t.~$M$ is computationally expensive. Instead, we take one of two approaches. The first approach is to assume that the optimal bandwidth along $M$ changes as a step function of $M$, where the points at which it changes are visibly observable from the distribution of $M$. The second approach is to assume that a particular location-scale transformation of $Z$ as a function of $\log(M)$ results in a constant optimal bandwidth for the kernel estimator of $F_{z|m}$. Specifically, we assume $ logit(Z) = \mu(M) + \sigma(M)\epsilon$ ,
where $\mu(M) = E[logit(Z)|\log(M)]$, $\sigma^2(M) = Var(logit(Z)|\log(M))$, and the conditional distribution estimator of $\epsilon|log(M)$ has a fixed optimal bandwidth. The intuition behind taking the $logit$ and $\log$ transformations comes from the context of $Z = h(X)$ being a classifier output and the protected variable $M$ usually being the invariant mass, which is a positively skewed variable in our experiments. Depending on the data set, the $logit$ and $\log$ transformations may or may not be necessary. 

Using this assumption, we estimate on the background data, $\widehat{\eta}_i = logit(h(X_i)) - \widehat{\mu}(M^\X_{i})$, where $\widehat{\mu}$ is an estimator of $\mu$ derived using non-parametric regression of $logit(h(X_i))$ on $\log(M^\X_i)$. Then we can estimate $\widehat{\sigma}^2(M^\X_{i})$'s using a non-parametric regression of $\widehat{\eta}^2_i$ on $\log(M^\X_i)$. Now, we can estimate $\epsilon$ using $\widehat{\epsilon}_i = \widehat{\eta}_i/\widehat{\sigma}(M^\X_{i})$ and estimate $F_{\epsilon|m}$ using a fixed-bandwidth conditional distribution estimator $\widehat{F}_{\epsilon|m}$. Then the conditional distribution of $Z|M^\X$ can be derived as $F_{z|m}(t) = F_{\epsilon|m}\left( (logit(t) - \mu(m))/{\sigma(m)}\right)$.
%\begin{align*}
%    F_{z|m}(t) &= P\left(Z \leq t | M^\X = m \right) \\
%    &= P\left(logit(Z) \leq logit(t) | M^\X = m \right)\\
%    &= P\left(\mu(m) + \sigma(m)\epsilon \leq logit(t) | M^\X = m \right)\\
%    &= P\left(\left. \epsilon \leq \frac{logit(t) - \mu(m)}{\sigma(m)} \right\vert M^\X = m \right) = F_{\epsilon|m}\left( \frac{logit(t) - \mu(m)}{\sigma(m)}\right).
%\end{align*}
Therefore, $\widehat{F}_{z|m}(t) = \widehat{F}_{\epsilon|m}\left( (logit(t) - \widehat{\mu}(m))/{\widehat{\sigma}(m)}\right).$
The CDOT (Classifier Decorrelated via Optimal Transport) procedure is detailed in Algorithm~\ref{method:CDOT} in Appendix~\ref{app::decorr}.

It is important to note that even though we call the classifier a decorrelated classifier in accordance with standard terminology in HEP, in reality the transformed classifier $T_M(h)$ is independent of $M$ for the background data and not just uncorrelated. Second, the best approach to estimate the conditional distribution of $Z|M$ (specifically the optimal bandwidth along the $m$-axis) in Step 3 of Algorithm~\ref{method:CDOT} in Appendix~\ref{app::decorr}, depends on the distribution of $M$ and $Z|M$ for the background data. We can use the available background validation data to decide which approach is better for a particular dataset. 

The optimal transport map $T_m(h)$ transports $p_b(z|m)$ to the marginal $p_b(z)$ for every fixed $m$. An additional contribution of this paper is to consider the optimal path or geodesic taken by the optimal transport map, i.e., to look at the intermediate steps of the transport problem. Let $\mathcal{F}$ be the set of all distributions and $g_m: [0, 1] \rightarrow \mathcal{F}$ is such that $g_m(0) = p_b(\cdot|m)$ and $g_m(1) = p_b(\cdot)$. Then $\{g_m(t): 0 \leq t \leq 1\}$ is a path connecting $p_b(z|m)$ and $p_b(z)$. The geodesic is the shortest path connecting $p_b(z|m)$ and $p_b(z)$ and it can be shown that for the geodesic path, the transformation corresponding to an intermediate point in the path $g_m(t)$ is given by $th + (1 - t) T_m(h)$  \citep{villani2021topics}. Hence, the geodesic gives rise to a range of decorrelated classifiers with varying degrees of decorrelation.

\subsection{Signal detection via semiparametric efficient estimation}\label{sec:test}

\Revision{We propose a test for signal detection using only the protected variable $M$. To simplify notation, we write the conditional mixture after the signal enrichment step \eqref{eq:mixture} as \begin{equation}\label{eq:model}
M_1,\dots,M_n\ \iid\ f(m) = (1-\lambda) \cdot b(m) + \lambda \cdot s(m)
\end{equation} where $f(m) = q(m \mid Z > t)$ is the mixture density, $b(m) = p_b(m)$ is the background density, $s(m) = p_s(m \mid Z > t)$ is the signal density, and $\lambda = \lambda_t = P(S = 1 \mid Z > t)$ is the signal strength. The goal is to distinguish between the null hypothesis $H_0: \lambda = 0$ (no signal) and the alternative $H_1: \lambda > 0$.

The background distribution is assumed to be either known or parametric, while the signal can be arbitrarily complex but resides in a predefined region.  We assume that the support of the mixture, denoted by $\Omega$, can be partitioned into disjoint signal and control regions, denoted by $\Sr$ and $\Cr$, such that the signal is entirely contained in $\Sr$: $\text{support}(f) = \Omega = \Sr \cup \Cr$, $\Sr \cap \Cr = \emptyset$, and $\int_{\Sr} s(m)\, dm = 1$. Under these assumptions, the model is semiparametric  \citep{bickelEfficientAdaptiveEstimation1993,vaartAsymptoticStatistics1998,kosorokIntroductionEmpiricalProcesses2008}, and the signal strength $\lambda$ is identifiable. 

We focus on tests that remain robust when the underlying signal is unknown, a scenario frequently encountered in practice. Although this robustness often comes at the expense of efficiency compared to tests tailored to known signals, it offers valuable protection in real-world settings. %However, we note that when even partial information about the signal is available, methods that exploit this structure can significantly enhance power \citep{Algeri20}.
}

In this section, we adopt the following conventions: if $b$ denotes the background density, then $B(A)$ indicates its integral over the set $A$, i.e., $B(A) = \int_A b(x)\ dx$. Analogously, let $F_n\left(A\right) = n^{-1}\sum_{i=1}^nI(M_i \in A)$ be the integral with respect to the empirical distribution. Furthermore, we let $z_{1-\alpha}$ denote the $1-\alpha$ quantile of the standard normal distribution.

A test is a function of the data $\Psi: \Omega^n\to\{0,1\}$ that returns $0$ if it considers that the data does not provide evidence against the null and $1$ otherwise. Furthermore, we say that a test $\Psi_\alpha$ is asymptotically valid at the $\alpha$-level if it asymptotically controls the probability of choosing the alternative hypothesis when the null hypothesis is true by $\alpha$, that is,
$\lim_{n\to\infty} B(\Psi_\alpha(X_1,\dots,X_n) = 1) \leq \alpha$. Finally, the power of a test is the probability that the alternative hypothesis is chosen when it is true, i.e. $F(\Psi_\alpha(X_1,\dots, X_n)=1)$ where $\lambda>0$ in \eqref{eq:model}. All tests in this section are asymptotically powerful, in the sense that for a sample size large enough, they detect the alternative hypothesis: fix $\lambda>0$, it follows that $\lim_{n\to\infty} F(\Psi_\alpha(X_1,\dots,X_n)=1) = 1$. 

To build intuition, assume the background distribution is known. Integrating the mixture density over the control region reveals that the signal strength corresponds to the ratio of the probabilities of observing an event in the control region under the mixture and background distributions: \begin{equation}
\lambda = 1 - \frac{F(\Cr)}{B(\Cr)} = 1 - \frac{1-F(\Sr)}{1-B(\Sr)}.
\end{equation} Substituting the empirical distribution $F_n$ for $F$ yields the plug-in estimator, which compares observed counts in the signal region to the expected background. In physics, this approach is commonly known as a counting experiment \citep{olafbook}. \Revision{The following lemma shows that the plug-in estimator is efficient for estimating $\lambda$. In this context, efficiency means that the estimator asymptotically achieves the minimum variance among all regular estimators. The proof and formal definition of efficiency appear in Appendix~\ref{appx:signal_detection}. \begin{restatable}{lemma}{EfficientEstimatorKnownBackground}\label{lemma:EfficientEstimatorKnownBackground}
Under model \eqref{eq:model} with known background and $\lambda \in (0,1)$, the plug-in estimator:  \begin{equation}\label{eq:efficient_estimator}
\lambda(F_n,B) = 1 - \frac{F_n(\Cr)}{B(\Cr)}
\end{equation} is efficient. Furthermore, the following test: \begin{equation}\label{eq:known_background_test}
\Psi_\alpha(F_n,B)=I\left(T(F_n,B) > z_{1-\alpha}\right) \where T(F_n,B) = \sqrt{n} \cdot\frac{F_n(\Sr)-B(\Sr)}{\sqrt{F_n(\Sr)(1-F_n(\Sr))}}
\end{equation} is an asymptotically valid test at level $\alpha$ for $\lambda \in [0,1)$ and $B(\Cr)>0$.
\end{restatable} } Physicists often do not know the exact boundaries of the signal region but expect minimal signal contribution in the control region. Consider a case where an $\epsilon$ fraction of the signal leaks into the control region, so that $\int_{\Cr} s(x)\ dx = \epsilon$. Assume $\epsilon$ is smaller than the background probability in the control region, $B(\Cr)$. Then, $\lambda(F_n,B)$ underestimates the signal strength:
\begin{equation}\label{eq:contamination}
E[\lambda(F_n,B)]=\lambda \cdot \left[1-\frac{\epsilon}{B(\Cr)}\right] .
\end{equation}
Greater contamination increases the estimator's bias and reduces the test's detection power, requiring more observations to detect the signal. However, the test remains asymptotically valid because contamination cannot occur under the null hypothesis.

Researchers do not usually know the background distribution. Thus, they might consider replacing $B$ in the test \eqref{eq:known_background_test} with a plug-in estimator of the background based on an auxiliary sample. However, this substitution introduces a non-negligible asymptotic error due to the efficient estimator being robust only to perturbations in the signal region.

A natural next step is to assume a parametric background in \eqref{eq:model}, that is,
$b = b_\gamma$ where $\gamma \in \mathbb{R}^K$. To prevent data from the signal region from biasing the background estimate, we fit the background using only observations from the control region. This leads to a censored maximum likelihood estimator: \begin{equation}
\label{eq:censored_mle}
(\gamma_*(F_n),\lambda_*(F_n)) = \argmax_{\tilde{\gamma},\tilde{\lambda}}\ \sum_{i=1}^n\ell(M_i,\tilde{\lambda},\tilde{\gamma}) \st B_\gamma(\Omega) = 1,
\end{equation} where $\ell(m,\lambda,\gamma)= I(m \in \Sr)\cdot\log\left((1-\lambda)\cdot B_\gamma(\Sr)+\lambda\right) +  I(m\in \Cr)\cdot \log\left((1-\lambda)\cdot b_\gamma(m)\right)$. \Revision{The following lemma shows that $(\gamma_*(F_n),\lambda_*(F_n))$ is the efficient estimator of $(\gamma_*(F),\lambda_*(F))$. The proof and the definition of $\tau_\lambda(F_n)$ appear in Appendix~\ref{sec:efficient_estimator_parametric_background}. \begin{restatable}{lemma}{EfficientEstimatorParametricBackground}\label{lemma:EfficientEstimatorParametricBackground}
Under the model \eqref{eq:model} with a parametric background, the censored MLE estimator \eqref{eq:censored_mle} is efficient for $\lambda \in (0,1)$ and $B_\gamma(\Cr)>0$. Furthermore, the following test: \begin{equation}\label{eq:censored_MLE_test}
\Psi^{(K)}_\alpha(F_n)=I\left(\sqrt{n}\cdot \frac{\lambda_*(F_n)}{\tau_\lambda(F_n)} > z_{1-\alpha}\right)
\end{equation} is an asymptotically valid test at level $\alpha$ for $\lambda \in [0,1)$ and $B_\gamma(\Cr)>0$.
\end{restatable}} Although the censored MLE does not provide a closed-form solution for $\gamma_*(F_n)$, looking at the first order optimality condition for the signal strength reveals that the efficient estimator is analogous to \eqref{eq:efficient_estimator} but using the estimated parametric background $
\lambda_*(F_n) = 1 - F_n(\Cr)/B_{\gamma_*(F_n)}(\Cr)
$. Since the background is estimated from the data, there is a trade-off with respect to the size of the signal region. A small signal region can lead to part of the signal leaking into the control region, biasing the estimation of the background and ultimately biasing the signal strength estimator $\lambda_*(F_n)$. Conversely, a large signal region protects against the contamination of the control region, but it increases the variance of the background estimator due to reducing the amount of data in the control region in addition to making it harder to interpolate the background in the signal region. This, in turn, increases the variance of $\lambda_*(F_n)$. In both situations, the test based on $\lambda_*(F_n)$ remains valid, but it loses power, i.e., it requires more observations to detect the signal. Finally, it is worth noting that the censored MLE is equivalent to maximizing the conditional likelihood over the control region and extending it into the signal region, see Appendix \ref{sec:conditional_MLE}. 

If the parametric form is a truncated series, $b = b_\gamma = \sum_{k=1}^K \gamma_k \cdot \phi_k(x)$ with $\gamma \in \mathbb{R}^K$, one can estimate the background parameters $\gamma$ using an expectation-maximization algorithm \citep{dempsterMaximumLikelihoodIncomplete1977}; see Appendix~\ref{sec:EM_censored_MLE}. This procedure corresponds to a version of the D'Agostini iteration \citep{agostini1,agostini2} that avoids binning. Its binned counterpart is widely used in unfolding problems \citep{adyeUnfoldingAlgorithmsTests2011}. Our numerical studies employ the $K$-th order Bernstein basis (Appendix~\ref{sec:bernstein_basis}) and apply data discretization before estimating the censored MLE to speed up computation (Appendix~\ref{sec:discretization}).

\section{Experiments}\label{sec:experiments}

We investigate the performance of our method for the detection of high-$p_T$ W bosons in Section \ref{sec:wtagging}, and of exotic high-mass resonance events in Section \ref{sec:3b_4b}. In all experiments, the data is split into training, validation, and test datasets. The training data is used to fit a random forest classifier, however, one could use any probabilistic classifier. \Revision{We use probability forests \citep{Malley2012Probability} as the classifier where the Gini index is used as the node impurity for splitting. The training data is roughly balanced in the exotic high-mass resonance experiment and skewed in favor of the signal in the W-tagging experiment.} The validation data is used to perform decorrelation, calibrate the classifier, and calibrate the signal-enriched test to achieve a type I error of $\alpha=0.05$ (further details in Appendix~\ref{sec:experiments_appx}).

Using the test dataset, we study the power of the signal-enriched test and the performance of the CDOT decorrelation algorithm. To analyze the power of the signal detection test, we proceed as follows: given a signal strength $\lambda$, we sub-sample $N=500$ datasets of $n=20000$ observations from the test dataset, such that $\lambda\%$ of those observations correspond to signal events. Then, for each dataset, we check if the test rejects the null hypothesis $H_0:\lambda=0$ and compute the empirical probability of rejecting the null hypothesis across datasets. For $\lambda=0$, that probability is the empirical type-I error, while for $\lambda \in \{0.01,0.02,0.05\}$, it is the empirical power. In all cases, we report the empirical results with their corresponding Clopper--Pearson confidence intervals. Finally, to understand the utility of the decorrelation algorithm, we study the power of the signal-enriched test both when using a non-decorrelated and decorrelated classifier to filter the observations. 

\subsection{Detection of  high-\texorpdfstring{$\pt$}{pt} W bosons}\label{sec:wtagging}

In the following, we consider the search for high-$\pt$ W bosons, which is referred to as W-tagging. Here the invariant mass is the protected variable and the jet sub-structure variables  are used for training the auxiliary classifier (details in Section~\ref{sec:data}).

We use the data in \cite{kasieczka2020disco} split into training (background: 110k, signal: 250k), validation (background: 330k, signal: 80k) and test (background: 770k, signal: 80k) datasets. The training data is used to fit a  random forest classifier with $1000$ trees and a minimal node size of $800$. The validation data is used to train the CDOT algorithm (Algorithm~\ref{method:CDOT} in Appendix~\ref{app::decorr}, using Approach 1) with splits set at $\{0.2, 0.4, 0.6\}$, i.e., a different fixed optimal bandwidth is chosen for each of the intervals $[0, 0.2], [0.2, 0.4], [0.4, 0.6]$, and $[0.6, 1]$ when training the kernel conditional distribution estimator.

\begin{figure}[ht]
\begin{center}
\includegraphics[width=\textwidth]{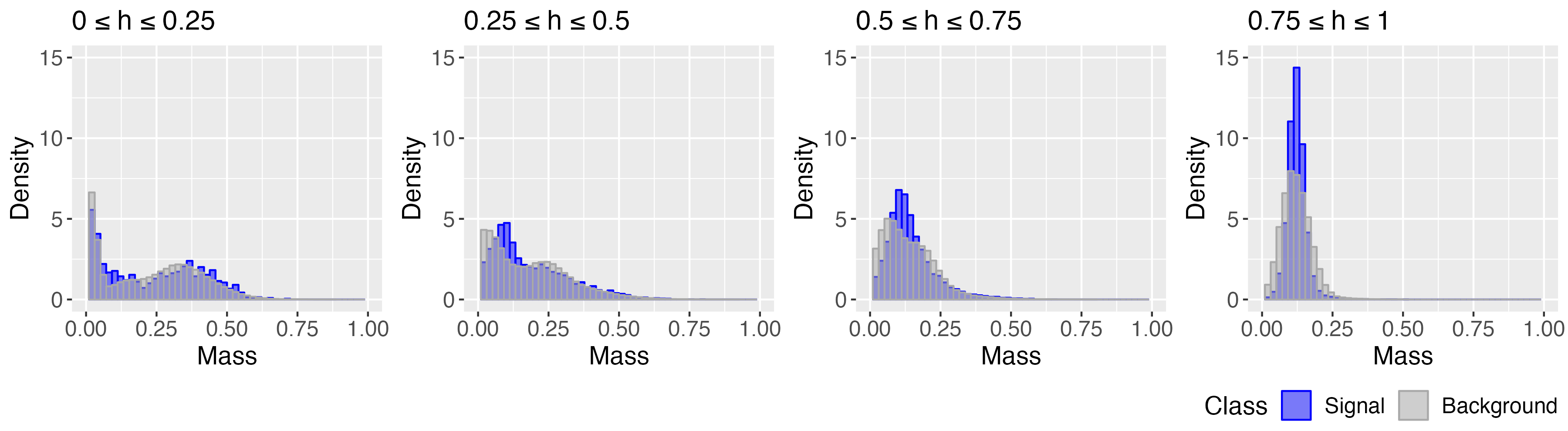}
\end{center}
\caption{Density plots of the invariant mass for the W-tagging data for different ranges of the classifier output ($h$) without any decorrelation.}
\label{img:WTaggingnoDecorrelation}
\end{figure}

\begin{figure}[ht]
\begin{center}
\includegraphics[width=\textwidth]{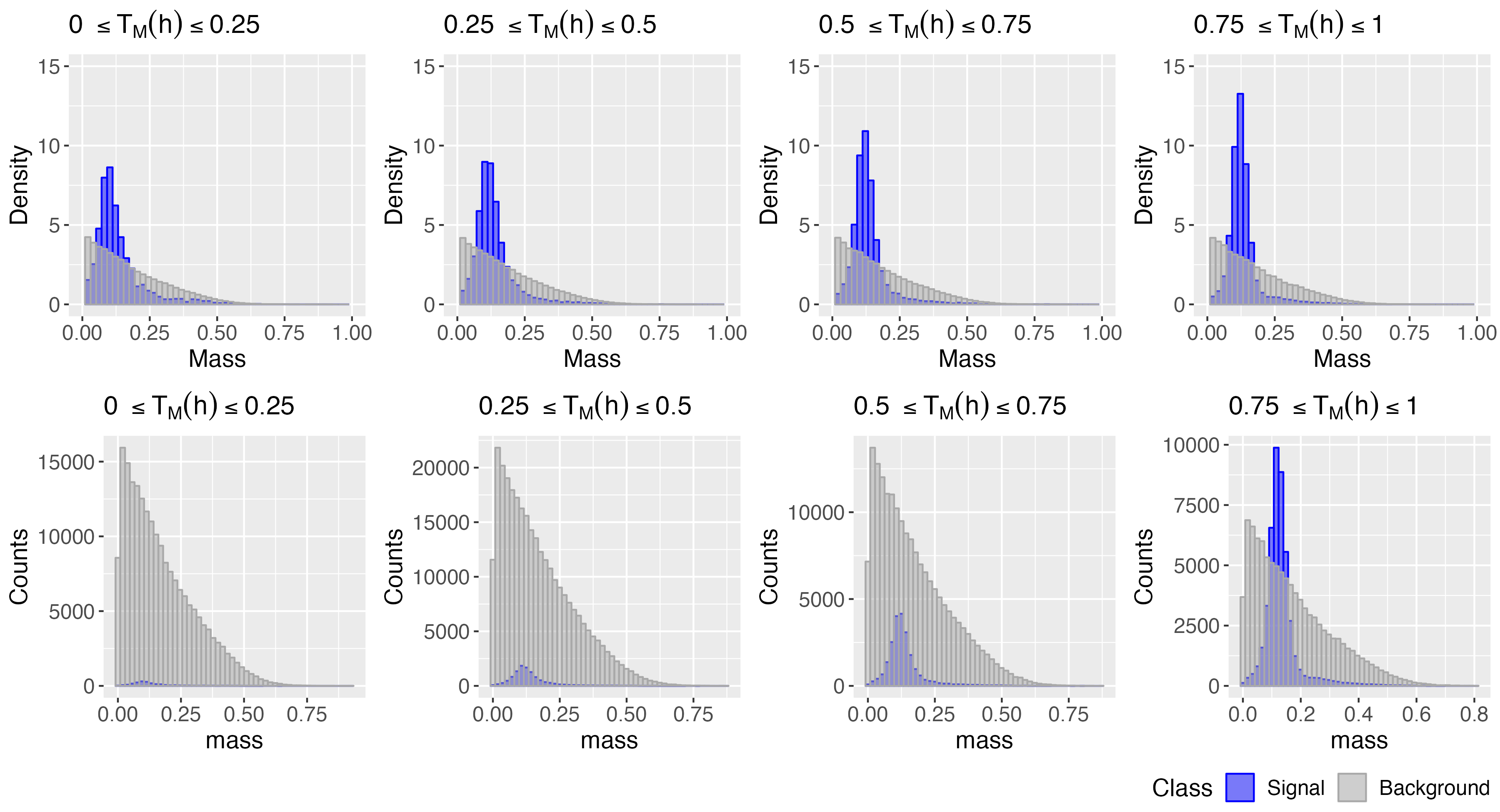}
\end{center}
\caption{Post-decorrelation density plots (top) and histograms (bottom) of the invariant mass for the W-tagging data, for different ranges of the transformed classifier ($T_M(h)$).}
\label{img:WTaggingDecorrelation}
\end{figure}

Figure~\ref{img:WTaggingnoDecorrelation} shows the effect of sculpting on the background distribution when signal enrichment is applied to the classifier output without any decorrelation. The  background distribution demonstrates a bump as the classifier output ($h$) increases making it hard to detect the signal bump using a test. Figure~\ref{img:WTaggingDecorrelation} shows the distributions of both the background and the signal after decorrelation, i.e., the signal enrichment cuts are now applied on the decorrelated classifier output ($T_M(h)$). As the classifier output increases, the top row of the figure shows that the background distribution's shape remains the same avoiding sculpting while the bottom row demonstrates the increase in the signal proportion. 

\begin{figure}[ht]
\begin{center}
\includegraphics[width=\textwidth]{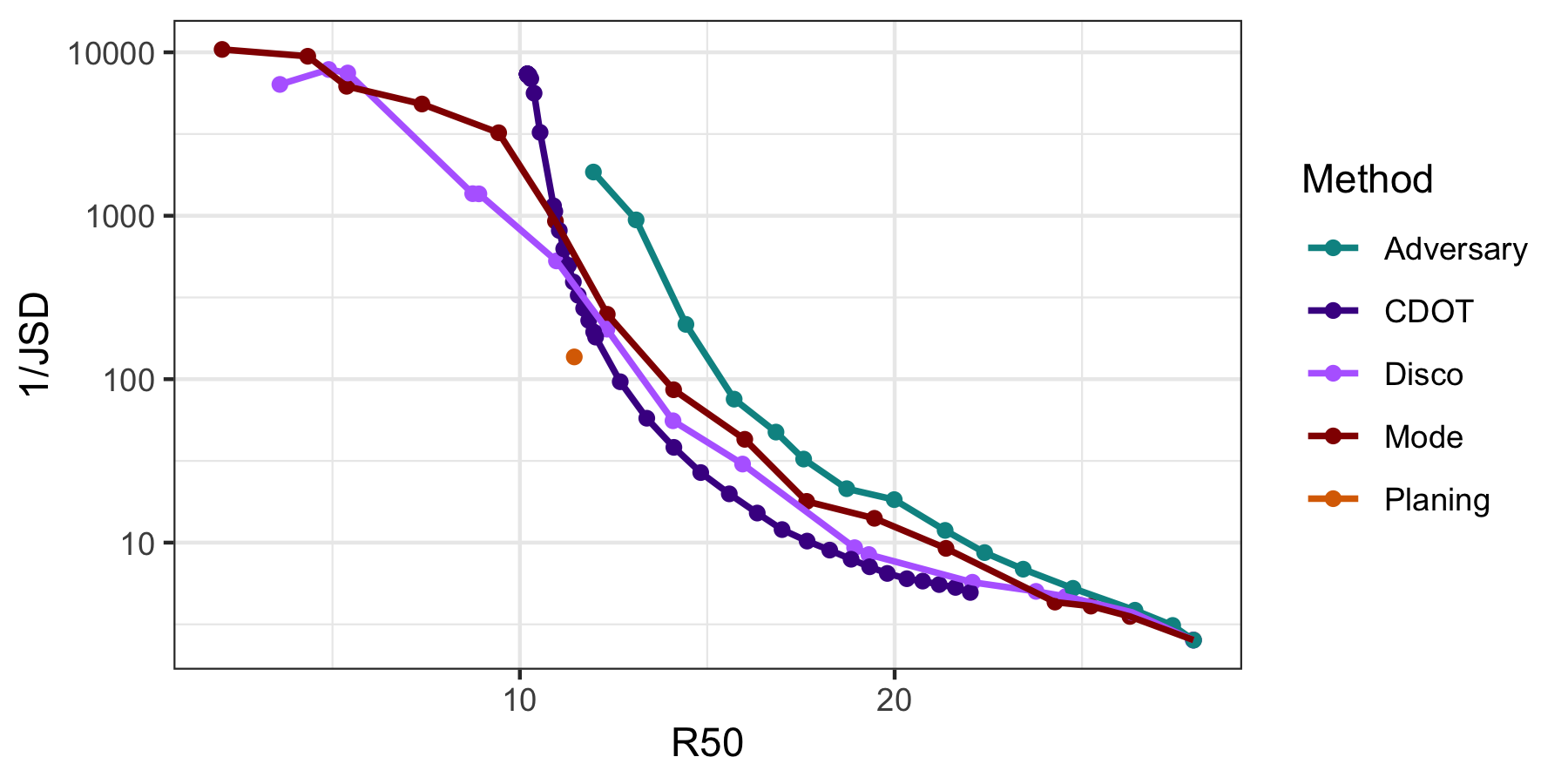}
\end{center}
\caption{Decorrelation (1/JSD on $\log_{10}$ scale) versus background-rejection power (R50) plot comparing CDOT to existing decorrelation methods. Based on \cite{kitouni2021enhancing}.%The original plot without CDOT is taken from Figure~7 in \cite{kitouni2021enhancing}
}
\label{fig:R50JSD}
\end{figure}

Figure~\ref{fig:R50JSD} compares CDOT to state-of-the-art methods introduced in Section~\ref{sec:related}, namely, Disco, ANN, MoDe, and Planing \citep{PhysRevD.97.056009}. We build on Figure~7 in \cite{kitouni2021enhancing} by adding the performance of CDOT using the same metrics to the figure. The plot compares the methods on 1/JSD versus R50, where R50 is the background rejection power (inverse false positive rate) at $50\%$ signal efficiency, and JSD is the Jensen-Shannon divergence between the background observations that are above the cut at $50\%$ signal efficiency and the background observations that are below the same cut. Higher R50 values represent good classifier performance in separating the signal from the background, and higher 1/JSD values represent higher decorrelation (independence) between the classifier and the protected variable for the background data. Figure~\ref{fig:R50JSD} shows that while Disco and MoDe are able to achieve higher decorrelation scores at the expense of lower R50 values, looking at around R50 = 10, we see that CDOT is able to achieve higher decorrelation scores for a higher amount of accuracy. The CDOT curve represents the performance of the intermediate classifiers given by the optimal transport geodesic. We see that by moving away from the fully decorrelated classifier, the drop in decorrelation is much faster than the gain in background-rejection power, which suggests that it is recommended to use the fully decorrelated classifier for the signal detection test. We also note that estimating 1/JSD is very sensitive to the selection of bins used to calculate it. We chose the same binning used in \cite{kitouni2021enhancing}, but changing the binning even slightly changes the 1/JSD scores.

\begin{figure}[ht]
\begin{center}
\includegraphics[width=\textwidth]{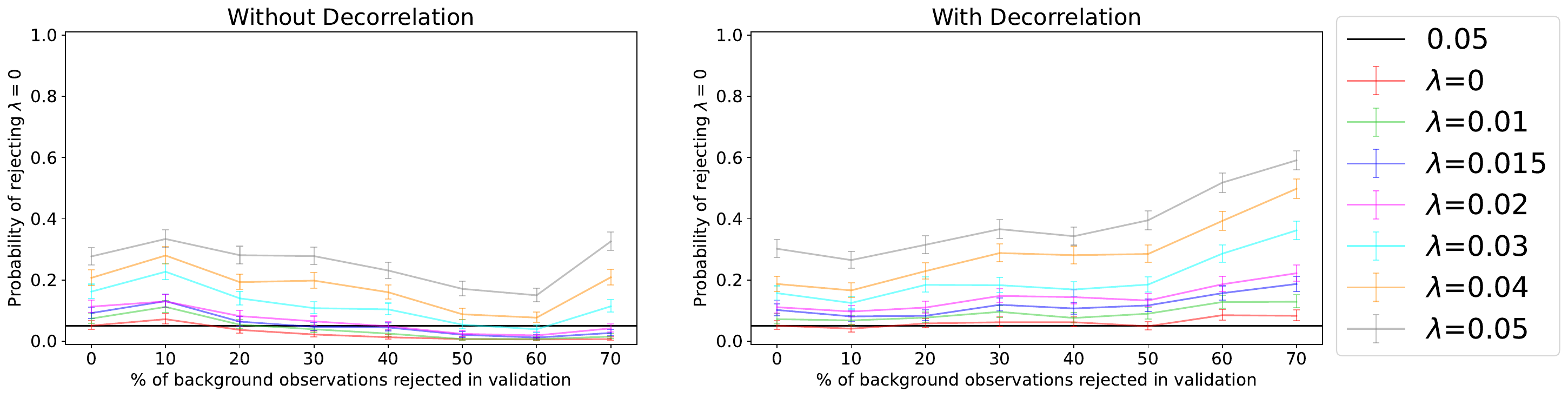}
\end{center}
\caption{The type-I error rate ($\lambda=0$) and power ($\lambda>0$) analysis for the W-tagging test dataset. Clopper-Pearson intervals are shown for all estimates. }
\label{img:wtagging_test_power}
\end{figure}

Regarding signal detection, using the model selection in eq.~\eqref{eq:test_selection}, we choose $K_*=35$ as the best-calibrated test. Figures displaying the empirical type-I error and p-value plots can be found in Appendix~\ref{sec:wtagging_appx}. Using the test dataset, we perform a power analysis of the test as the non-decorrelated and decorrelated classifiers are used to filter more data; see Figure~\ref{img:wtagging_test_power}. Figures~\ref{img:wtagging_background_fit_correlated} and \ref{img:wtagging_background_fit_decorrelated} show one instance of the observed data and fitted background for the original and decorrelated classifiers, respectively. Note that when sculpting is present, the estimated background includes a bump in the signal region. Hence, the signal-enriched tests that use the non-decorrelated classifier become more conservative, as shown by the near-zero type-I error rate in Figure~\ref{img:wtagging_test_power}, and lose power. Conversely, the signal-enriched test based on the decorrelated classifier avoids sculpting and approximately achieves the desired type-I error while improving its power as more data is filtered.

% Regarding the background estimation presented in Figures \ref{img:wtagging_background_fit_correlated} and \ref{img:wtagging_background_fit_decorrelated}, we note that although all the information outside the signal region is used to estimate the background, the censored MLE test, see eq.~\eqref{eq:censored_MLE_test}, gives negligible weight to regions with low counts of events. Thus, the right tail of the invariant mass spectrum does not introduce bias in the background estimation. 

\begin{figure}[ht]
\begin{center}
\includegraphics[width=\textwidth]{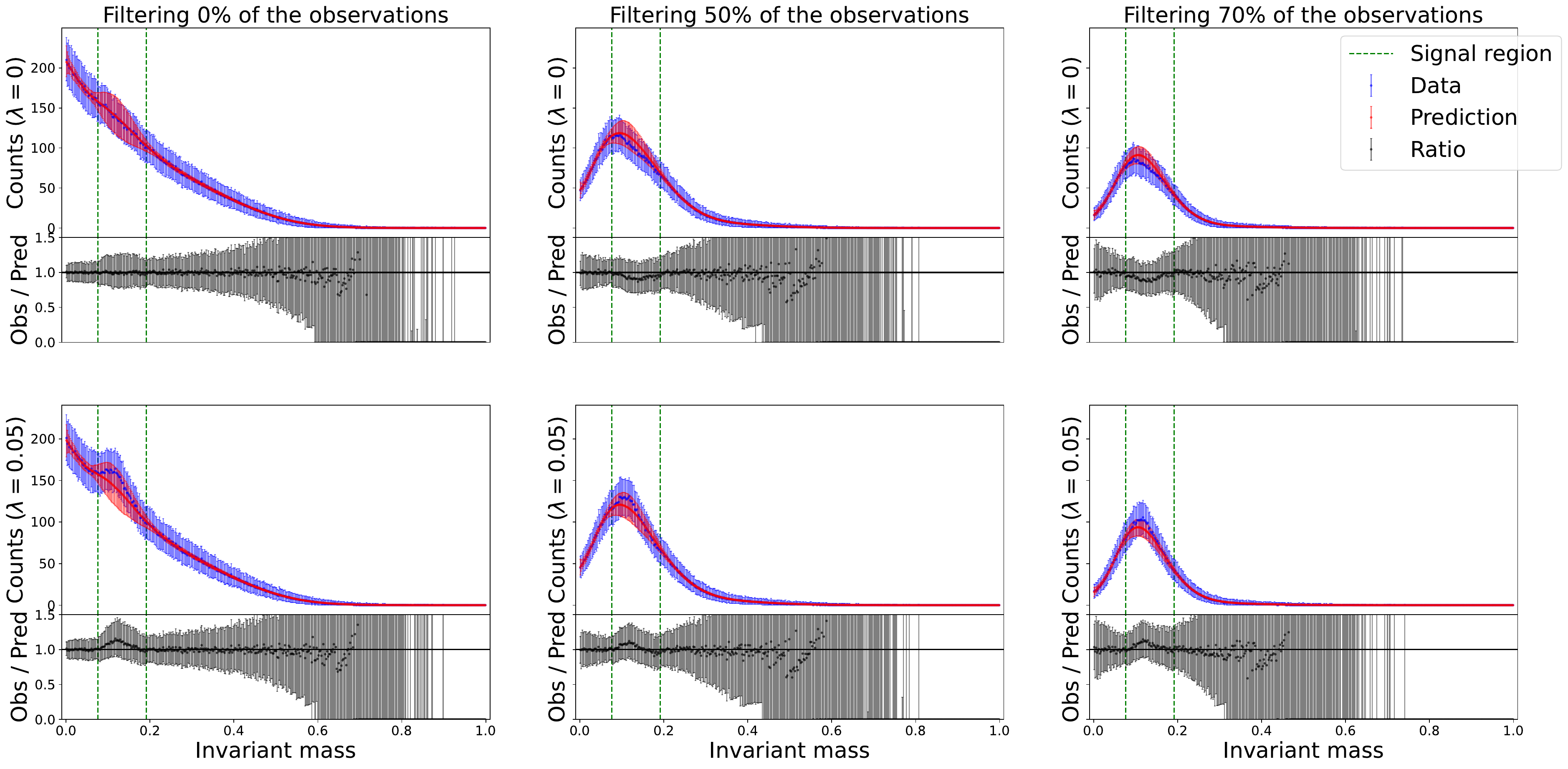}
\end{center}
\caption{Background estimates for the W-tagging test dataset as observations are filtered by a non-decorrelated classifier. In the first row, the null hypothesis is true ($\lambda=0$), while in the second it is false ($\lambda=0.05$). We report $95\%$ variability intervals from 1000 simulations. %Note the sculpting in the signal region. , and their midpoint is the median of the simulations.
}\label{img:wtagging_background_fit_correlated}
\end{figure}

\begin{figure}[ht]
\begin{center}
\includegraphics[width=\textwidth]{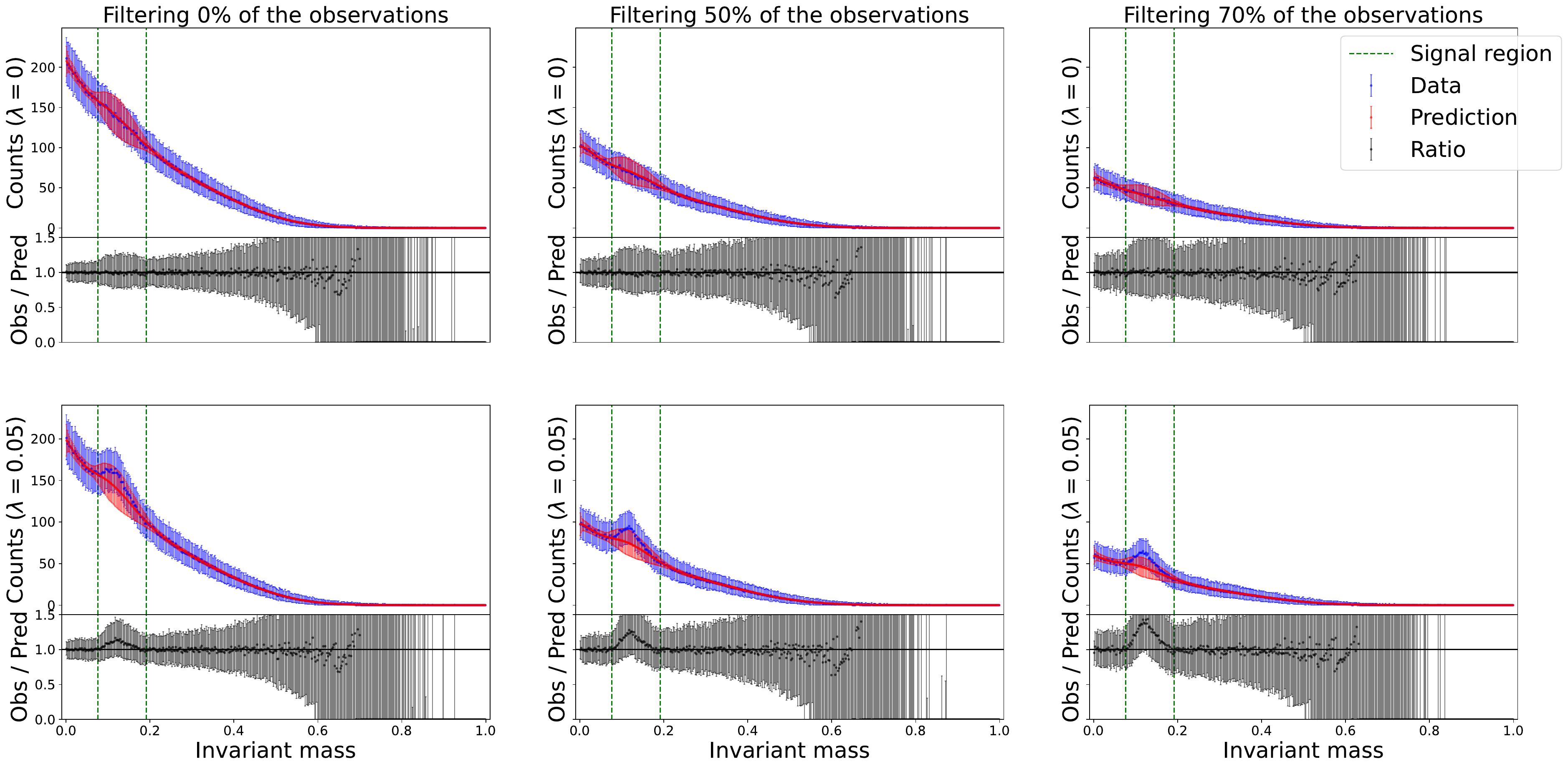}
\end{center}
\caption{Background estimates for the W-tagging test dataset as observations are filtered by the decorrelated classifier. We report $95\%$ variability intervals from 1000 simulations. %, and their midpoint is the median of the simulations. Note that there is no sculpting.
}\label{img:wtagging_background_fit_decorrelated}
\end{figure}

\subsection{Detection of exotic high-mass resonance events}\label{sec:3b_4b}
Here we consider the search for exotic high-mass resonance events. The invariant mass is the protected variable, and the kinematic properties of the four jets produced in the experiment are used for training the auxiliary classifier (details in Section~\ref{sec:data}).

The dataset contains 463,848 3b events, 62,993 4b events, and 44,196 signal events. This was split into training data (3b: 50k, signal: 40k), validation data (3b: 120k) and test data (3b $\approx$ 294k, 4b $\approx$ 63k, signal $\approx$ 4k). 
The 3b and signal training data are used to fit a random forest classifier with 1000 trees and a minimal node size of 100 samples. The 3b validation data is used to train the CDOT Algorithm in Appendix~\ref{app::decorr} using Approach 2. The test data is used for the evaluation of the decorrelation algorithm and the signal detection test. This experiment is also used to check the robustness of our procedure to background misspecification, as the training is performed on the 3b background, whereas the test is performed on the 4b background.

\begin{figure}[hb]
\begin{center}
\includegraphics[width=\textwidth]{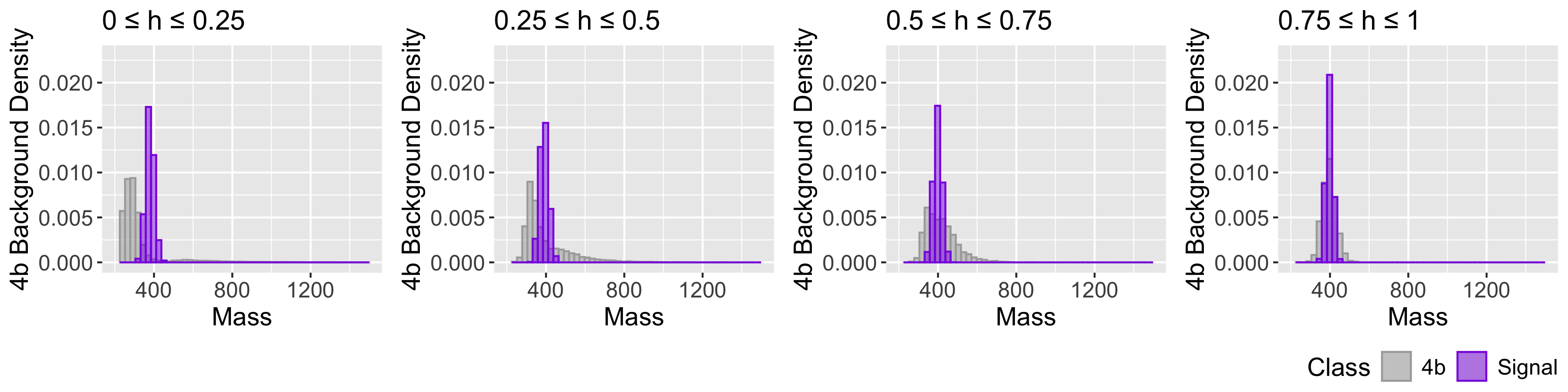}
\end{center}
\caption{Density plots of the invariant mass for the 4b data set with signal, for different ranges of the classifier ($h$) without any decorrelation. We observe sculpting.}
\label{img:4bnoDecorrelation}
\end{figure}
 
Figure \ref{img:4bnoDecorrelation} shows the effect of sculpting produced on the 4b background when signal enrichment is performed using the non-decorrelated classifier. In contrast, Figure \ref{img:3b4bDecorrelation} displays the distributions after decorrelating the classifier using CDOT. Note that though CDOT is trained on just the 3b background, it is robust to decorrelating the 4b background too.

\Revision{Finally, we turn to signal detection. The proposed test \eqref{eq:censored_MLE_test} depends on fitting the background with the Bernstein basis \eqref{eq:bernstein_basis}, which is supported on the interval $[0,1]$.} Since the $3b$ and $4b$ data sets are not supported on the unit interval, we map them to this interval. A visual inspection of the $3b$ background data in the validation dataset indicates that it decreases approximately exponentially. Thus, we transform all data in this section using the inverse CDF function of an exponential. Namely, given an observation, we apply the function $m \mapsto 1 - \exp\left\{- r \cdot (m - b)\right\}$, where $b$ is the minimum invariant mass of the $3b$ validation background and $r=0.003$, which provide an adequate fit to the data.

\begin{figure}[ht]
\begin{center}
\includegraphics[width=\textwidth]{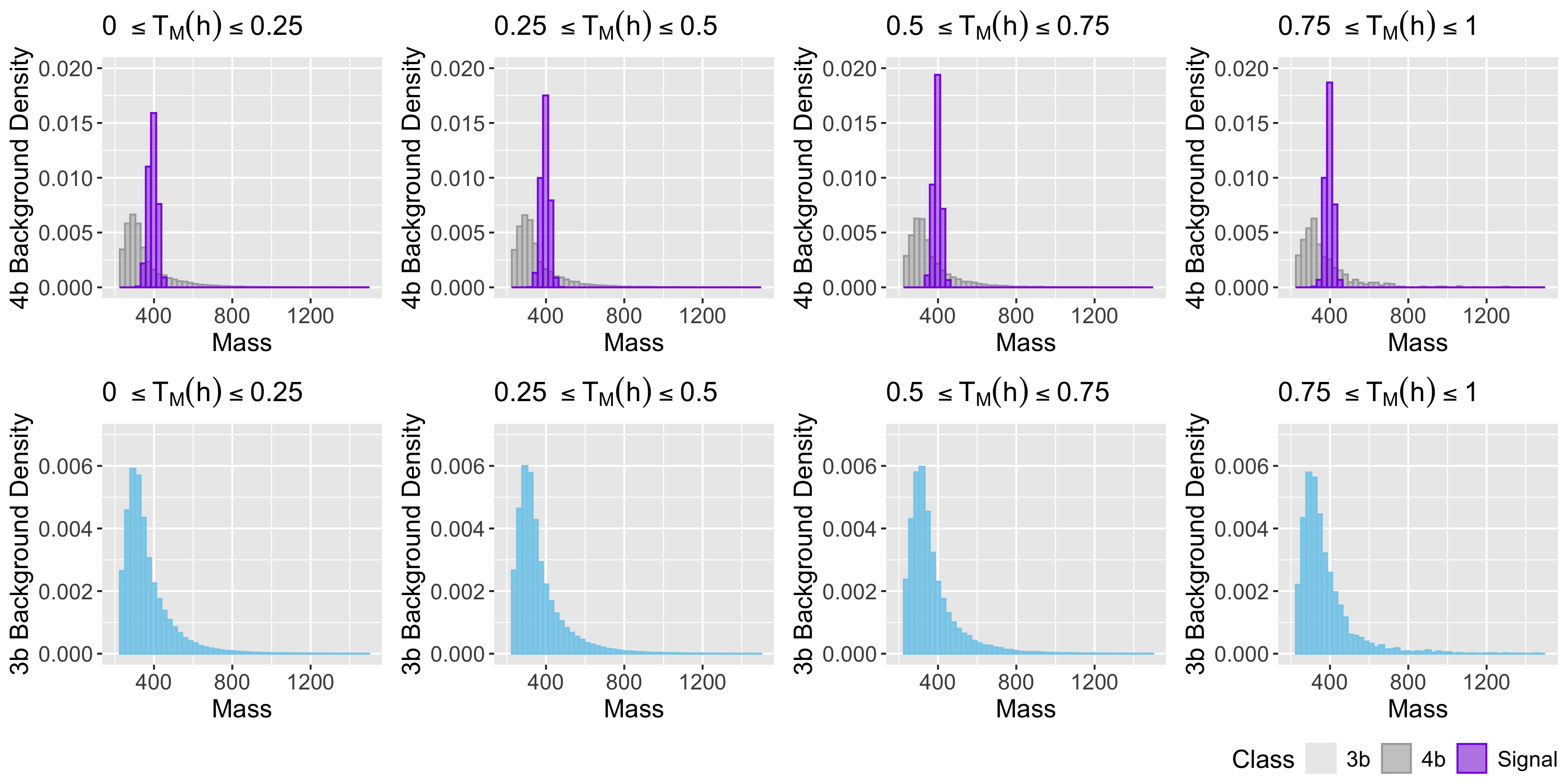}
\end{center}
\caption{Post-decorrelation density plots of the invariant mass for the 4b data set with signal (top) and the test 3b background (bottom), for different ranges of the transformed classifier ($T_M(h)$). We see that CDOT trained on 3b avoids sculpting 4b as well as 3b.}
\label{img:3b4bDecorrelation}
\end{figure}

\begin{figure}[H]
\begin{center}
\includegraphics[width=\textwidth]{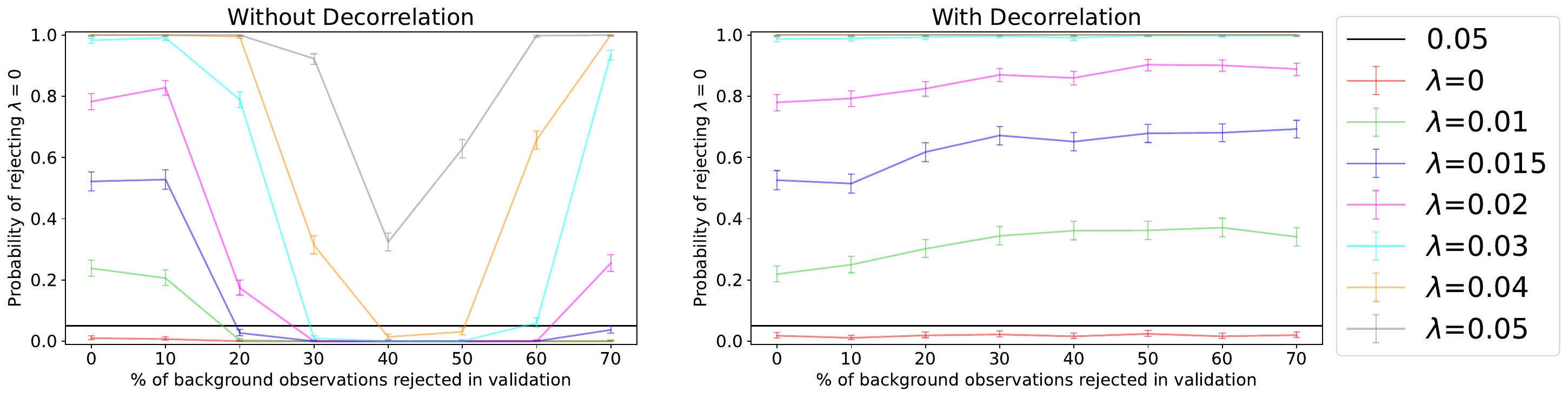}
\end{center}
\caption{Validity and power analysis on the 3b test dataset. }\label{img:3b_3b_test_power}
\end{figure}

\begin{figure}[H]
\begin{center}
\includegraphics[width=\textwidth]{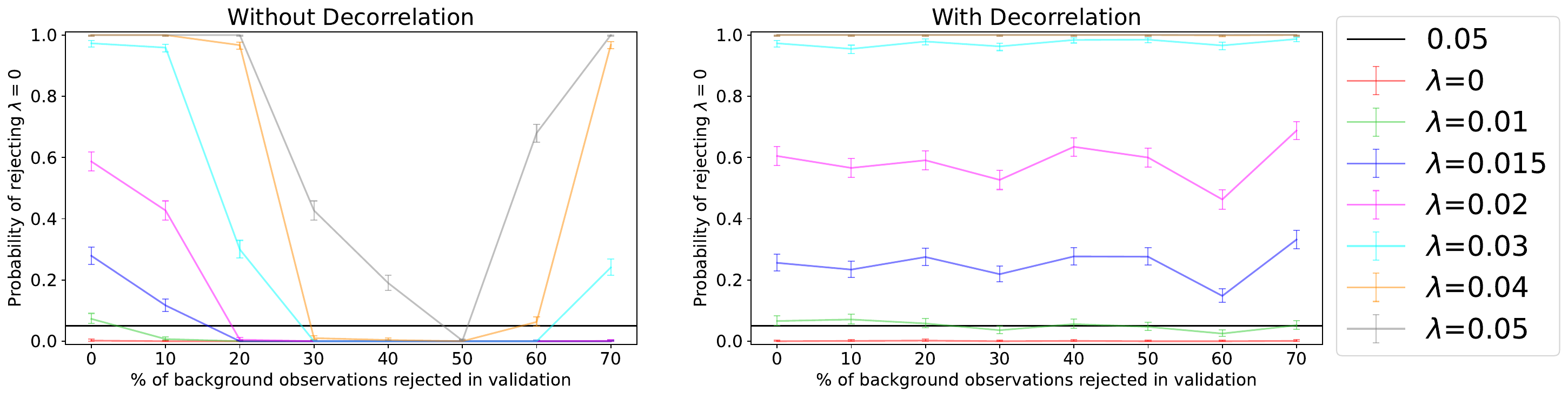}
\end{center}
\caption{Validity and power analysis on 4b test dataset. }\label{img:3b_4b_test_power}
\end{figure}

Analogously to the previous analysis, using the 3b validation dataset, we find that $K_*=20$ provides the best-calibrated test (see corresponding plots  in Appendix \ref{sec:3b_4b_appx}). The power analysis with the 3b test dataset is shown in Figure~\ref{img:3b_3b_test_power}, while the same analysis with the 4b test dataset is displayed in Figure~\ref{img:3b_4b_test_power}. When there is no distribution shift between the validation and test datasets, the conclusions are similar to the previous section. The decorrelated classifier preserves the shape of the background distribution; hence, the tests remain calibrated when enriched by a decorrelated classifier, which increases their power. However, when there is a distribution shift between the validation and test datasets, we observe that the decorrelated classifier does not lead to an improvement in power, while the non-decorrelated classifier actually decreases the power. Importantly, all tests remain valid, albeit conservative, despite the background misspecification. Plots displaying the background estimation are deferred to Appendix~\ref{sec:3b_4b_appx}.

\section{Discussion}
\label{sec:conc}

In this paper, we implemented the entire pipeline of signal detection in particle physics using a signal detection test after performing signal enrichment using a decorrelated transformed classifier. We conclude from the experiments that both signal enrichment using cuts on a classifier as well as decorrelating the classifier before performing cuts improves the power of the signal detection test. We also demonstrate in Section~\ref{sec:3b_4b} that the signal detection test, as well as the decorrelation algorithm proposed in this paper, CDOT, are robust to some background model misspecification. In the future, it would be interesting to study the robustness of the procedures to different kinds of model misspecifications. 

We note that CDOT, along with CNOTS \citep{algren2024decorrelation}, are post-processing decorrelation algorithms, meaning they can be applied to any pre-trained classifier to decorrelate it. The advantage of this is that we do not have to re-train the classifier again, lowering the computational cost of the whole procedure. It also has the advantage that domain experts can train a classifier for a particular application using their expert knowledge, and then use CDOT to decorrelate this pre-trained classifier. There may be some loss of signal detection power since we are optimizing the two problems, finding the optimal classifier and a decorrelated classifier, sequentially instead of together. From Section~\ref{sec:wtagging}, it appears that CDOT compares favorably with existing state-of-the-art decorrelation methods.

Finally, we note that it is possible to extend the decorrelation algorithm to accommodate multivariate protected variables by changing the conditional density estimation to accommodate higher-dimensional conditioned variables. However, the signal detection test would have to be modified to use a multivariate score test as in \cite{bickelTailormadeTestsGoodness2006}.

\pagebreak

\if0\blind
{
\textbf{Acknowledgments} We thank Patrick Bryant, John Alison, Ben Nachman, David Shih, Sara Algeri, Tudor Manole, Kenta Takatsu, Lukas Heinrich, Tobias Golling, Malte Algren, Johnny Raine, Samuel Klein, the STAMPS research group at CMU and the participants of \textit{Systematic Effects and Nuisance Parameters in Particle Physics Data Analyses} workshop and \textit{PHYSTAT - Statistics meets ML} at Imperial for useful discussions and comments. The authors gratefully acknowledge NSF grants PHY-2020295, DMS-2053804, and DMS-2310632.
}\fi

% \bibliography{paper}
% \putbib
% \end{bibunit}

% \newpage
\appendix

\bigskip
\begin{center}
{\large\bf SUPPLEMENTARY MATERIAL}
\end{center}

\addcontentsline{toc}{section}{Supplementary material}
\etocdepthtag.toc{mtappendix}
\etocsettagdepth{mtchapter}{none}
\etocsettagdepth{mtappendix}{subsection}
\etocsettagdepth{mtreferences}{section}
\renewcommand{\contentsname}{Supplementary material}
{
\renewcommand{\contentsname}{\normalsize Table of contents}
\tableofcontents
}

\pagebreak

% Reset the ToC depth for the appendix
% \etocsetnexttocdepth{section}

% Generate a local table of contents for the appendix
% \section*{Appendix Table of Contents}
% \localtableofcontents

% Appendix with Second Bibunit
% \begin{bibunit}

\section{Derivation and implementation of CDOT}
\label{app::decorr}

In this section, we show that the transformed classifier we need in Section~\ref{subsec:decorr} is given by the optimal transport map from the conditional density $p_b(z|m)$ to $p_b(z)$, where $Z = h(X)$, $X \sim p_b$, and $M$ is the protected variable.

\begin{lemma}\label{lemma:OT}
Given any two random variables $(Z, M)$, a map that minimizes \begin{equation}
\E\left[||T_M(Z) - Z||^2\right]
\end{equation} subject to $T_M(Z)$ independent of $M$ and $T_M(Z)$ has the same marginal distribution as $Z$, is given by the optimal transport map $T_m$ from $F_{z|m}$ to $F_z$ under the $\ell^2$ cost function, where $F_{z|m}$ and $F_z$ are the cdfs of $Z|M$ and $Z$ respectively.
\end{lemma}
{\emph{Proof:}} Let $T_m$ denote the optimal transport map from $F_{z|m}$ to $F_z$, then by definition $T_m$ minimizes $\E\left[c(Z, T_M(Z))|M = m\right]$ 
for each $m$
and $P(T_m(Z) \leq t | M = m) = F_z(t)$, where $c(x,y)$ is the cost function. Now let $f_m(Z)$ be any other map such that given $M = m$, $f_m(Z)$ is distributed as $F_z$, that is, $P(f_m(Z) \leq t | M = m) = F_z(t)$. Then, by definition, when the cost function is the $\ell^2$ loss, $c(x,y) = ||x - y||^2$, 
\begin{align}
    & \E\left[||T_m(Z) - Z||^2|M = m\right]dP(m) \leq \E\left[||f_m(Z) - Z||^2|M = m\right]dP(m) 
\end{align}
Then
\begin{align}
    \E\left[||T_m(Z) - Z||^2\right] &= \int \E\left[||T_m(Z) - Z||^2|M = m\right]dP(m) \\
    &\leq \int \E\left[||f_m(Z) - Z||^2|M = m\right]dP(m) = \E\left[||f_m(Z) - Z||^2\right].
\end{align}
This shows that $T_m(Z)$ is independent of $M$, its marginal distribution is $F_z$ and it minimizes $\E[||T_M(Z) - Z||^2]$ among all such maps. Notice that in this result, $M$ or $Z$ can be of any dimension since the dimension does not enter in the argument.

In our case where we make the classifier output $Z = h(X)$ independent of $M$,  the form of the optimal transport map from $F_{z|m}$ to $F_z$ when $h(X)$ and $M$ are one-dimensional is derived by observing that by definition,
\begin{align*}
   & F_z(T_m(t)) = P(T_m(Z) \leq T_m(t) | M = m) = F_{z|m}(t) \\
  \implies  & T_m(t) = F_z^{-1}\left(F_{z|m}(t) \right).
\end{align*}

We remark that an alternative approach is to choose $T_m(z)$ such that $M$ is independent of $T_M(Z)$ and such that the Wasserstein distance between
$(Z,M)$ and $(T_M(Z),M)$ is minimized. The solution in this case is to define $T_m(z)$ to be the optimal transport map from $F_{z|m}$ to $B$ where $B$ is the barycenter, that is, the distribution that minimizes $\int W^2(F_{z|m},B) dP(m)$.

The CDOT (Classifier Decorrelated via Optimal Transport) procedure is detailed below: 

\begin{method}
\label{method:CDOT}
CDOT decorrelated classifier output on the experimental data $\mathcal{E}$. \\
Input: $\X_i \in \mathcal{B}$ (Training: $\mathcal{B}_1$, Validation: $\mathcal{B}_2$), $\Y_j \in \mathcal{S}, \W_k \in \mathcal{E}, M_{i}^\X, M_{j}^\Y, M_{k} \ \forall \ i,j,k$.\\
Output: $T_{M_k}(h(W_{k}))$, decorrelated classifier output for every $k$.
\begin{enumerate}
    \item Train a probabilistic classifier $h$ on $X_i \in \mathcal{B}_1, Y_j \in \mathcal{S} \ \forall \ i,j$, which gives the probability of being a signal event given the data.
    \item Evaluate $h(X_i), h(Y_j), h(W_k) \ \forall \ i,j,k$.
    \item Estimate the conditional distribution $F_{z|m}$ for $X_i \in \mathcal{B}_2$ ($Z = h(X)$) using one of the two approaches.
    \begin{enumerate}
        \item Approach 1: 
        \begin{enumerate}
            \item Choose splits in the range of $M$, $\{m_0, m_1, \ldots, m_l\}$.
            \item For each split $i \in \{1, \ldots, l\}$, estimate $\widehat{F}_{z|m}^{(i)}$ using a kernel conditional distribution estimator with fixed optimal bandwidths on  $Z_j = h(X_j)$ and $M_{j}^\X$ for $\X_j \in \mathcal{B}_2$ such that $m_{i - 1} \leq M_{j}^\X \leq m_i$.
            \item Then for the experimental data, \begin{equation}
               \widehat{F}_{z|m}(h(W_k)) = \widehat{F}_{z|m}^{(i)}(h(W_k)) \textif m_{i - 1} \leq M_{k} \leq m_i 
            \end{equation}
        \end{enumerate}
        \item Approach 2:
        \begin{enumerate}
            \item Estimate $\widehat{\mu}(M_{i}^\X) = \widehat{\E}\left[logit(h(X_i))| \log(M_{i}^\X)\right]$ (non-parametric regression).
    \item Evaluate $\widehat{\eta}_i = logit(h(X_i)) - \widehat{\mu}(M_{i}^\X)$.
    \item Estimate $\widehat{\sigma}^2(M_{i}^\X) = \widehat{\E}\left[\widehat{\eta}_i^2| \log(M_{i}^\X)\right]$ (non-parametric regression).
    \item Evaluate $\widehat{\epsilon}_i = \widehat{\eta}_i/\widehat{\sigma}(M_{i}^\X)$.
    \item Estimate $\widehat{F}_{\epsilon|m}$ using a kernel conditional distribution estimator with a fixed optimal bandwidth on $\epsilon_i$'s given $\log(M_{i}^\X)$'s. 
    \item Then for the experimental data, \begin{equation}
    \widehat{F}_{z|m}(h(W_k)) = \widehat{F}_{\epsilon|m}\left[ \left(logit(h(W_k)) - \widehat{\mu}(M_{k})\right)/\widehat{\sigma}(M_{k})\right]
    \end{equation}
        \end{enumerate}
    \end{enumerate}
    \item Estimate  $\widehat{F}_z(t)$ using the empirical CDF estimator on the background data $\mathcal{B}_2$.
    \item Evaluate the CDOT classifier output as $T_{M_k}(h(W_k)) = \widehat{F}_z^{-1}\left(\widehat{F}_{z|m}(h(W_k))\right)$. 
\end{enumerate}
\end{method}

Note that all the kernel conditional distribution (CCDF) estimators were obtained using the \texttt{npcdist} function in the R package \texttt{np} \citep{li2008nonparametric, hayfield2008nonparametric, li2013optimal}. The fixed optimal bandwidths in both approaches above were chosen using the least-squares cross-validation method of \cite{li2008nonparametric} and \cite{li2013optimal}.

\Revision{
\section{Semiparametric signal detection}\label{appx:signal_detection}

In \cref{sec:efficient_estimator}, we define efficient estimators. \Cref{sec:efficient_score_function} presents a method for identifying such estimators using score functions. Both sections summarize key results from \cite{vaartAsymptoticStatistics1998,vaartSemiparametric2002}. Finally, in \cref{sec:efficient_estimator_known_background} and \cref{sec:efficient_estimator_parametric_background}, we apply the results to find efficient estimators and test for the model \eqref{eq:model} assuming known and parametric background, respectively.

Henceforth, we let $L_2(F)$ denote the set of all measurable functions $g: \Omega \to \mathbb{R}$ such that their second moment with respect to the measure $F$ exists: $\int g^2\ dF <\infty$.

\subsection{Asymptotically semiparametric efficient estimators}\label{sec:efficient_estimator}

This section defines the concept of an efficient estimator, see \cref{def:efficient_estimator}. Intuitively, the idea is that an efficient estimation is asymptotically unbiased and has the \textit{smallest variance}. To formalize this intuition, we pursue the following approach: given a model $\mathbb{F}$, and a functional $\psi: \mathbb{F}\to \mathbb{R}^k$, estimating the functional on the model $\mathbb{F}$ cannot be harder than estimating the functional over a sub-model $\mathbb{F}_0 \subset \mathbb{F}$. Thus, by choosing sub-models that are hard enough, we can characterize the hardness of estimating a functional in the original model $\mathbb{F}$.

\begin{definition}[Differentiable path or sub-model] We call the map $t \to F_{t,g}$ a differentiable path at $t=0$ with score function $g$ if \begin{equation}\label{eq:score_function}
\int\left(\frac{dF_{t,g}^{1/2}-dF^{1/2}}{t}-\frac{1}{2}\ g\ dF^{1/2}\right)^2 \to 0 \as t \to 0.
\end{equation}\end{definition}

In \eqref{eq:score_function}, $g$ plays the role of a derivative. Each differentiable path, defines a sub-model $\{F_{t,g}: t\geq 0\}\subseteq \mathbb{F}$ with score function $g$. 

We briefly note that all score functions are centered and have finite variance, which guarantees that the next statements are well defined \begin{proposition}
For any $g$ that satisfies \eqref{eq:score_function}, it holds that \begin{equation}
E_{X\sim F}\left[g(X)\right]=0 \textand V_{X\sim F}\left[g(X)\right]<\infty.
\end{equation}
\end{proposition}

We define a tangent set as the collection of score functions corresponding to the considered sub-models.  

\begin{definition}[Tangent set]\label{def:tangent_set} Given a collection of differentiable paths $\mathbb{M}$, 
\begin{equation}
\Lambda(F) = \bigcup\{g: \exists\ m\in\mathbb{M} \text{ that is a differentiable path with score function } g \text{ at } t=0 \}.
\end{equation}
\end{definition}

Given the tangent set, we can define a functional derivative for $F \to \psi(F)$.

\begin{definition}[Differentiable functional] A functional $\psi:\mathbb{F} \to \mathbb{R}^k$ is differentiable at $F \in \mathbb{F}$ relative to the tangent set $\Lambda(F)$ if there exists a continuous bounded linear map $\dot\psi_F: L_2(F) \to \mathbb{R}$ in the tangent set $\Lambda(F)$ such that for every map $t\to F_{t,g}$ with $g\in \Lambda(F)$, it holds that \begin{equation}\label{eq:psi_derivative}
\frac{\psi(F_{t,g})-\psi(F)}{t}\to \dot\psi_F\ g \as t\to0. 
\end{equation}
\end{definition}%
By the Riesz representation theorem, if $\dot\psi_F$ exists, then it can be represented as an inner product: \begin{equation}\label{eq:riesz}
\dot\psi_F\ g = E_{X\sim F}\left[ \psi^*_F(X) \cdot g(X) \right].
\end{equation} Note that $\psi^*_F$ might not be uniquely defined in $\Lambda(F)$, but it is guaranteed to be unique in the closure of the linear span of $\Lambda(F)$, which we denote by $\overline{\lin(\Lambda(F))}$.

\begin{definition}[Efficient influence function]\label{def:efficient_influence_function} If $\dot\psi$ satisfying \eqref{eq:psi_derivative} exists, we call the unique solution $\psi^*_F$ of \eqref{eq:riesz} in $\overline{\lin(\Lambda(F))}$ the efficient influence function.
\end{definition} The efficient influence function by taking any solution of \eqref{eq:riesz} and projecting it into the closure of the linear space of $\Lambda(F)$.

The following lemma shows that the variance of the efficient influence function characterizes the best risk achievable by an estimator. In general, let $V^*_F$ be the covariance of the efficient influence function \begin{equation}
V^*_F = E_{X\sim F}\left[\psi^*_F(X)\cdot \psi^*_F(X)^T\right], 
\end{equation} which reduces to $V^*_F = V_{X\sim F}\left[\psi^*_F(X) \right]$ if if $\psi^*_F(X)$ is a scalar.

\begin{lemma}[Variance of efficient influence function lower bounds squared risk]\label{lemma:LAM}
Let $\psi$ be a differentiable functional at $F$ relative to the tangent set $\Lambda(F)$ with efficient influence function $\psi^*_F$. It holds that \begin{equation}\label{eq:LAM}
\lim_{\delta\to0}\liminf_{n\to\infty}\sup_{Q: \TV(F,Q)<\delta}E_{X\sim Q}\left[\sqrt{n}(\ T_n(X)-\psi(Q)\ )\right]^2 \geq E_{X\sim \Normal\left(0,V^*_F\right)}\norm{X}_2^2
\end{equation} where $\TV$ denotes the total variation distance.
\end{lemma} 

Importantly, recall from \cref{def:efficient_influence_function} that the efficient influence function is defined relative to $\Lambda(F)$. As we consider more sub-models, the tangent space expands, and the right-hand side of \eqref{eq:LAM} cannot decrease. 

In practice, we aim to construct an estimator that attains the lower bound in \eqref{eq:LAM}. To achieve this, we specify a sub-model and an estimator, then check whether the estimator's variance equals that of the efficient influence function. If not, we may consider additional sub-models. For our purposes, however, a single sub-model always suffices.

To ensure asymptotically valid inference, we need the limit distribution of a sequence of estimators. Although we omit the technical details, \eqref{eq:LAM} suggests that an optimal limit distribution is a centered Gaussian distribution with covariance given by $V_{X\sim F}[\psi^*_F(X)]$. For a rigorous treatment of this result, see Chapter 2 of \cite{vaartAsymptoticStatistics1998}. The next two definitions characterize an efficient estimator as any smooth estimator that achieves this normal limit.

\begin{definition}[Regular estimator] A sequence of estimators $T_n$ is regular at $F$ for estimating $\psi(F)$ relative to the tangent set $\Lambda(F)$ if there exists a probability measure $L$ such that \begin{equation}
\sqrt{n}\left(T_n-\psi(F_{n^{-1/2},g})\right) \toD L \as n\to\infty \quad \forall g \in \Lambda(F)
\end{equation}
\end{definition}

\begin{definition}[Asymptotic efficiency]
\label{def:efficient_estimator}
A sequence of estimators $T_n$ is asymptotically efficient for estimating $\psi(F)$ if it is regular at $F$ relative to the tangent set $\Lambda(F)$ and has limit distribution $L = \Normal\left(0, V^*_F\right)$.
\end{definition}

This definition implies the following convergence in distribution:
\begin{equation}
\sqrt{n}\left(T_n - \theta\right) \toD \Normal\left(0, V^*_F\right) \quad \text{as } n \to \infty.
\end{equation} By applying Slutsky's theorem, we derive a result that facilitates practical inference:
\begin{equation}\label{eq:practical_clt}
\left[V^*_{F_n}\right]^\dagger \sqrt{n}\left(T_n - \theta\right) \toD \Normal\left(0, I\right) \quad \text{as } n \to \infty,
\end{equation}
where $\left[V^*_{F_n}\right]^\dagger$ denotes the Moore–Penrose pseudoinverse of $V^*_{F_n}$ and $I$ is the identity matrix.

\subsection{The efficient score function estimator}\label{sec:efficient_score_function}

In this section, we present the standard approach for deriving the efficient influence function and constructing asymptotically efficient estimators. This method relies on the model's score function to enable practical inference through equation~\eqref{eq:practical_clt}.

In our work, there are two models. Let 
\begin{equation}\label{eq:model_appx}
\mathbb{F}=\{F_{\theta,s} : \theta \in \Theta \textand s \in \mathcal{S}\} \where \mathcal{S} = \left\{s \in \mathbb{D}: \int_{\Sr} s(x)\ dx=1\right\}.
\end{equation} where $\mathbb{D}$ is the set of densities supported on $\Omega$. When the background is known, we define
\begin{equation}\label{eq:model_known_background}
f_{\theta,s} = (1-\lambda) \cdot b + \lambda \cdot s \comma \theta=\lambda \textand \Theta=(0,1).
\end{equation} Alternatively, when the background is parametric, we define
\begin{equation}\label{eq:model_parametric_background}
f_{\theta,s} = (1-\lambda) \cdot b_\gamma + \lambda \cdot s \comma \theta=(\lambda,\gamma) \textand \Theta=(0,1) \times A
\end{equation} where $A$ is an open subset of $
\{\gamma \in \mathbb{R}^k : b_\gamma \in \mathbb{D}\}$.

The score function of the model \eqref{eq:model_appx} is given by \begin{equation}
g_{\theta,s} = \nabla_{\theta}\log f_{\theta,s}.
\end{equation}

In all cases, we use sub-models defined by differentiable paths of the form \begin{equation}\label{eq:submodel}
t \to F_{\theta, s_t(h)} \where s_t(h)= s\cdot (1+t\cdot h) \textand \mathbb{H} = \left\{ h \in \mathcal{S}: \int_{\Omega} h(x)\ dx \cdot s(x) = 0 \right\}.
\end{equation} It follows that $t \to F_{\theta, s_t}$ is a differentiable path at $t=0$ with score function \begin{equation}
g_{s,h} = \left[\frac{\partial}{\partial t}\log f_{\theta,s_t(h)}\right]_{t=0}= \lambda \cdot \frac{s\cdot h}{f_{\theta,s}}.
\end{equation}

We obtain the tangent space, see \cref{def:tangent_set}, by considering all such differentiable paths: \begin{equation}\label{eq:tangent_set}
\Lambda(F) = \left\{g_{s,h}: \exists h \in \mathbb{H}\ \text{ such that } g_{s,h} \text{ is the score function of }(t \to F_{\theta, s_t(h)}) \right\}.
\end{equation}

\begin{definition}[Efficient score function]
We define the efficient score function $g^*_{\theta}$ as the orthogonal projection of the score function $g$ onto the closure of the linear span of $\Lambda(F)$. \begin{equation}\label{eq:efficient_score}
g^*_{\theta} = g_{\theta,s} - \Pi g_{\theta,s} \where \Pi g_{\theta} = \inf_{q \in \overline{\lin \Lambda(F)}} \int  \norm{g_{\theta}-q}_2^2\ dF_{\theta,s}
\end{equation}.
\end{definition} Note that we have omitted $s$ in the notation for the efficient score function. This is because, in our analyses, the efficient score function does not depend on the nuisance signal $s$, allowing us to use a simpler CLT for our Z-estimator than is usually required. 

Given the efficient score, we define the set of efficient score estimands \begin{equation}\label{eq:efficient_score_estimand}
\Theta_* \in \left\{\theta: E_{X\sim F_{\theta,s}}\left[g^*_{\theta}\right](X)=0\right\}.
\end{equation} Furthermore, we define an efficient score estimator $\theta_*(F_n)$ as any zero of the efficient score evaluated on the available data:
\begin{equation}\label{eq:efficient_score_estimator}
\theta_*(F_n) \st \sum_{i=1}^n g^*_{\theta_*(F_n)}(X_i)=0.
\end{equation} 

Finally, the following theorem establishes the conditions under which an efficient score estimator is asymptotically efficient.

\begin{theorem}\label{clt_z_estimator} If the following conditions hold \begin{itemize}
\item $\{g^*_\theta : \theta \in \Theta\}$ is $F$-Donsker
\item The map $\theta \to g^*_\theta$ is differentiable at $\theta_*$ with non-singular derivative. 
\item The map $\theta \to g^*_\theta$ is continuous in $L_2(F)$ at $\theta_*$
\end{itemize} Then, the efficient influence function, see \cref{def:efficient_influence_function}, of $\psi(F_{\theta,s})=\theta$ is \begin{equation}\label{eq:efficient_influence_function}
\psi^*_{F_{\theta_*,s}} = \left[V^*_{F_{\theta_*,s}}\right]^{-1}g^*_{\theta_*} \where V^*_{F_n}=E_{X\sim F_{\theta_*,s}}\left[g^*_{\theta_*}(X)g^*_{\theta_*}(X)^T\right].
\end{equation} Furthermore, if for any $\theta_* \in \Theta_*$, $\theta_*(F_n)$ is consistent estimator of $\theta_*$,  it follows that $\theta_*(F_n)$ is asymptotically efficient for $\psi(F_{\theta_*,s})=\theta_*$, see \cref{def:efficient_estimator}.
\end{theorem} 

We do not detail what $F$-Donsker condition entails, but note that its definition can be found in section 6.2 of \cite{vaartAsymptoticStatistics1998}, and that all the efficient score functions that we use in this article satisfy the condition.

In practice, \cref{clt_z_estimator} is useful for inference due to the following limit holding: \begin{equation}
\left[V^*_{F_n}\right]^\dagger\sqrt{n}\left(\theta_*(F_n)-\theta_*(F_{\theta,s})\right) \toD \Normal\left(0,I\right).
\end{equation}

% The following lemma is an informal version of the asymptotic normality theorem for Z-estimators presented in section 5.3 of \cite{vaartAsymptoticStatistics1998}. We refer the reader to it for the specific regularity conditions needed. It should be noted that the result does not require probability measures to have a density. Hence, we use $\int g\ dF$ to denote that we integrate the function $g$ with respect to the measure $F$.

% \begin{lemma}\label{clt_z_estimator}
% Let $\theta(F)$ be the parameters that satisfies the following zero equation \begin{equation}
% \int h(x,\theta(F))\ dF(x)= 0.
% \end{equation} 
% Let $F_n$ denote the empirical distribution.
% Then
% \begin{align}
% \sqrt{n}\left(\theta(F_n)-\theta(F)\right) &\toD \Normal(0,\tau^2(F))\\
% \where \tau^2(F) &= \int \psi^2(F,x) dF(x) \comma\\
% \psi(F,x) &= -[H(F)]^{\dagger} \cdot h(x,\theta(F))\\
% \textand H(F) &= \int \pdv{h}{\theta}(x,\theta(F))\ dF(x)
% \end{align} where $[H(F)]^\dagger$ denotes the Moore–Penrose pseudo-inverse of $H(F)$. Consider a coordinate $\lambda$ of $\theta$.
% Then 
% \begin{equation}
% \sqrt{n}\left(\frac{\theta_{\lambda}(F_n)-\theta_{\lambda}(F)}{\sqrt{\left[\tau^2(F_n)\right]_{\lambda,\lambda}}}\right) \toD \Normal(0,1).
% \end{equation}\end{lemma}

\subsection{Efficient estimator with known background}\label{sec:efficient_estimator_known_background}

\EfficientEstimatorKnownBackground*
\begin{proof}[Proof of \cref{lemma:EfficientEstimatorKnownBackground}]

For the model \eqref{eq:model_known_background}, the score function is \begin{equation}
g = \frac{s-b}{f}.
\end{equation}%
%
% The score function for $\lambda$ is $g $. 
% Consider the parametric submodel $f_t = (1-\lambda) b + \lambda s (1+t\cdot h)$ where $h$ is supported on the signal region $\Sr$, and it is mean zero with respect to the signal. 
%The nuisance tangent set is 
% \begin{equation}
% \Lambda = \left\{ \frac{h\cdot s}{f} :  h \in \mathbb{H} \right\} \where \mathbb{H} = \left\{ h : \int h \cdot s = 0 \textand \int_{\Sr} h = 1\right\}.
% \end{equation} 
Consider the submodel \eqref{eq:submodel}, the orthogonal projection of the score function $g$ onto the tangent set $\Lambda(F)$ \eqref{eq:tangent_set} is given by 
\begin{equation}
\Pi g = \frac{f}{s}\left[g-\frac{B(\Cr)}{F(\Sr)}\right] \cdot I_{\Sr}.
\end{equation} 
The efficient score $g^*$ \eqref{eq:efficient_score} is: \begin{equation}
g^* = g - \Pi g = \frac{B(\Cr)}{F(\Sr)}\cdot I_{\Sr}-\frac{1}{1-\lambda}\cdot I_{\Cr}.
\end{equation}% 
The efficient score estimand $\lambda$ is defined as the solution of \begin{equation}
\int g^*(x,\ \lambda\ )\ f(x)\ dx = 0.
\end{equation} which is \begin{equation}\label{eq:estimand1}
\lambda= 1 - \frac{F(\Cr)}{B(\Cr)}.
\end{equation} Analogously the efficient score estimator $\lambda(F_n,B)$ is defined as the plug-in estimator of \eqref{eq:estimand1}: \begin{equation}
\lambda(F_n,B) = 1 - \frac{F_n(\Cr)}{B(\Cr)}.
\end{equation} 
Since we have an explicit formula for the estimator and the estimator, we do not need to rely on \cref{clt_z_estimator} to obtain the asymptotic distribution of $\lambda(F_n,B)$ and can use a standard CLT argument.
Note that it is unbiased, \begin{equation}
E[\lambda(F_n,B)]=\lambda,
\end{equation} and its variance is \begin{equation}
V[\lambda(F_n,B)]=\frac{1}{n}\cdot (1-\lambda)\cdot \frac{F(\Sr)}{B(\Cr)}=\frac{1}{n}\cdot (1-\lambda)\cdot \left(\frac{B(\Sr)}{B(\Cr)}+\lambda\right).
\end{equation} Thus, by the standard central limit theorem and the continuous mapping theorem, we have that \begin{equation}\label{eq:clt_known_background}
\sqrt{n}\left(\lambda(F_n,B) -\lambda\right)\toD \Normal\left(0\ ,\ (1-\lambda)\cdot \frac{F(\Sr)}{B(\Cr)}\right) \as n \to \infty
\end{equation} insofar as $\lambda \in [0,1)$ and $B(\Cr)>0$.

To verify that $\lambda(F_n,B)$ is the asymptotically efficient estimator of $\lambda$, we must check that its variance asymptotically matches the variance of the efficient influence function. %
The efficient influence function \eqref{eq:efficient_influence_function}  is 
\begin{align}
\psi(F,x) &= \left[-\int \pdv{g^*}{\lambda}(t)\ f(t)\ dt\right]^{-1}g^*(x)\\
&= \frac{F(\Sr)(1-\lambda)}{B(\Cr)}\left[\frac{B(\Cr)}{F(\Sr)}\cdot I(x\in\Sr)-\frac{I(x\in\Cr)}{1-\lambda}\right] \\
&= (1-\lambda)\cdot I(x\in\Sr)-\frac{F(\Sr)}{B(\Cr)}\cdot I(x\in\Cr),
\end{align} and its variance is 
\begin{align}
\tau^2(F)=\int \psi^2(F,x)\ f(x)\ dx = (1-\lambda)^2F(\Sr) + \frac{F^2(\Sr)}{B(\Cr)} (1-\lambda) = (1-\lambda)\cdot \frac{F(\Sr)}{B(\Cr)}.%\\
% &= \frac{1-\lambda}{B(\Cr)} - (1-\lambda)^2.
\end{align} Since the variance of the estimator $\lambda(F_n,B)$ matches the variance of the efficient influence function, by \cref{def:efficient_estimator} and \cref{clt_z_estimator}, $\lambda(F_n,B)$ is the efficient estimator of $\lambda$ for $\lambda\in (0,1)$ and $B(\Cr)>0$. 

Furthermore, by \eqref{eq:clt_known_background} and Slutsky's theorem, it follows that:
\begin{equation}\label{eq:limit1}
\sqrt{n}\left(\frac{\lambda(F_n,B)-\lambda}{\tau(F_n)}\right) \toD \Normal(0,1). \for \lambda\in[0,1) \textand B(\Cr)>0
\end{equation} 
Consequently, the test $\Psi_\alpha(F_n,B)=I(T(F_n,B) > z_{1-\alpha})$ where \begin{equation}
T(F_n,B) = \sqrt{n} \cdot \frac{\lambda(F_n,B)}{\tau(F_n)}=\sqrt{n} \cdot\frac{F_n(\Sr)-B(\Sr)}{\sqrt{F_n(\Sr)(1-F_n(\Sr))}},
\end{equation} is an asymptotically valid $\alpha$-level test.
\end{proof}

\subsection{Efficient estimator with parametric background}\label{sec:efficient_estimator_parametric_background}

\EfficientEstimatorParametricBackground*
% \begin{customLemma}{2}
% Under the model in \cref{eq:model} with a parametric background, the censored MLE estimator in \cref{eq:censored_mle} is efficient and induces an asymptotically valid $\alpha$-level test\begin{equation}\label{eq:censored_MLE_test}
% \Psi^{(K)}_\alpha(F_n)=I\left(\sqrt{n}\cdot \frac{\lambda_*(F_n)}{\tau_\lambda(F_n)} > z_{1-\alpha}\right).
% \end{equation} 
% \end{customLemma}
\begin{proof}[Proof of \cref{lemma:EfficientEstimatorParametricBackground}] For the model \eqref{eq:model_known_background}, the score function is $g=(g_\lambda,g_\gamma)$ where \begin{align}
g_\lambda= (b_\gamma-s)/f \textand 
g_\gamma=\frac{1-\lambda}{f}\cdot \pdv{b}{\gamma}.
\end{align}
%
% $f_t = (1-\lambda)\cdot b_\gamma + \lambda\cdot s\cdot (1+t\cdot h)$ where $h$ is supported on the signal region $\Sr$ and it is mean zero with respect to the signal. 
% That is, the nuisance tangent set is defined by \begin{equation}
% \Lambda = \left\{ \frac{h\cdot s}{f} : \exists h \in \mathbb{H} \right\} \where \mathbb{H} = \left\{ h : \int h \cdot s = 0 \textand \int_{\Sr} h = 1\right\}.
% \end{equation} 
Consider the sub-mode \eqref{eq:submodel}. The orthogonal projection of $g$ onto $\Lambda$ is given by \begin{equation}
\Pi g = \left[g  - \frac{E[g \cdot I_{\Sr}]}{F(\Sr)}\right] I_{\Sr}. 
\end{equation} Consequently, we the efficient score \eqref{eq:efficient_score} is: \begin{equation}
g^* = g - \Pi g = g \cdot I_{\Cr} + \frac{E[g \cdot I_{\Sr}]}{F(\Sr)} I_{\Sr}, 
\end{equation} where \begin{align}
g^*_{\lambda} &= \frac{B_\gamma(\Cr)}{(1-\lambda)B_\gamma(\Sr)+\lambda} \cdot I_{\Sr}- \frac{1}{1-\lambda} \cdot I_{\Cr} \textand\\
g^*_{\gamma} &=  \frac{1-\lambda}{(1-\lambda)B_\gamma(\Sr)+\lambda} \cdot \left(\int_{\Sr} \pdv{b_\gamma}{\gamma}\right) \cdot I_{\Sr}  + \pdv{b_\gamma}{\gamma} \cdot \frac{1}{b_\gamma} \cdot I_{\Cr}.
\end{align} Then, the efficient score estimands \eqref{eq:efficient_score_estimand} of interest are: \begin{equation}
(\lambda_*(F),\gamma_*(F)) \st \int g^*(x,\lambda_*(F),\gamma_*(F))\cdot f(x)\ dx = 0 ,
\end{equation} which correspond to the following $M$-estimand \begin{align}\label{eq:estimand2}
(\lambda_*(F),\gamma_*(F)) = \argmax_{\tilde{\gamma},\tilde{\lambda}}\ \int\ \ell(x,\tilde{\lambda},\tilde{\gamma}) \cdot f(x)\ dx \st B_\gamma(\Omega) = 1
\end{align} where $\ell(x,\lambda,\gamma)= I(x \in \Sr)\cdot\log\left((1-\lambda)B_\gamma(\Sr)+\lambda\right) +  I(x\in \Cr)\cdot \log\left((1-\lambda)b_\gamma(x)\right)$. The efficient score estimators \eqref{eq:efficient_score_estimator} are the corresponding plug-in estimators $(\lambda_*(F_n),\gamma_*(F_n))$. It is worth noting that the efficient estimator for the signal strength is the same as the efficient estimator with known background, but replacing the known background with the estimated parametric background \begin{equation}
\lambda_*(F_n) = 1 - \frac{F_n(\Cr)}{B_{\gamma_*(F_n)}(\Cr)}.
\end{equation} However, there is no closed-form solution for the background parameters $\gamma_*(F_n)$. Thus, we propose an expectation-maximization \cite{dempsterMaximumLikelihoodIncomplete1977} algorithm to solve for them in \cref{sec:EM_censored_MLE}. 

Regarding inference, by \cref{clt_z_estimator}, the efficient influence function \eqref{eq:efficient_influence_function} is \begin{align}\label{eq:influence_function}
\psi(F,x) = \left[\int \pdv{\ell}{\theta,\theta}(x,\theta_*(F))\cdot f(x)\ dx\right]^\dagger \cdot \pdv{\ell}{\theta}(x,\theta_*(F)) 
\end{align} where $\theta_*(F)=(\lambda_*(F),\gamma_*(F))$, and its variance is $
\tau^2(F)=\int \psi(F,x)\cdot \psi^t(F,x)\cdot f(x)\ dx
$. Furthermore, let $\tau_\lambda(F_n)=\sqrt{\tau^2_{\lambda,\lambda}(F_n)}$. It follows by \cref{clt_z_estimator}  that \begin{equation}\label{eq:limit2}
\sqrt{n}\left(\frac{\lambda_*(F_n)-\lambda_*(F)}{\tau_\lambda(F_n)}\right) \toD \Normal(0,1) \for \lambda_*(F) \in (0,1) \textand B_{\gamma_*(F)}(\Cr)>0.
\end{equation} Finally, under the null hypothesis $H_0: \lambda_*(F)=0$, the variance of $\lambda_*(F_n)$ does not vanish. Consequently, \eqref{eq:limit2} also holds for $\lambda_*(F)=0$ and the test: \begin{align}\label{eq:efficient_test_known_background}
\Psi_\alpha(F_n)=I(T(F_n) > Z_{1-\alpha})
\where T(F_n) = \sqrt{n} \cdot \frac{\lambda_*(F_n)}{\tau_\lambda(F_n)}
\end{align} is an asymptotically valid $\alpha$-level test. \end{proof}}

\subsection{Expectation-Maximization for censored MLE}\label{sec:EM_censored_MLE}

Henceforth, assume that the parametric background is given by a truncated series \begin{equation}\label{eq:truncated_series}
b=b_\gamma=\sum_{k=1}^K \gamma_k \cdot \phi_k(x),  \where \gamma \in \mathbb{R}^K, 
\end{equation} and  recall the censored MLE estimator
\begin{align}
\label{eq:loss}
(\gamma_*(F_n),\lambda_*(F_n)) = \argmax_{\gamma,\lambda}\sum_{i=1}^n\ell(M_i,\lambda,\gamma) \st B_\gamma(\Omega) = 1 
\end{align} where the loss censors the signal region \begin{equation}
\ell(m)= I(m \in \Sr)\cdot\log\left((1-\lambda)\cdot\left(1-B_\gamma(\Cr)\right)+\lambda\right) +  I(m\in \Cr)\cdot \log\left((1-\lambda)\cdot b_\gamma(m)\right)
\end{equation} In the following let $n_{\Sr}$ denote the number of observations in the signal region $n_{\Sr}= \sum_{i=1}^nI(M_i\in \Sr)$, and $n_{\Cr} = n- n_{\Sr}$ denote the number of observations in the control region.

The above optimization doesn't have a closed-form solution, but the first-order optimality conditions indicate that the solution satisfies the following equalities \begin{gather}\label{eq:opt_cond}
\lambda_*(F_n)=1-\frac{F_n(\Cr)}{B_{\gamma_*(F_n)}(\Cr)} \comma B_{\gamma_*(F_n)}(\Omega)=1 \\
\textand \frac{B_{\gamma_*(F_n)}(\Cr)}{F_n(\Cr)} \cdot \phi_k(\Cr) =  \frac{1}{n}\sum_{i=1}^n I(x\in \Cr)\cdot \frac{\phi_k(M_i)}{b_{\gamma_*(F_n)}(M_i)} \ \for 1\leq k \leq K
\end{gather} In the following, we will see that we can use an Expectation-Maximization (EM) approximation \citep{dempsterMaximumLikelihoodIncomplete1977} to obtain an iterative algorithm that converges to to solutions that satisfy the above equations. In particular, we will use the methodology push-forward by \cite{approxem}, which approximates the origin loss function by an EM-like loss without using an statistical argument. 

We start by introducing a guess for the solution of $\gamma$ at the $q$th iteration
\begin{align}
\max_{\gamma: B_\gamma(\Omega)=1}\ \sum_{i=1}^n\ell(M_i)= \max_{\lambda,\gamma: B_\gamma(\Omega)=1}\ &n_{\Sr} \cdot   \log\left((1-\lambda)\cdot\left(1-B_\gamma(\Cr)\right)+\lambda\right) + n_{\Cr} \cdot   \log\left(1-\lambda\right) \\
&+  \sum_{i=1}^n I(M_i \in \Cr)\cdot \log\left(\sum_{k=1}^K \gamma_k \cdot \phi_k(X_i) \cdot \frac{\gamma_k^{(q)}\cdot b_{\gamma^{(q)}}(M_i)}{\gamma_k^{(q)}\cdot b_{\gamma^{(q)}}(M_i)}\right)
\end{align} Then, noting that $b_\gamma^{(q)}(M_i)=\sum_{k=1}^K \gamma_k^{(q)} \cdot \phi_k(M_i)$, we lower-bound the objective via Jensen's inequality \begin{align}
\max_{\lambda,\gamma: B_\gamma(\Omega)=1}\  &n_{\Sr}\cdot \log\left((1-\lambda)\cdot\left(1-B_\gamma(\Cr)\right)+\lambda\right) + n_{\Cr} \cdot   \log\left(1-\lambda\right) \\
&+  \sum_{k=1}^K\sum_{i=1}^n I(M_i \in \Cr)\cdot   \frac{\gamma_k^{(q)}\cdot\phi_k(M_i)}{b_{\gamma^{(q)}}(M_i)} \cdot \log\left(  \frac{\gamma_k \cdot b_{\gamma^{(q)}}(M_i)}{\gamma_k^{(q)}}\right)
\end{align} Finally, we remove the terms that are constant w.r.t. the optimization \begin{align}
\max_{\lambda,\gamma: B_\gamma(\Omega)=1}\ &n_{\Sr}\cdot \log\left((1-\lambda)\cdot\left(1-B_\gamma(\Cr)\right)+\lambda\right) + n_{\Cr} \cdot   \log\left(1-\lambda\right) \\
&+ \sum_{k=1}^K   a_k \cdot \gamma_k^{(q)} \cdot \log\left(\gamma_k\right)
\end{align} where \begin{equation}
a_k = \sum_{i=1}^n I(M_i \in \Cr)\cdot\frac{\phi_k(M_i)}{b_{\gamma^{(q)}}(M_i)}
\end{equation} Let $(\lambda^{(q+1)},\gamma^{(q+1)})$ denote the solution of the above system of equations. It follows that the solution must satisfy the following conditions \begin{equation}\label{eq:opt_cond2_a}
\lambda^{(q+1)}=1-\frac{F_n(\Cr)}{B_{\gamma^{(q+1)}}(\Cr)}\ \comma\ B_{\gamma^{(q+1)}}(\Omega)=1 
\end{equation} and \begin{align}\label{eq:opt_cond2_b}
\for 1 \leq k \leq K\quad \gamma_k^{(q+1)} &= \gamma_k^{(q)}\cdot \frac{1-B_{\gamma^{(q+1)}}(\Cr)+\frac{\lambda^{(q+1)}}{1-\lambda^{(q+1)}}}{n_{\Sr}} \cdot \frac{ a_k}{\phi_k(\Cr)}\\
&= \gamma_k^{(q)}\cdot \frac{B_{\gamma^{(q+1)}}(\Cr)}{F_n(\Cr)}  \cdot \frac{1}{\phi_k(\Cr)}\cdot \frac{a_k}{n}
\end{align} Solving the system of equations leads to a normalized D'Agostini iteration \citep{agostini1,agostini2}  for $\gamma^{(q+1)}$ that doesn't depend on $\lambda^{(q+1)}$ \begin{equation}\label{eq:continuous_dagostini_iteration}
\gamma_{k}^{(q+1)} = \frac{\tilde{\gamma}_k}{B_{\tilde{\gamma}}(\Omega)} \where \tilde{\gamma}_k = \gamma_k^{(q)}\cdot \frac{a_k}{\phi_k(C)} 
\end{equation}  \Revision{Finally, note that at the fix point $\gamma_\infty$ , \eqref{eq:opt_cond2_a} and \eqref{eq:opt_cond2_b} become \eqref{eq:opt_cond}. Since \eqref{eq:loss} is concave w.r.t. $\gamma$, the iteration converges to the unique maximizer. Ergo, $\gamma^{(\infty)}=\gamma_*(F_n)$, and consequently $\lambda^{(\infty)}=\lambda_*(F_n)$}.

\subsection{Discretized censored MLE}\label{sec:discretization}

An issue when implementing the test \eqref{eq:efficient_test_known_background} in the continuous case is to obtain a stable and fast implementation of the influence function \eqref{eq:influence_function}, which requires the inversion of the empirical Hessian. In order to sidestep this issue, we discretize the data and rely on the discrete delta method to construct or asymptotically valid test.  Henceforth, assume that the parametric background is given by a truncated series \Revision{\eqref{eq:truncated_series} and discretize the density model \eqref{eq:model}}. That is, we partition the control region into $L$ bins and take the whole signal region as one bin. Namely, define the disjoint sets \begin{equation}
\Cr = \cup_{l=1}^L \Cr_l \st  \Cr_i \cap \Cr_j = \emptyset \for i\not= j  
\end{equation} and their corresponding counts \begin{equation}
n_{\Sr} = n \cdot F_n(\Sr) \textand  n_l = n \cdot F_n(\Cr_l)\ \for 1 \leq l \leq L
\end{equation} The counts follow a Multinomial distribution \begin{align}\label{eq:multinomial_model}
(n_{1},\dots,n_{L},n_{\Sr})\ \sim\ \text{Multinomial}(n_C,p) \where p_l &= (1-\lambda) B_\gamma(\Cr_l) \for 1 \leq l \leq L\\
p_{\Sr} &= (1-\lambda) B_\gamma(\Sr_l) + \lambda \end{align} Let $\mathbb{F}_n(\Cr)$ denote the estimated probabilities in the control region \begin{equation}
\mathbb{F}_n(\Cr)= \begin{bmatrix}F_n(\Cr_1),\dots,F_n(\Cr_L)\end{bmatrix}^t 
\end{equation} The discretized censored maximum likelihood estimator is \begin{align}
(\gamma_*(F_n),\lambda_*(F_n)) = \argmax_{\gamma,\lambda} \ell(\mathbb{F}_n(\Cr),\lambda,\gamma) \st B_\gamma(\Omega) = 1 
\end{align} where \begin{align}
\ell(\mathbb{F}_n(\Cr),\lambda,\gamma)&= (1-F_n(\Cr)) \cdot \log\left((1-\lambda)\cdot (1-B_{\gamma}(\Cr)) + \lambda \right)\\
&+ F_n(\Cr) \cdot \log\left(1-\lambda\right)\\
&+ \sum_{l=1}^L\ F_n(\Cr_l) \cdot \log\left(B_{\gamma}(\Cr_l)\right)
\end{align} Following an analogous EM-approximation as in section \cref{sec:EM_censored_MLE}, we can derive the following D'Agostini iteration that converges to $\gamma_*(F_n)$ \begin{equation}
\gamma_{k}^{(q+1)} = \frac{\tilde{\gamma}_k}{B_{\tilde{\gamma}}(\Omega)} \where \tilde{\gamma}_k = \gamma_k^{(q)}\cdot \frac{ a_k}{\phi_k(C)} \textand a_k = \sum_{l=1}^L F_n(\Cr_l) \cdot \frac{\phi_k(\Cr_l)}{B_{\gamma^{(q)}}(\Cr_l)}
\end{equation} Note that the above iteration is the natural discretization of the continuous iteration \eqref{eq:continuous_dagostini_iteration}. Finally, $\lambda_*(F_n)$ can still be computed by the ratio between observed and expected counts $\lambda_*(F_n)=1-\frac{F_n(\Cr)}{B_{\gamma_*(F_n)}(\Cr)}$. 

Noting that $\lambda_*(F_n)$ is only a function of the counts in the control region, we can proceed analogously to \cref{sec:EM_censored_MLE} and obtain a central limit by the discrete delta method. Let $\mathbb{F}(\Cr)$ denote the population vector of probabilities in the control region\begin{equation}
\mathbb{F}(\Cr)=\begin{bmatrix}F(\Cr_1),\dots,F(\Cr_L)\end{bmatrix}
\end{equation} it follows by the central limit theorem that \begin{equation}
\sqrt{n}\cdot\left(\mathbb{F}_n(\Cr)-\mathbb{F}(\Cr)\right)\toD \Normal(0,D)
\end{equation} where $D \in \mathbb{R}^{(L-1)\times(L-1)}$ , $D_{l,l}=F(\Cr_l)\cdot (1-F(\Cr_l))$ and $D_{l,j}=-F(\Cr_l)\cdot F(\Cr_j)$ for $j\not= l$. Thus, by the discrete delta method and Slutsky's theorem, we have that \begin{equation}
\sqrt{n} \cdot \frac{\lambda_*(F_n)-\lambda_*(F)}{\sqrt{g(F_n)^t\cdot D_n \cdot g(F_n)}}  \toD \Normal(0,1) 
\end{equation} where $g(F_n)=\nabla_{\mathbb{F}_n(\Cr)}\lambda_*(F_n)$ is the gradient with respect to the empirical probabilities, and $D_n$ is the empirical covariance matrix of the probabilities $D_n \in \mathbb{R}^{(L-1)\times(L-1)}$, $[D_n]_{l,l}=F_n(C_l)\cdot (1-F_n(C_l))$ and $[D_n]_{l,j}=-F_n(C_l)\cdot F_n(C_j)$ for $j\not= l$.

Finally, analogously to the test in the continuous case \eqref{eq:efficient_test_known_background}, we define our asymptotically valid $\alpha$-level test to be \begin{align}
\Psi_\alpha(F_n)&=I(T(F_n) > Z_{1-\alpha})\\ \where T(F_n) &= \sqrt{n} \cdot \frac{\lambda_*(F_n)}{\sqrt{g(F_n)^t\cdot D_n \cdot g(F_n)}}
\end{align}

\subsection{Equivalence between censored and conditional MLE}\label{sec:conditional_MLE}

Consider the mixture model \eqref{eq:model} where the background is known to belong to the set $\mathcal{B}$. Since the signal vanishes in the control region, it follows that the conditional distribution of the mixture on the control region depends only on the background distribution \begin{equation} M | M \in C \sim \frac{f(m)}{F(\Cr)} \cdot I(m \in \Cr) = \frac{b(m)}{B(\Cr)}\cdot I(m \in \Cr)\label{eq:conditional}
\end{equation} A natural idea is to estimate the conditional background distribution on the control region, and then extend it to the signal region, that is, to force the estimated background to integrate to one on the whole domain. This is justified if there is a unique way of extending the conditional background, and consequently the conditional measure identifies the measure on the whole domain. For instance, this is the case when the true background distribution is known to be a polynomial.

\begin{proposition}[\textbf{Polynomial densities that agree on $\Cr$ must agree on $\Omega$}]\label{lemma:poly_background}
    Let $\mathcal{B}$ be some set of densities supported on $\Omega$ such that $b \in \mathcal{B}$ and for any $\tilde{b} \in \mathcal{B}$ it holds that \begin{align}
	(0)\	\tilde{b} \text{ is a polynomial, } \tilde{b}\geq0 \text{ and } \tilde{B}(\Omega)=1
	\end{align} Furthermore, consider the function $d(f,g)$ in $\mathcal{B}\times\mathcal{B}$ such that \begin{align}
		(1)\ d(f,f)=0 \quad
		(2)\ d(f,g)=0 \implies f=g\ \text{a.e.} \quad
		(3)\ d(f,g)\geq0\quad \forall f,g\in \mathcal{B}
	\end{align} then\begin{equation}\label{eq:equivalence}
		b = \argmin_{\tilde{b} \in \mathcal{B}} d(\ b,\ \tilde{b}\ ) = \argmin_{\tilde{b} \in \mathcal{B}} d(\ \frac{b}{B(C)} \cdot I_{\Cr}\ ,\ \frac{\tilde{b}}{\tilde{B}(C)} \cdot I_{\Cr}\ )
	\end{equation} 
\end{proposition} The lemma allows us to rewrite the signal strength as a function of $F$ by exploiting the equivalence in equation \eqref{eq:equivalence}. Namely, \begin{equation}\label{eq:lambda_F}
\lambda^*(F) = 1 - \frac{F(\Cr)}{B^*_{F}(\Cr)} \where b_F^* = \argmin_{\tilde{b} \in \mathcal{B}} d(\ \frac{f}{F(\Cr)} \cdot I_{\Cr}\ ,\ \frac{\tilde{b}}{\tilde{B}(\Cr)} \cdot I_{\Cr}\ )
\end{equation} Furthermore, consider using the Kullback-Leibler divergence \citep{kullbackInformationSufficiency1951} for $d$, then we obtain the conditional maximum likelihood estimator \begin{equation}\label{eq:lambda_Fn}
\lambda^*(F_n) = 1 - \frac{F_n(\Cr)}{B^*_{F_n}(\Cr)}
\end{equation} where \begin{equation}\label{eq:conditional_MLE}
b^*_{F_n} = \argmax_{\tilde{b}\in\mathcal{B}} \sum_{i=1}^n I(M_i \in \Cr)\cdot\log\left(\frac{\tilde{b}(M_i)}{\tilde{B}(\Cr)}\right) \st \tilde{B}(\Omega)=1
\end{equation} Note that the condition $\tilde{B}(\Omega)=1$ guarantees the valid extension of the conditional density to the whole domain. 

Finally, regardless of the uniqueness of the extended background, both the conditional and censored MLE coincide from an algorithmic point of view. Recall the censored maximum likelihood estimator is given by\begin{align}\label{eq:censored_MLE_opt}
(\lambda^*(F_n),b_{F_n}^*) = \argmax_{\tilde{\lambda},\tilde{b}\in\mathcal{B} : \tilde{B}(\Omega)=1} &n_{\Sr}\cdot \left((1-\tilde{\lambda})\cdot (1-\tilde{B}(\Cr))+ \tilde{\lambda}\right)\\
&+\sum_{i=1}^n I(M_i \in \Cr)\cdot\log\left((1-\tilde{\lambda})\cdot\tilde{b}(M_i)\right)
\end{align} The first order optimality condition for $\lambda_*$ is given by \eqref{eq:lambda_Fn}. Plugging the result back into the optimization \eqref{eq:censored_MLE_opt} gives us the conditional MLE objective \eqref{eq:conditional_MLE}.

\begin{proof}[Proof of \cref{lemma:poly_background}]
Since $b \in \mathcal{B}$, by (1) and (3), it follows that the minimum must be achieved \begin{equation}
    d(\ \frac{b}{B(\Cr)} \cdot I_{\Cr} , \frac{b^*}{B^*(\Cr)} \cdot I_{\Cr} ) = 0
\end{equation} by (2), we know that the function must agree a.e. on the control region \begin{equation}
    \frac{b}{B(\Cr)} = \frac{b^*}{B^*(\Cr)} \text{ a.e. } x \in \Cr
\end{equation}by (0), we can extend the previous equation to all the domain \begin{equation}
    \frac{b}{B(\Cr)} = \frac{b^*}{B^*(\Cr)} \text{ a.e. } x \in \Omega
\end{equation} Integrating to both sides over $\Omega$, and using the fact that $\tilde{B}(\Omega)=1$ for all $\tilde{b} \in \mathcal{B}$, we get\begin{equation}
B(\Cr) = B^*(\Cr)
\end{equation} and consequently $b = b^* \text{ a.e. } x \in \Omega$
\end{proof}

\section{Simulation details}

\subsection{Code}

\if1\blind
{
The code to reproduce the experiments can be accessed at \begin{center}
\url{https://anonymous.4open.science/r/CDOT-53DC}
\end{center}
} 
\fi

\if0\blind
{
The code to reproduce the experiments can be accessed at \begin{center}
\url{https://github.com/lkania/cdot}
\end{center}
} 
\fi

\subsection{Bernstein basis}\label{sec:bernstein_basis}

In our numerical studies, we use the $K$th order Bernstein basis \begin{equation}\label{eq:bernstein_basis}
\phi_k(x) = \binom{K}{k} \cdot x^{k}\cdot (1-x)^{K-k} \for 0\leq x\leq 1
\end{equation} while restricting the parameters $\gamma$ to be non-negative $\gamma_k \geq 0\ \forall k$ and satisfy $\sum_{k=0}^K \gamma_k=K+1$. These restrictions guarantee that $b_\gamma$ is a density function.
Polynomials in Bernstein form have been routinely used in high-energy physics for background modeling  \citep{bernstein_4_1,bernstein_4_2} since they uniformly approximate any continuous function on $\Omega$ at a rate of $O(1/\sqrt{K})$ \citep{lorentz2013} and do not suffer from boundary bias \citep{ghosal2001}. Furthermore, in order to accelerate computations, we discretize the data using bins before estimating the censored MLE, see Appendix \ref{sec:discretization}.

\subsection{Calibration of signal-enriched test}\label{sec:experiments_appx}

In all experiments, the data is split into training, validation, and test datasets. The training background and signal samples are used to fit a supervised probabilistic classifier. Let $h: \mathcal{X} \to [0,1]$ denote the trained classifier. Conditioned on the training dataset, the classifier is a deterministic function that maps each observation to a score, with higher scores indicating that an event is more likely to be a signal event. 

The validation dataset is used for three purposes. First, if performing decorrelation, the validation data is used to train the decorrelation algorithm by finding the optimal transport map, $T_M(h)$, on the validation background data. Note that if decorrelation is not used, the transport map is defined as the identity map $T_M(h) = h$. Second, the validation dataset and the transformed classifier output are used to find the cut-off point that filters $t \%$ of the background samples. To do this, we compute the score $T_M(h)$ of all background observations and find their $t$-th quantile, denoted by $q(t)$. Thus, the indicator function $I\left(T_M(h(X)) \geq q(t)\right)$ filters $t\%$ of the background observations in the validation dataset. Third, the validation dataset is used to calibrate the test. \Revision{Namely, it is used to specify the signal region $\Sr$ in \eqref{eq:model} and the degree $K$ of the basis \eqref{eq:bernstein_basis}}. The signal region is fixed to be the interval between the 10\% and 90\% quantiles of the empirical signal distribution in the validation dataset. That is, we allow for $20\%$ of signal contamination in the control region. Inspecting the empirical distribution of the background, we find that the ratio between contamination and background in the control region, $\epsilon/B(\Cr)$ in \eqref{eq:contamination}, ranges between 0.21-0.23 for the high-mass resonance experiments and 0.29-0.3 for the decaying high-pT W-boson experiments. Thus, on average, we expect to underestimate the true signal strength by roughly $25\%$. 

To fix the degree of the basis, we sub-sample $N$ datasets of $n$ background observations from the validation dataset. That is, these datasets follow the null distribution. Given $X_1,\dots,X_n \sim f$, let  $\Psi_\alpha^{(K)}$ denote a test defined by \eqref{eq:censored_MLE_test}, where the parametric background (\eqref{eq:truncated_series}) is given by the Bernstein polynomial basis (\eqref{eq:bernstein_basis}). Then for $K \in \{5,10,15,20,25,30,35,40\}$, we compute the empirical type-I error probability and choose $K$ such that the corresponding test is closest to the desired type-I error rate. Concretely, we use the test  $\Psi_\alpha^{(K_*)}$, where \begin{equation}\label{eq:test_selection}
K_* = \argmin_{K} \left|\ \alpha - 
\frac{1}{N} \sum_{j=1}^{N}\ \Psi_{\alpha}^{(K)}\left(F_n^{(j)}\right)\ \right|
\end{equation} and $F_n^{(j)}$ is the empirical distribution of $\left\{X_i^{(j)}\right\}_{i=1}^n$. For all forthcoming experiments, we set $\alpha=0.05$, $n=20000$, and $N=500$. Given the above test, a signal-enriched test can be built by first filtering observations using the trained classifier and then applying the test. Namely, $
\Psi_\alpha^{(K_*)}(X,t) = \Psi_\alpha^{(K_*)}(\tilde{X}) \where X_i \in \tilde{X}\ \textif I(T_M(h(X))\geq q(t))$. Thus, conditioned on the training and validation dataset, $\Psi_\alpha^{(K_*)}(\cdot,t)$ is a deterministic function. Note that the signal-enriched test maintains approximate validity due to the decorrelation algorithm preserving the shape of the background distribution. 

Using the test dataset, we study the power of the signal-enriched test and the performance of the CDOT decorrelation algorithm. To analyze the power of the signal detection test, we proceed as follows: given a signal strength $\lambda$, we sub-sample $N$ datasets of $n$ observations from the test dataset, such that $\lambda\%$ of those observations correspond to signal events. Then, for each dataset, we check if the test rejects the null hypothesis $H_0:\lambda=0$ and compute the empirical probability of rejecting the null hypothesis across datasets. Note that for $\lambda=0$, that probability is the empirical type-I error, while for $\lambda \in \{0.01,0.02,0.05\}$  is the empirical power. In all cases, we report the empirical results with their corresponding Clopper--Pearson confidence intervals. Finally, to understand the utility of the decorrelation algorithm, we study the power of the signal-enriched test both when using a non-decorrelated and a decorrelated classifier to filter the observations. 

\subsection{W-tagging dataset}\label{sec:wtagging_appx}

\begin{figure}[H]
\begin{center}
\includegraphics[width=\textwidth]{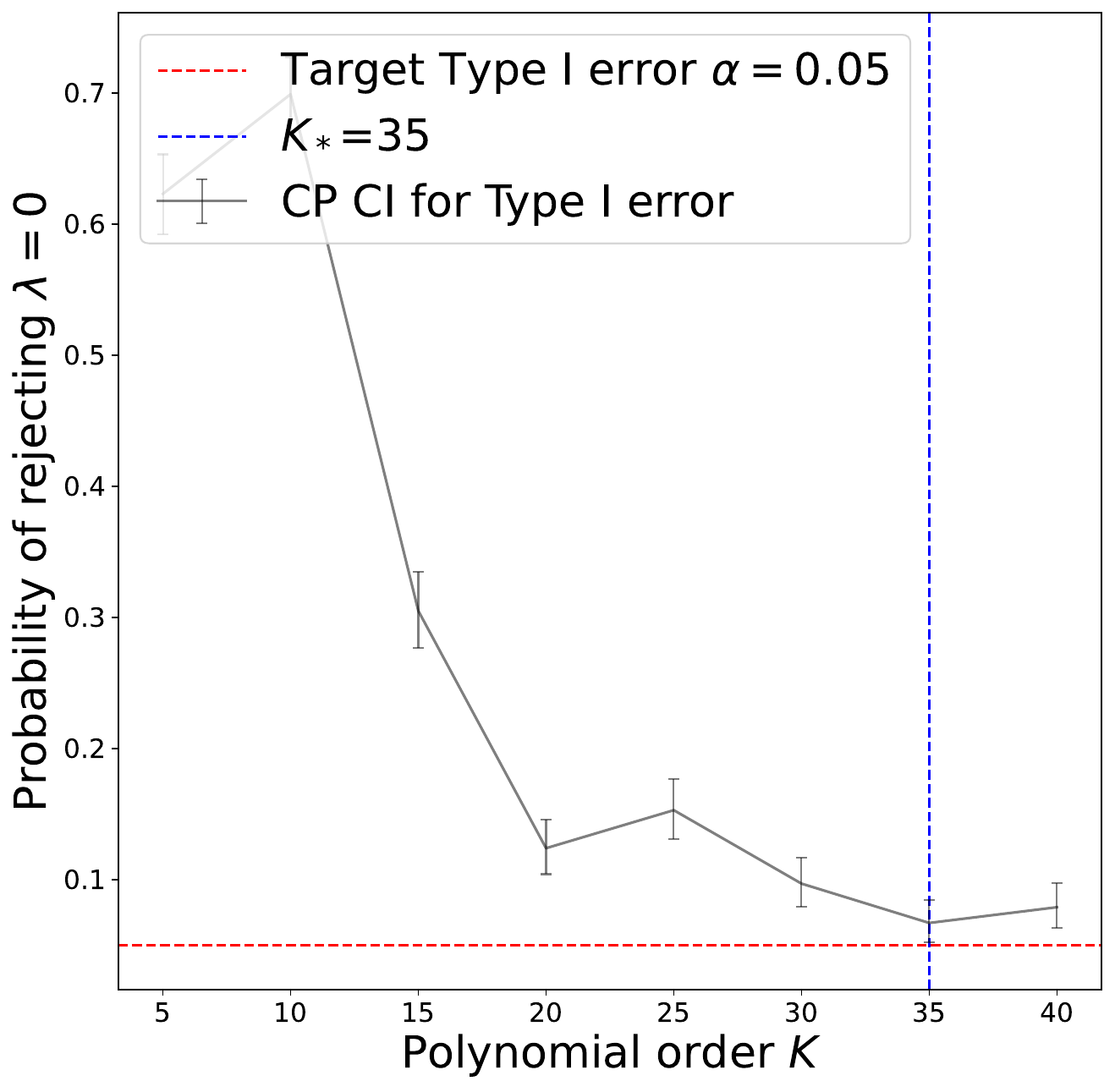}
\end{center}
\caption{Model selection using simulated datasets composed only of background observations from the W-tagging validation dataset. The target type-I error rate is $0.05$. However, none of the asymptotic tests achieve it. The Bernstein basis of order $35$ is selected, which achieves an average type-I error rate between 0.05 and 0.1.}
\label{img:wtagging_test_selection}
\end{figure}

\begin{figure}[H]
\begin{center}
\includegraphics[width=\textwidth]{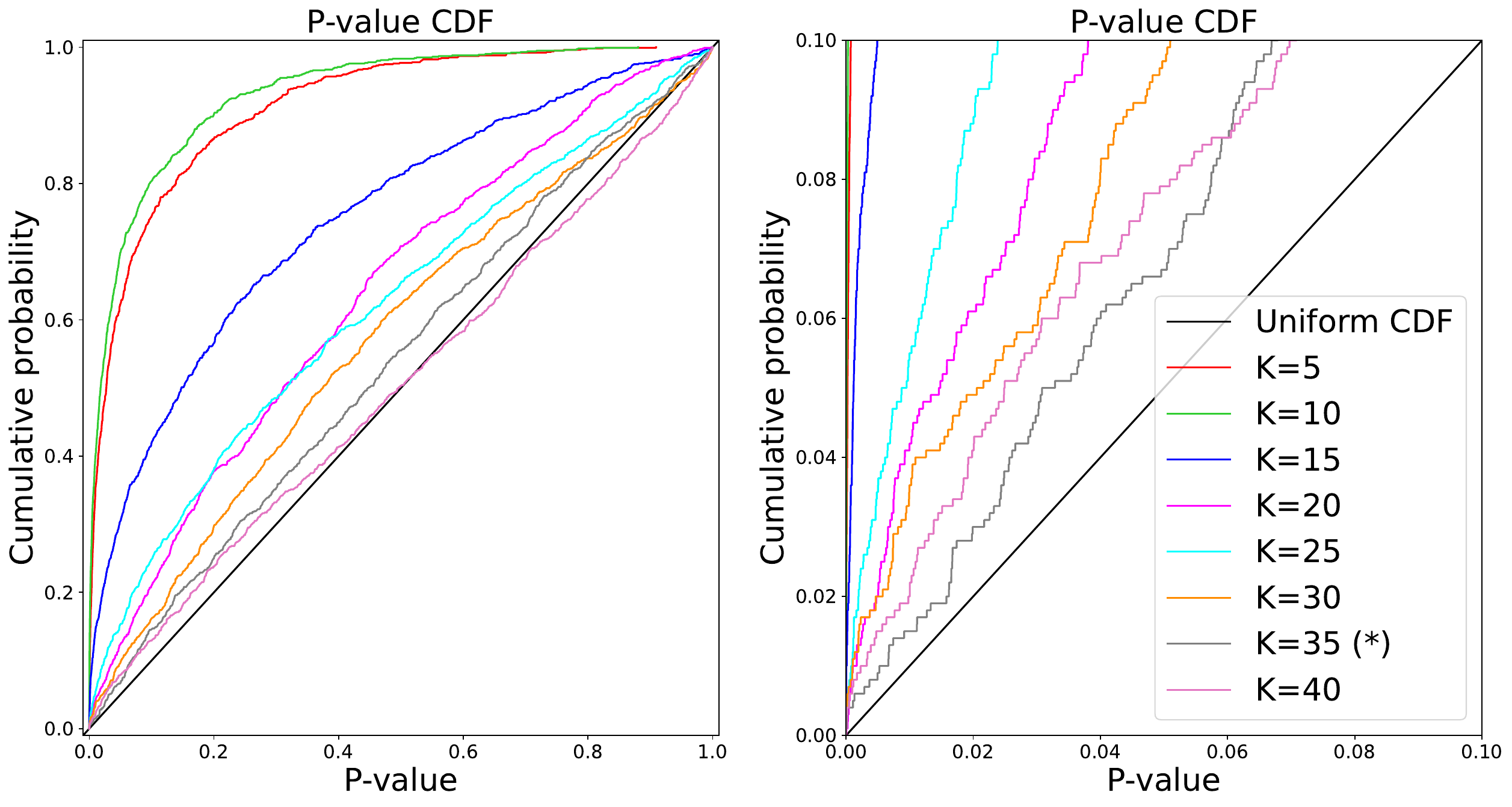}
\end{center}
\caption{CDFs of the empirical p-value distributions corresponding to the different tests considered in Figure \ref{img:wtagging_test_selection}.}
\end{figure}

\begin{figure}[H]
\begin{center}
\includegraphics[width=\textwidth]{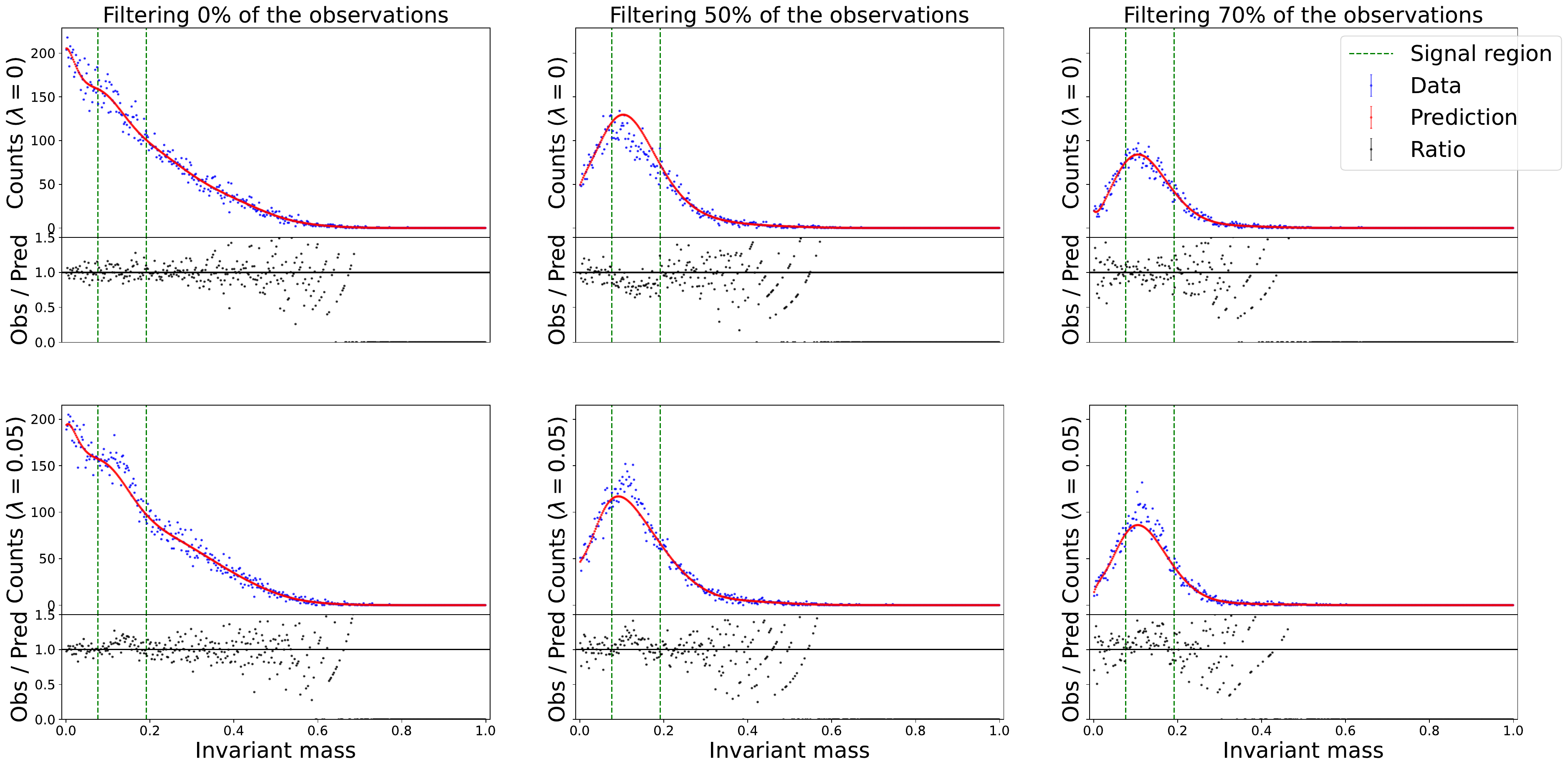}
\end{center}
\caption{Single background estimate for the W-tagging test dataset as the number of observations filtered by the \textbf{non-decorrelated} classifier is increased. The first row corresponds to the null hypothesis, i.e., no signal. In the second row, $5\%$ of the data comes from the signal. Note that the shape of the background distribution does not change. }
\end{figure}

\begin{figure}[H]
\begin{center}
\includegraphics[width=\textwidth]{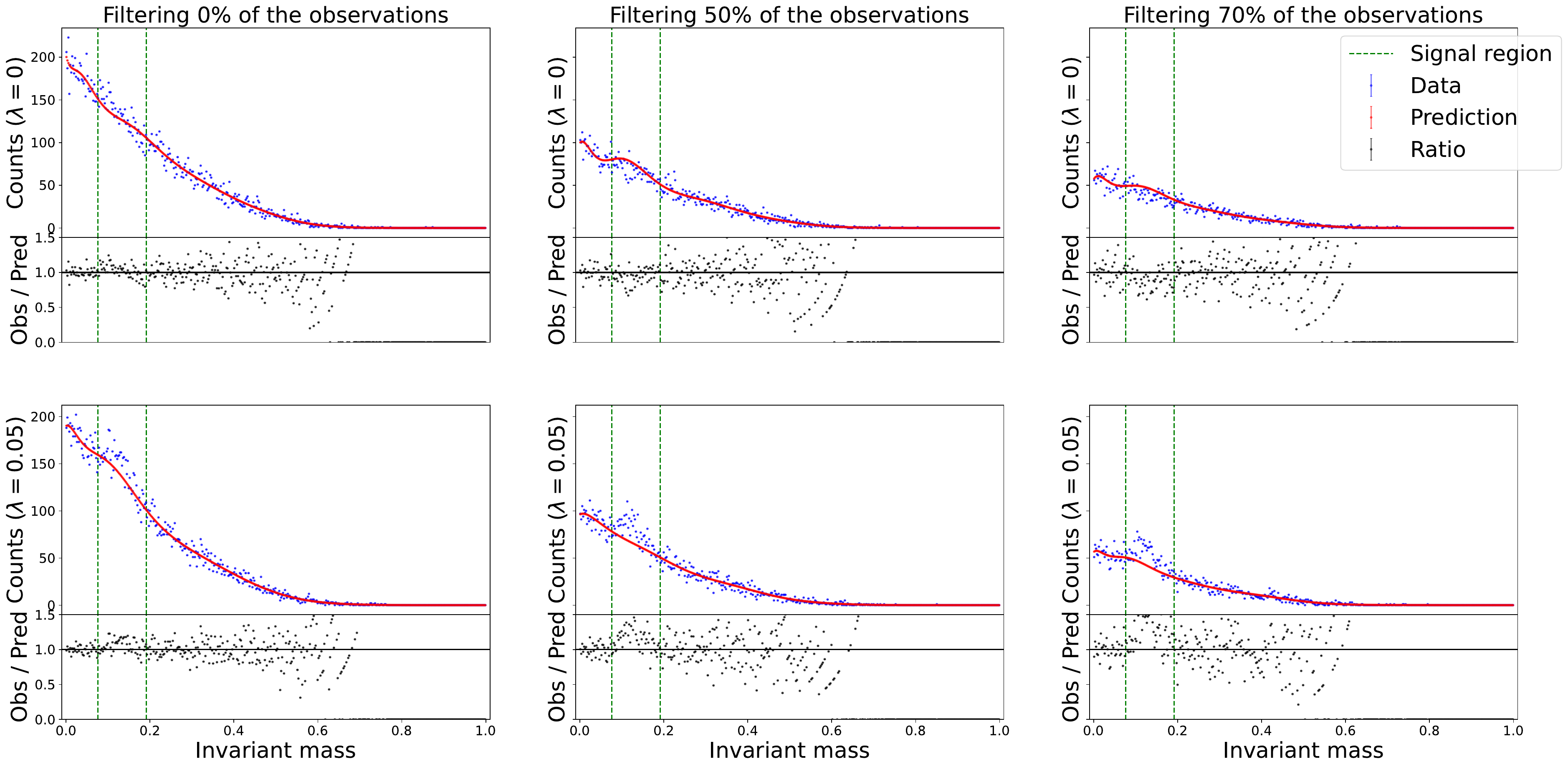}
\end{center}
\caption{Single background estimate for the W-tagging test dataset as the number of observations filtered by the \textbf{decorrelated} classifier is increased. The first row corresponds to the null hypothesis, i.e., no signal. In the second row, $5\%$ of the data comes from the signal. Note that the shape of the background distribution does not change. }
\end{figure}

\subsection{3b and 4b datasets}\label{sec:3b_4b_appx}

\begin{figure}[H]
\begin{center}
\includegraphics[width=0.5\textwidth]{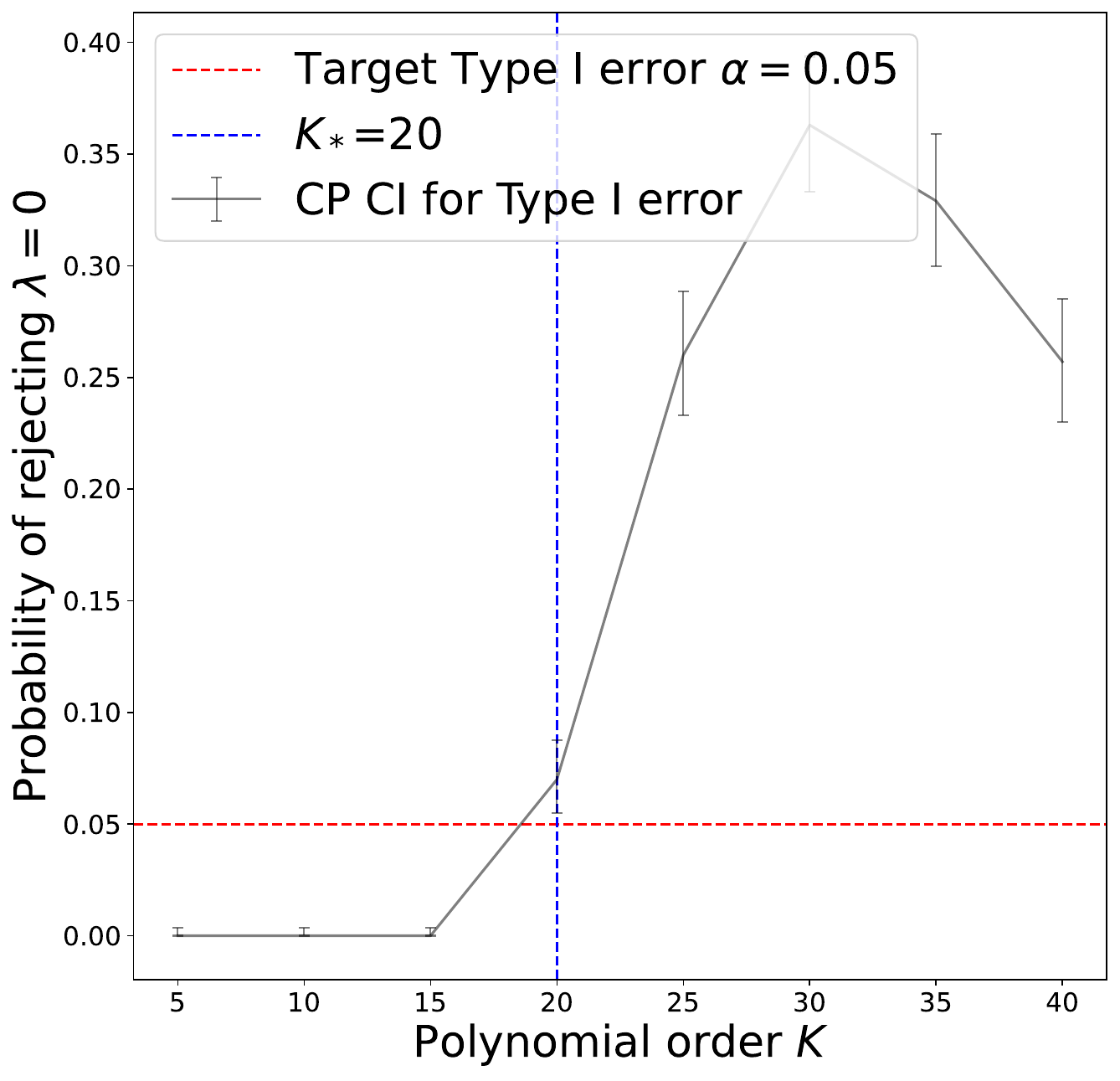}
\end{center}
\caption{Model selection using simulated datasets composed only of background observations from the 3b validation dataset. The target type I error is $0.05$. The Bernstein basis of order $20$ is selected.}
\label{fig:3b_model_selection}
\end{figure}

\begin{figure}[H]
\begin{center}
\includegraphics[width=\textwidth]{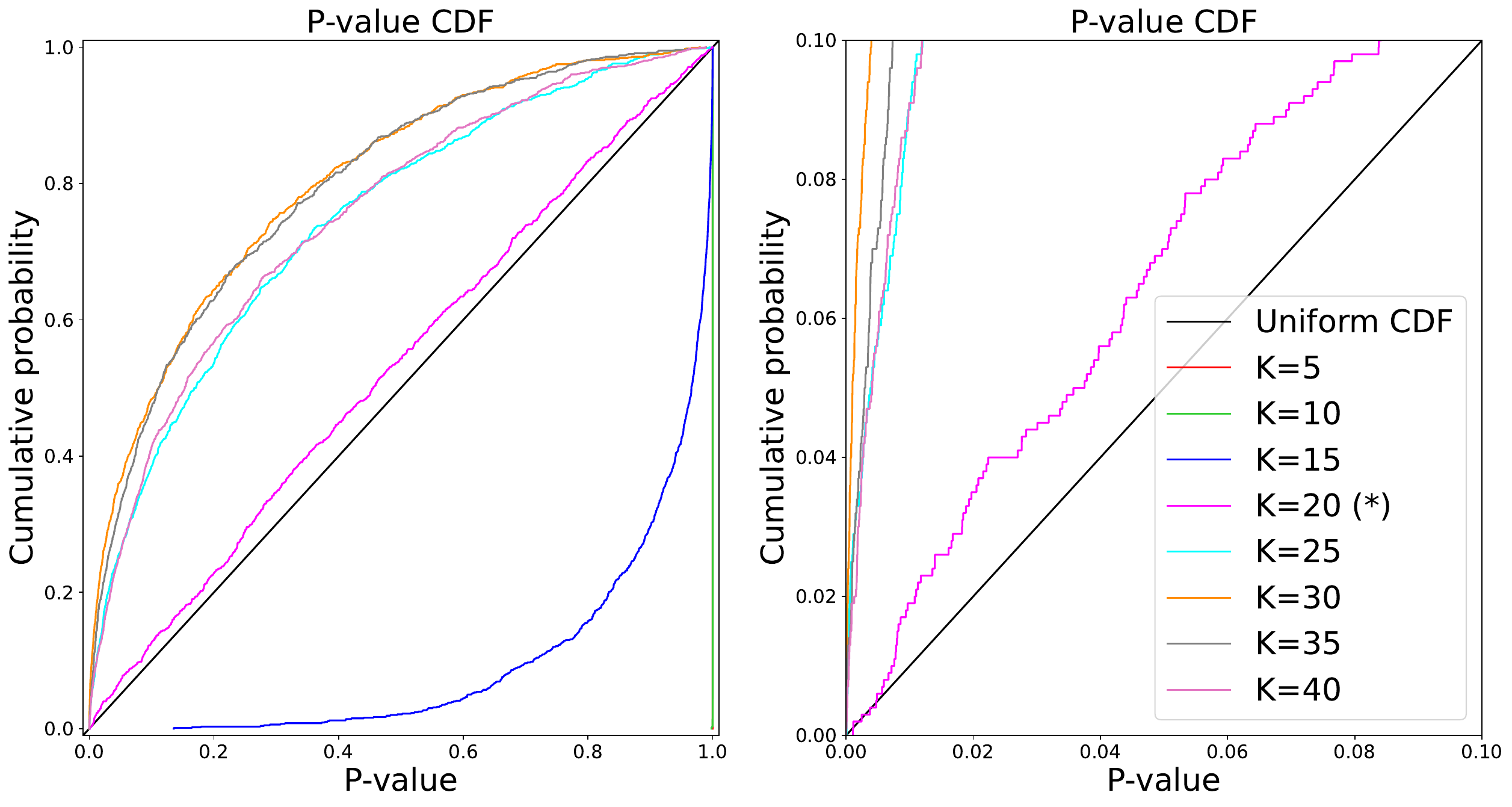}
\end{center}
\caption{CDFs of the empirical p-value distributions corresponding to the different tests considered in Figure \ref{fig:3b_model_selection}.}
\end{figure}

\begin{figure}[H]
\begin{center}
\includegraphics[width=\textwidth]{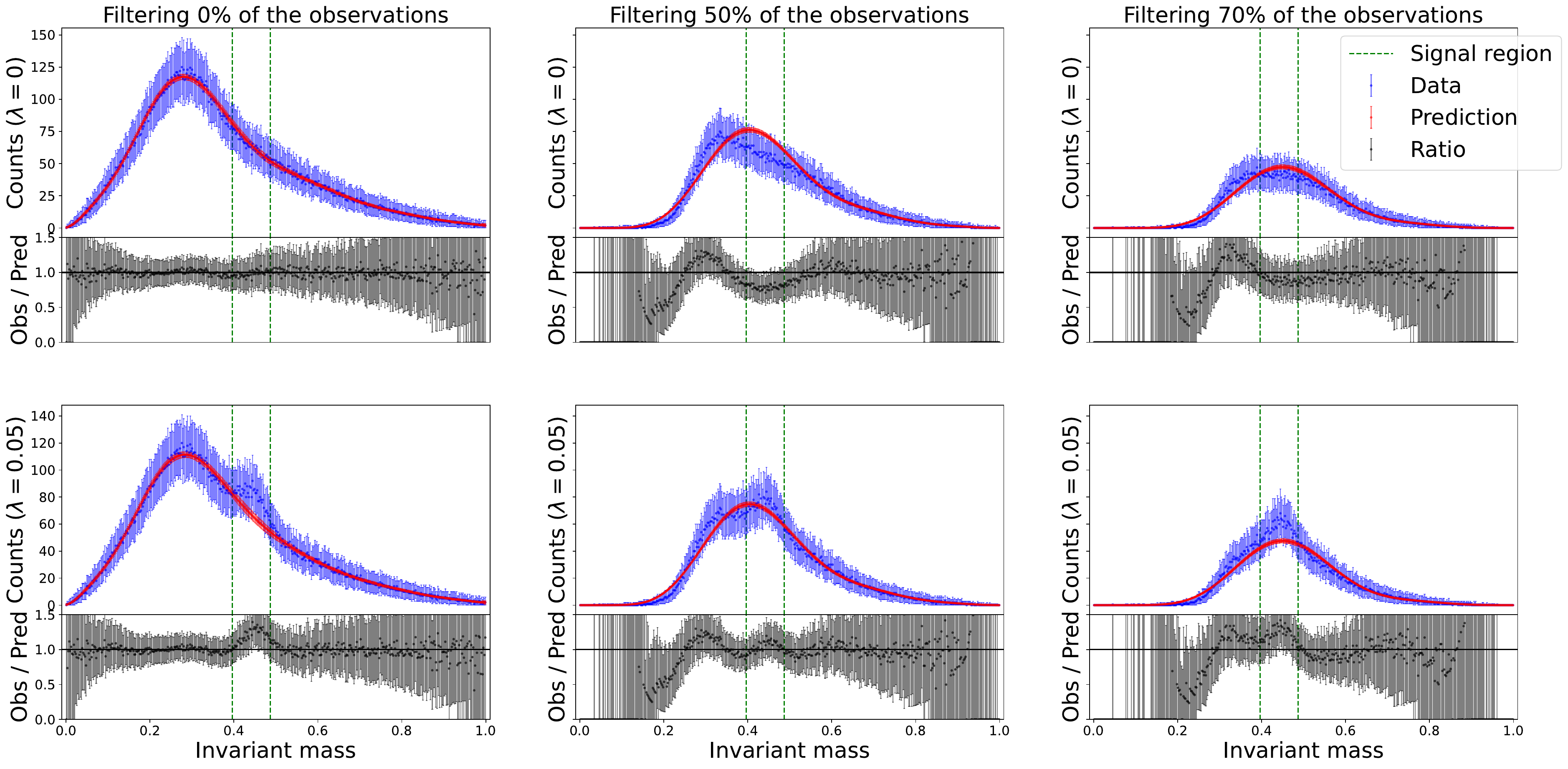}
\end{center}
\caption{Background estimates for the 3b test dataset as the number of observations filtered by the \textbf{non-decorrelated} classifier is increased. The first row corresponds to the null hypothesis, i.e., no signal is present. In the second row, $5\%$ of the data comes from the signal distribution. Note that the shape of the background distribution changes, producing a bump in the signal region. The intervals are 95\% variability intervals based on 1000 simulations, and their midpoint is the median of the simulations.}
\end{figure}

\begin{figure}[H]
\begin{center}
\includegraphics[width=\textwidth]{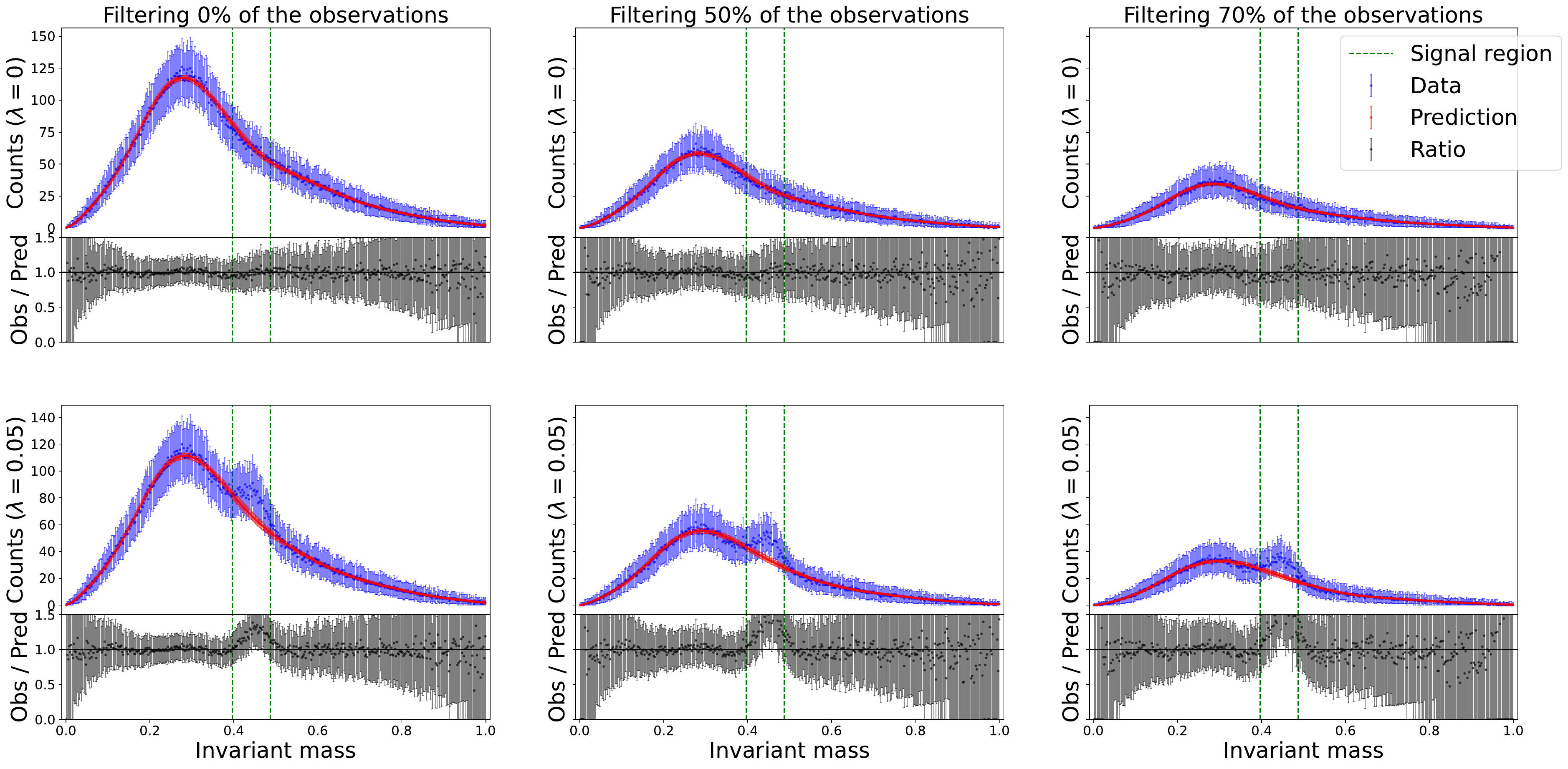}
\end{center}
\caption{Background estimates for the 3b test dataset as the number of observations filtered by the \textbf{decorrelated} classifier is increased. The first row corresponds to the null hypothesis, i.e., no signal. In the second row, $5\%$ of the data comes from the signal. Note that the shape of the background distribution does not change. The intervals are 95\% variability intervals based on 1000 simulations, and their midpoint is the median of the simulations.}
\end{figure}

% \paragraph{Simulations 3b Dataset}

\begin{figure}[H]
\begin{center}
\includegraphics[width=\textwidth]{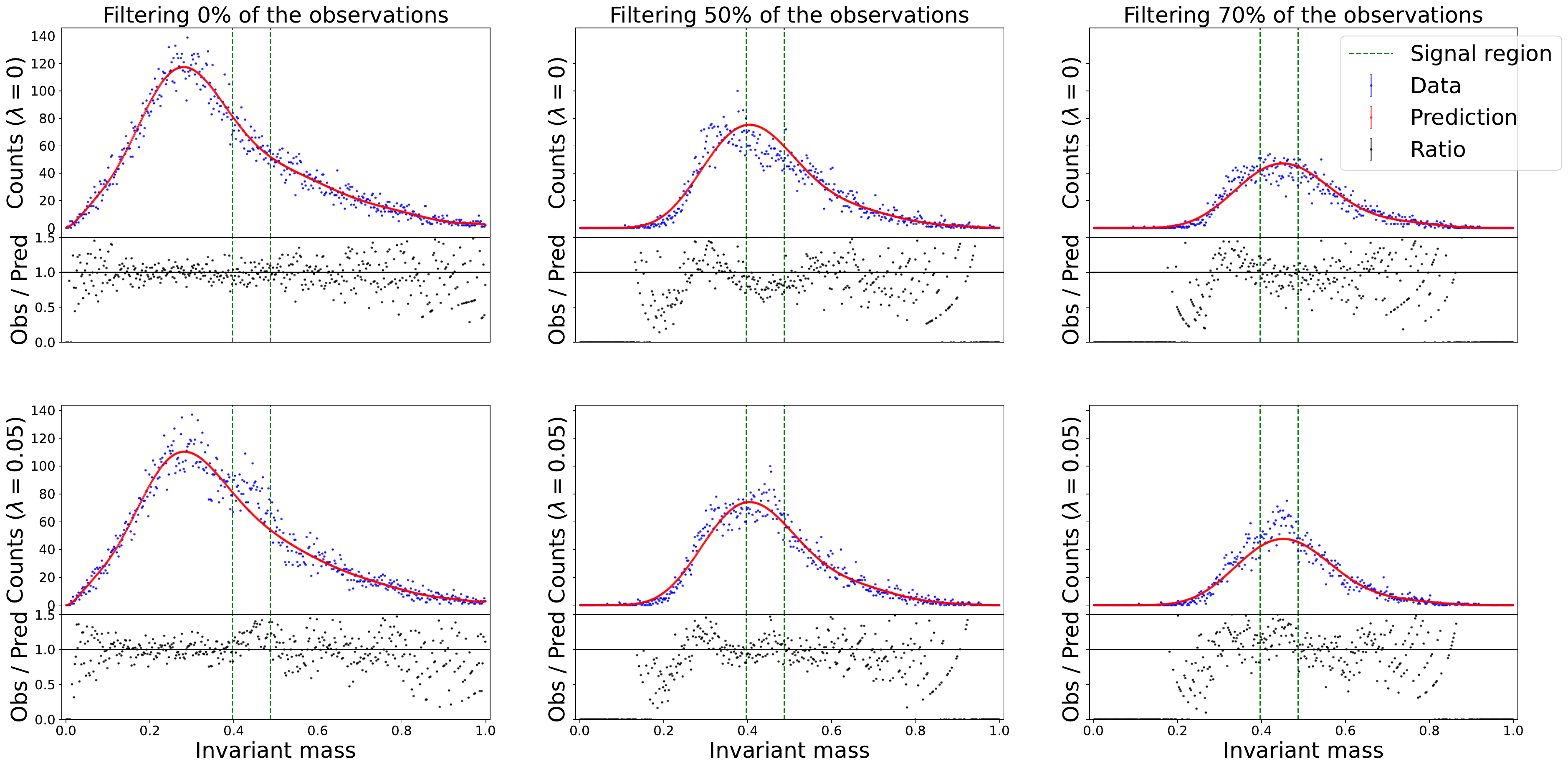}
\end{center}
\caption{Single background estimate for the 3b test dataset as the number of observations filtered by the \textbf{non-decorrelated} classifier is increased. The first row corresponds to the null hypothesis, i.e., no signal. In the second row, $5\%$ of the data comes from the signal. Note that the shape of the background distribution does not change. }
\end{figure}

\begin{figure}[H]
\begin{center}
\includegraphics[width=\textwidth]{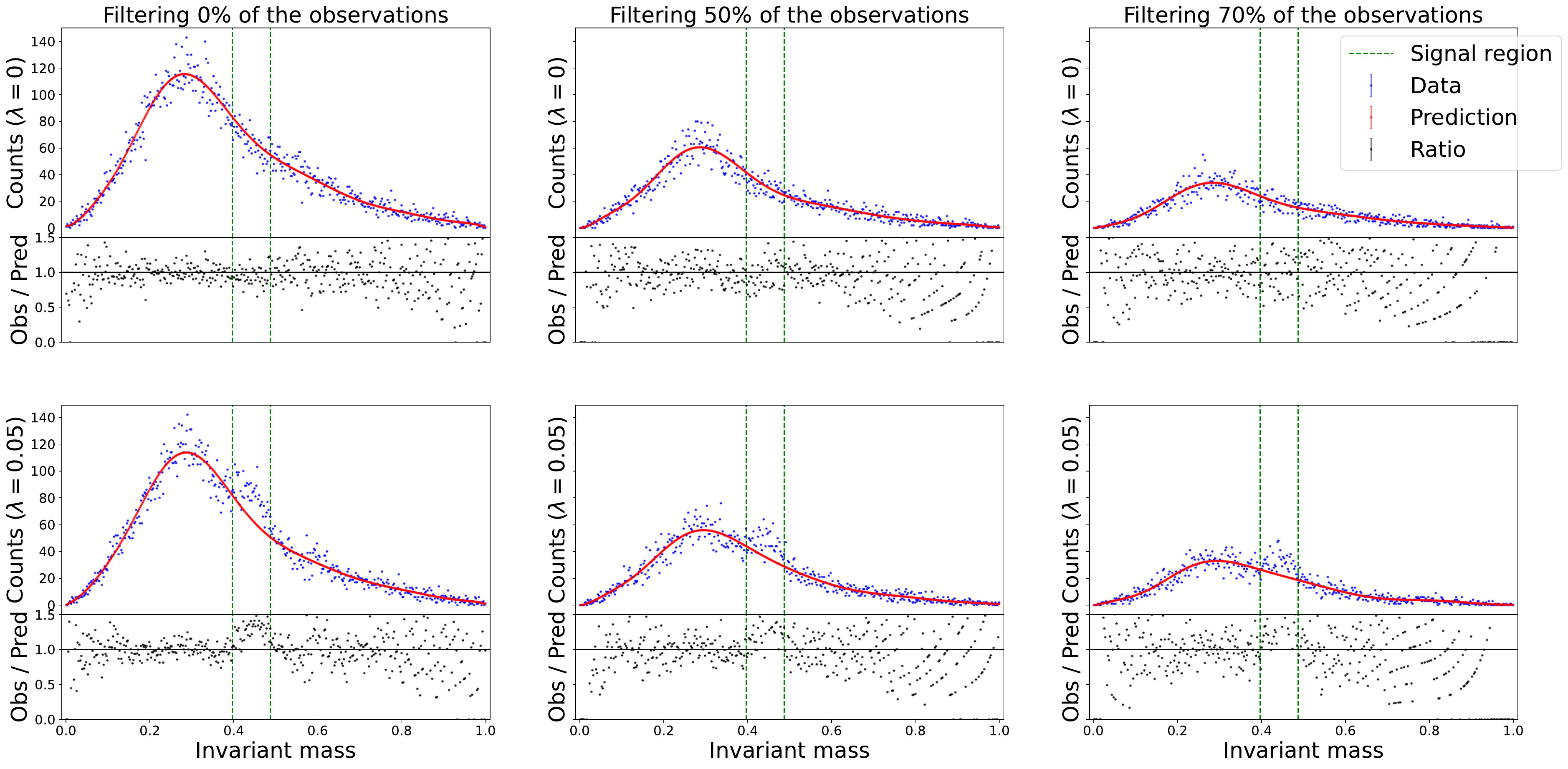}
\end{center}
\caption{Single background estimate for the 3b test dataset as the number of observations filtered by the \textbf{decorrelated} classifier is increased. The first row corresponds to the null hypothesis, i.e., no signal. In the second row, $5\%$ of the data comes from the signal. Note that the shape of the background distribution does not change. }
\end{figure}

\begin{figure}[H]
\begin{center}
\includegraphics[width=\textwidth]{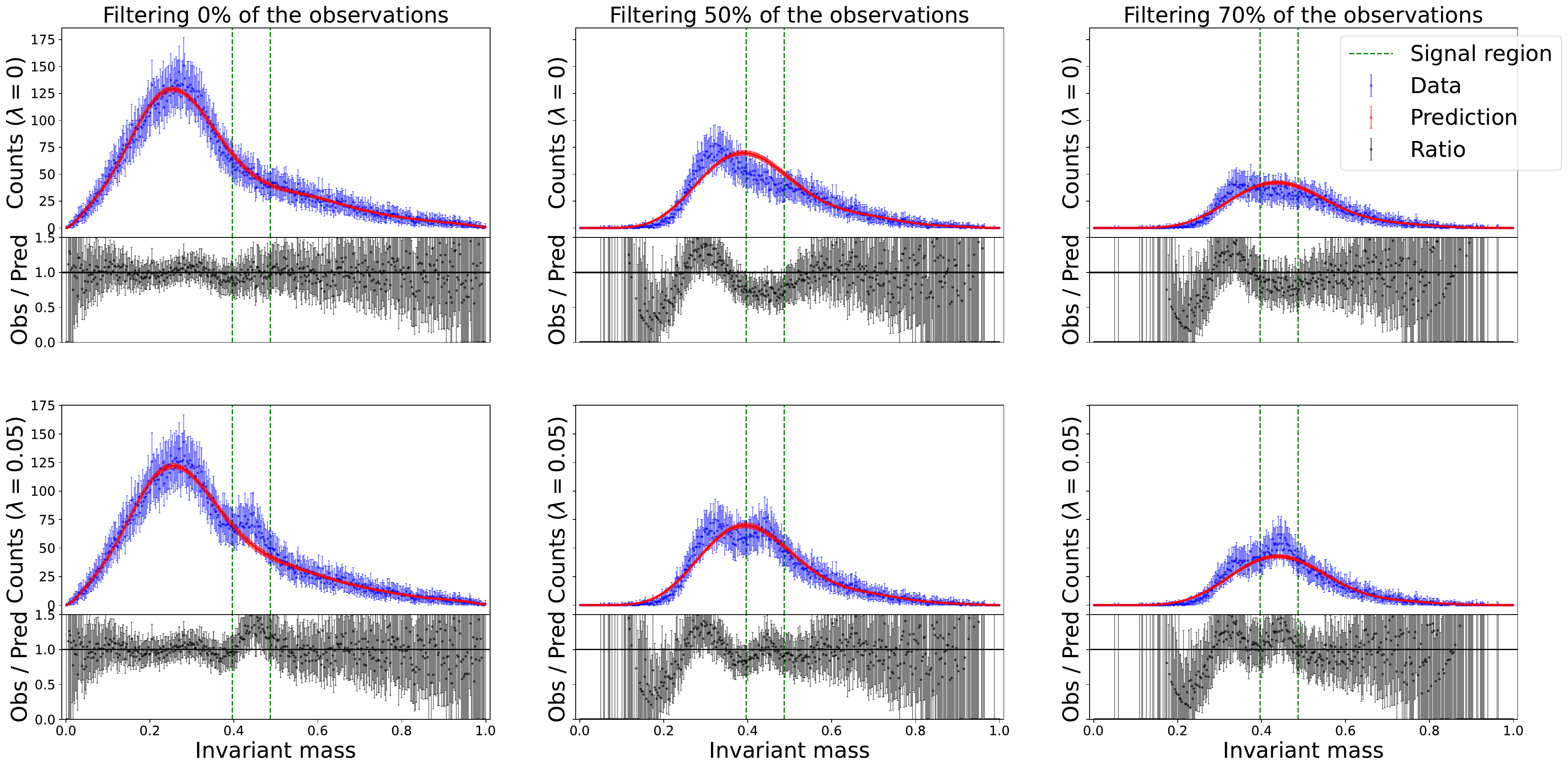}
\end{center}
\caption{Background estimates for the 4b test dataset as the number of observations filtered by the \textbf{non-decorrelated} classifier is increased. The first row corresponds to the null hypothesis, i.e., no signal is present. In the second row, $5\%$ of the data comes from the signal distribution. Note that the shape of the background distribution changes, producing a bump in the signal region. The intervals are 95\% variability intervals based on 1000 simulations, and their midpoint is the median of the simulations.}
\end{figure}

\begin{figure}[H]
\begin{center}
\includegraphics[width=\textwidth]{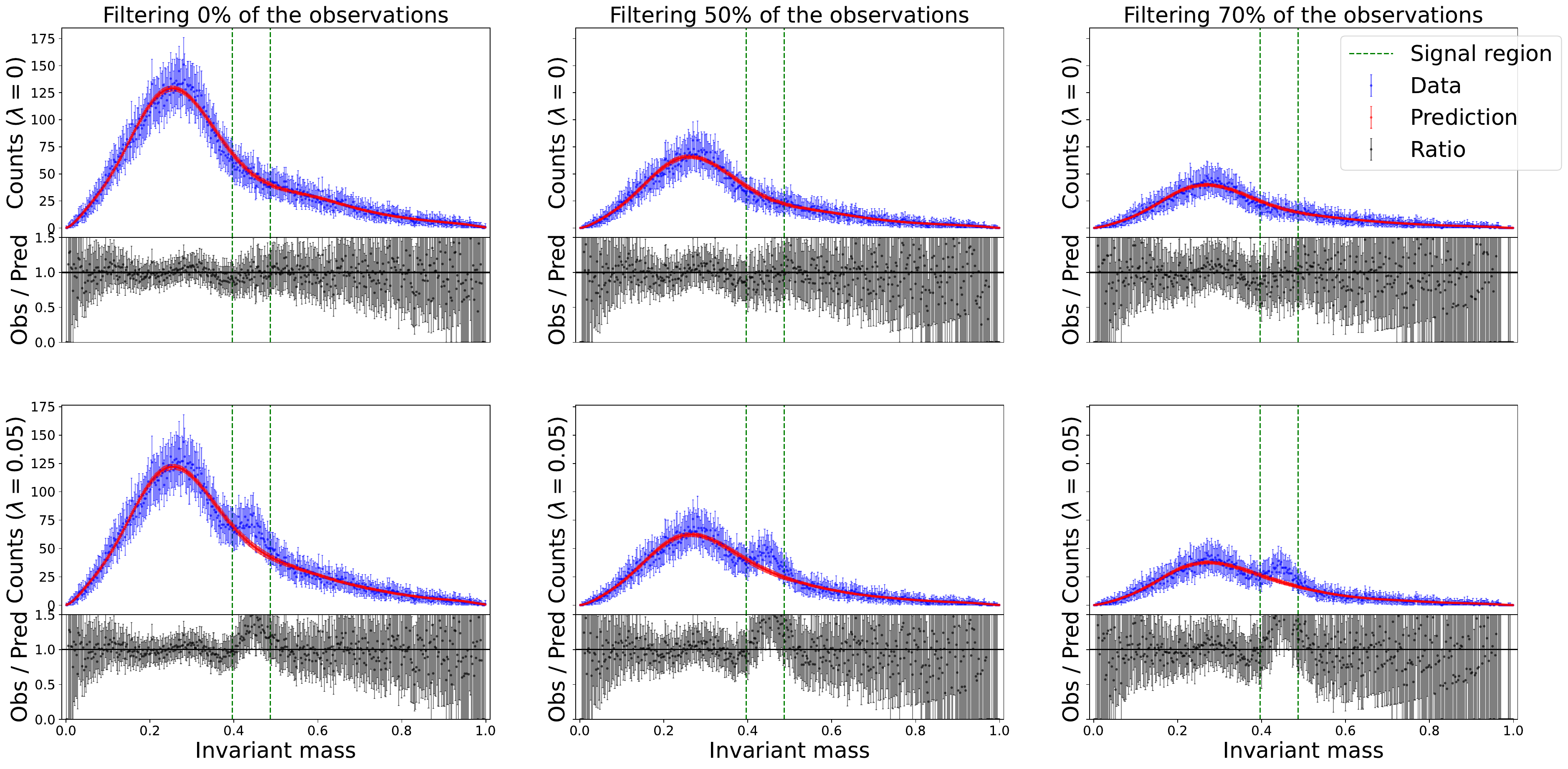}
\end{center}
\caption{Background estimates for the 4b test dataset as the number of observations filtered by the \textbf{decorrelated} classifier is increased. The first row corresponds to the null hypothesis, i.e., no signal. In the second row, $5\%$ of the data comes from the signal. Note that the shape of the background distribution does not change. The intervals are 95\% variability intervals based on 1000 simulations, and their midpoint is the median of the simulations.}
\end{figure}

\begin{figure}[H]
\begin{center}
\includegraphics[width=\textwidth]{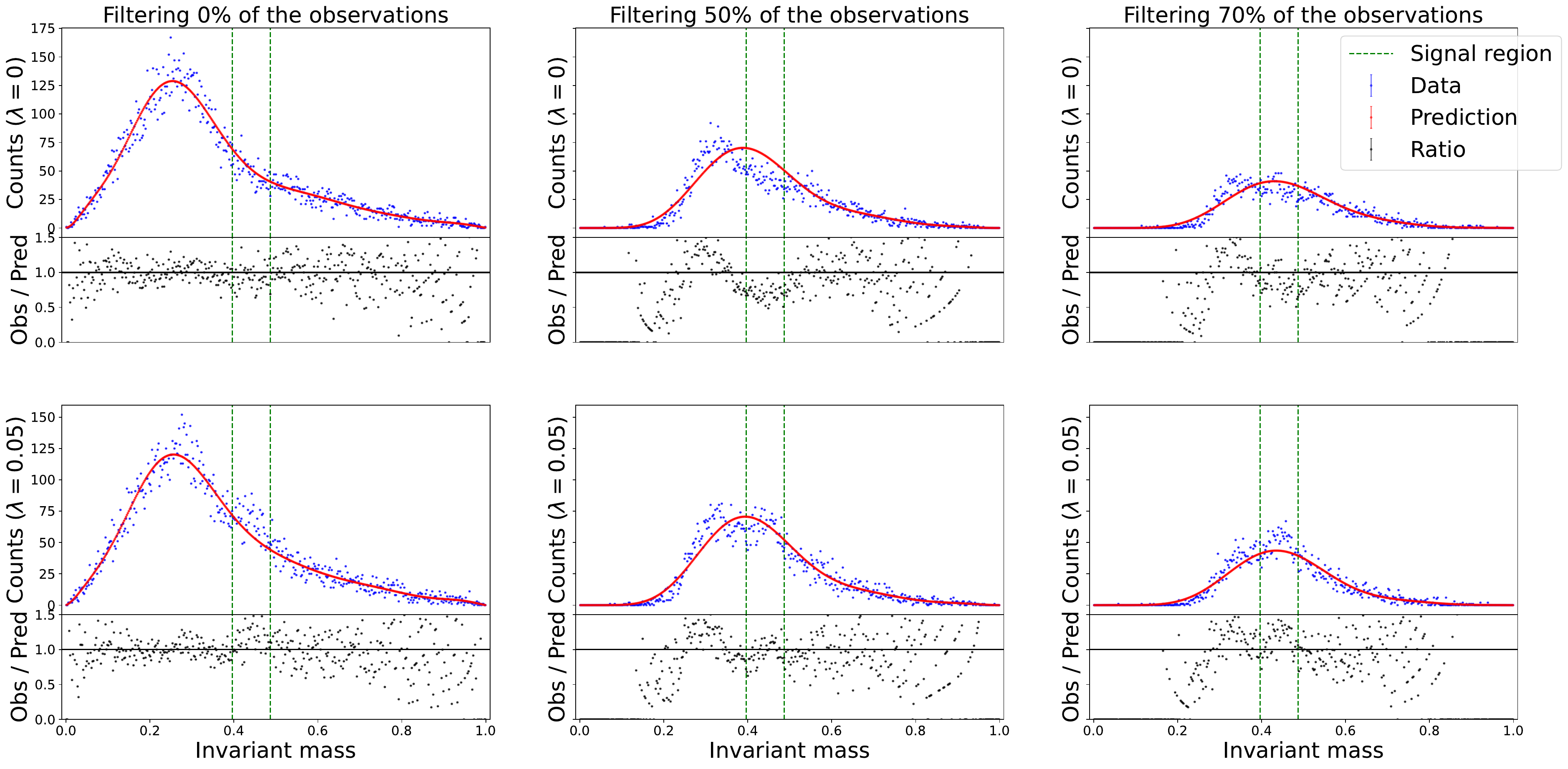}
\end{center}
\caption{Single background estimate for the 4b test dataset as the number of observations filtered by the \textbf{non-decorrelated} classifier is increased. The first row corresponds to the null hypothesis, i.e., no signal. In the second row, $5\%$ of the data comes from the signal. Note that the shape of the background distribution does not change. }
\end{figure}

\begin{figure}[H]
\begin{center}
\includegraphics[width=\textwidth]{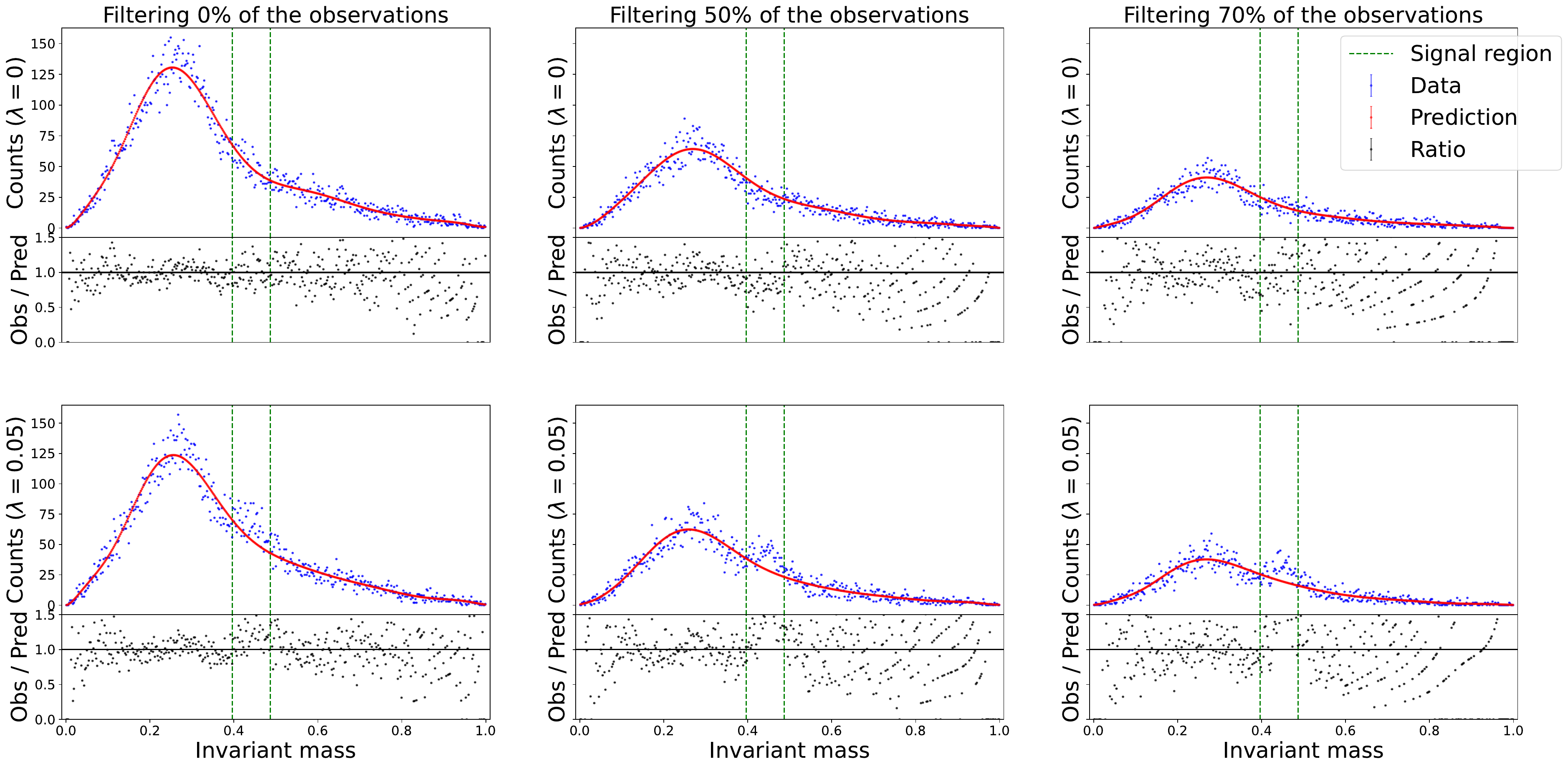}
\end{center}
\caption{Single background estimate for the 4b test dataset as the number of observations filtered by the \textbf{decorrelated} classifier is increased. The first row corresponds to the null hypothesis, i.e., no signal. In the second row, $5\%$ of the data comes from the signal. Note that the shape of the background distribution does not change. }
\end{figure}

% \putbib
% \end{bibunit}

% \pagebreak

\etocdepthtag.toc{mtreferences}
\addcontentsline{toc}{section}{References}
\bibliography{paper}

\end{document}